\pdfoutput=1 %
\PassOptionsToPackage{backref=page,bookmarksopen=true}{hyperref}

\documentclass[11pt,letterpaper,final]{article}
\usepackage[utf8]{inputenc}
\usepackage[T1]{fontenc}
\usepackage[in]{fullpage}
\usepackage{microtype}

\usepackage{caption}
\usepackage{subcaption}
\usepackage{soul}
\usepackage{amsmath}
\usepackage{amsthm}
\usepackage{amsfonts}
\usepackage{mathtools}
\usepackage{centernot}

\usepackage{algorithm}
\usepackage[noend]{algpseudocode}
\algnewcommand\LeftComment[1]{\State\(\triangleright\)~\textit{#1}}
\algrenewcommand\algorithmiccomment[1]{\hfill\(\triangleright\)~\textit{#1}}

\usepackage{enumitem}
\setlist{nolistsep}
\usepackage{thmtools}
\usepackage{thm-restate}

\usepackage{enumitem}
\usepackage{float}
\usepackage[section,above,below]{placeins} %
\usepackage{adjustbox}
\usepackage{graphicx}
\usepackage[svgnames]{xcolor}
\usepackage{tikz}
\usetikzlibrary{calc,positioning,shapes,patterns}
\usetikzlibrary{decorations.markings}
\usetikzlibrary{snakes}
\pgfdeclarelayer{background}
\pgfsetlayers{background,main}

\usepackage{titlesec}
\titlespacing{\paragraph}{%
0pt%
}{%
0.35\baselineskip plus 0.25\baselineskip minus 0.1\baselineskip%
}{%
1em%
}

\usepackage[textsize=scriptsize,textwidth=2.125cm,obeyFinal]{todonotes}
\usepackage{hyperref} %
\usepackage[notref,notcite,]{showkeys} %

\newcommand{\OO}{\mathcal{O}}

\newcommand{\idx}{\operatorname{index}}

\DeclarePairedDelimiter{\abs}{\lvert}{\rvert}

\DeclarePairedDelimiter{\set}{\lbrace}{\rbrace}
\newcommand{\suchthat}{\mathrel{}\mathclose{}\ifnum\currentgrouptype=16\middle\fi\vert\mathopen{}\mathrel{}}

\newcommand{\dual}[1]{{#1}^\ast}

\newcommand{\Emb}{\operatorname{Emb}}
\newcommand{\EmbGood}{\Emb^\star}

\newcommand{\flipdist}{\operatorname{dist}}

\newcommand{\critcost}{\operatorname{critical-cost}}
\newcommand{\solidcost}{\operatorname{solid-cost}}
\newcommand{\offstruts}{\operatorname{off-critical-struts}}
\newcommand{\critstruts}{\operatorname{critical-struts}}
\newcommand{\solidstruts}{\operatorname{solid-struts}}

\newcommand{\struts}{\operatorname{struts}}

\newcommand{\meet}{\operatorname{meet}}

\declaretheorem[style=plain,name={Theorem}]{theorem}
\declaretheorem[style=plain,name={Lemma},sibling=theorem]{lemma}
\declaretheorem[style=plain,name={Corollary},sibling=theorem]{corollary}

\declaretheorem[style=plain,name={Observation},sibling=theorem]{observation}
\declaretheorem[style=definition,name={Definition},sibling=theorem]{definition}

\newcommand{\Patrascu}{P\v{a}tra\c{s}cu}

\usepackage{authblk}

\title{Fully-dynamic Planarity Testing in Polylogarithmic Time}
\author[1]{Jacob Holm\thanks{Partially supported by the VILLUM Foundation grant 16582, Basic Algorithms Research Copenhagen (BARC).}}
\author[2]{Eva Rotenberg\thanks{Partially supported by Independent Research Fund Denmark grant ``AlgoGraph'' 2018-2021 (8021-00249B).}}
\affil[1]{University of Copenhagen \hspace{1em}{\small \href{mailto:jaho@di.ku.dk}{jaho@di.ku.dk}}}
\affil[2]{Technical University of Denmark \hspace{1em}{\small \href{mailto:eva@rotenberg.dk}{erot@dtu.dk}}}
\date{}

\begin{document}
\thispagestyle{empty}
\setcounter{page}{0}
\maketitle
\begin{abstract}
Given a dynamic graph subject to insertions and deletions of edges, a natural question is whether the graph presently admits a planar embedding. 
We give a deterministic fully-dynamic algorithm for general graphs, running in amortized $\OO(\log^3 n)$ time per edge insertion or deletion, that maintains a bit indicating whether or not the graph is presently planar. This is an exponential improvement over the previous best algorithm [Eppstein, Galil, Italiano, Spencer, 1996] which spends amortized $\OO(\sqrt{n})$ time per update.

\end{abstract}

\thispagestyle{empty}
\newpage
\setcounter{page}{1}

\section{Introduction}
A linear time algorithm for determining whether a graph is planar was found by Hopcroft and Tarjan~\cite{DBLP:journals/jacm/HopcroftT74}.
For the partially dynamic case, where %
one only allows insertion of edges, the problem was solved by La
Poutr\'{e}~\cite{DBLP:conf/stoc/Poutre94}, who improved on work by Di Battista,
Tamassia, and Westbrook~\cite{DBLP:journals/algorithmica/BattistaT96,DBLP:journals/jal/Tamassia96,DBLP:conf/icalp/Westbrook92}, to obtain a amortized running time of $\OO(\alpha(q,n))$
where $q$ is the number of operations, and where $\alpha$ is the inverse-Ackermann function. 
Galil, Italiano, and Sarnak~\cite{DBLP:journals/jacm/GalilIS99} made a data 
structure for fully dynamic planarity testing with $\OO(n^{\frac{2}{3}})$ amortized time per update, which was improved to $\OO(\sqrt{n})$ by
Eppstein et al.~\cite{Eppstein:1996}. 

While there has for a long time been no improvements upon~\cite{Eppstein:1996}, there have been different approaches in works that address the task of maintaining an embedded graph.
In \cite{DBLP:conf/esa/ItalianoPR93}, Italiano, La Poutr\'{e}, and Rauch 
present a data structure 
for maintaining a planar embedded graph while allowing insertions that do not violate the embedding, but allowing arbitrary deletions; its update time is $\OO(\log^2 n)$.  
Eppstein~\cite{DBLP:conf/soda/Eppstein03} presents a data
structure for maintaining %
a dynamic embedded graph, which handles updates in $\OO(\log
n)$ time if the embedding remains plane---this data structure maintains the genus of the embedding, but does not answer whether another embedding 
of the same graph 
with a lower genus exists.

\Patrascu{} and Demaine~\cite{DBLP:conf/stoc/PatrascuD04} give a lower bound of $\Omega(\log n)$ for fully-dynamic planarity testing. For other natural questions about fully-dynamic graphs, such as fully dynamic shortest paths, even on planar graphs, there are conditional lower bounds based on popular conjectures that indicate that subpolynomial update bounds are unlikely~\cite{AW14,AD16}. 

In this paper, we show that planarity testing does indeed admit a subpolynomial update time algorithm, thus exponentially improving the state of the art for fully-dynamic planarity testing. 
We give a deterministic fully-dynamic algorithm for general graphs, running in amortized $\OO(\log^3 n)$ time per edge insertion or deletion, that explicitly maintains a single bit indicating whether the graph is presently planar, and that given any vertex can answer whether the connected component containing that vertex is planar in worst case $\OO(\log n/\log\log n)$ time.

In fact, our algorithm not only maintains \emph{whether} the graph is
presently planar, but also implicitly maintains \emph{how} the graph
may be embedded in the plane, in the affirmative case.
Specifically, we give a deterministic algorithm that maintains a
planar embedding of a fully dynamic planar graph in amortized
$\OO(\log^3 n)$ time per edge insertion, and worst case $\OO(\log^2
n)$ time per edge deletion.  In this algorithm, attempts to insert
edges that would violate planarity are detected and rejected, but may
still change the embedding.
The algorithm for fully-dynamic general graphs then follows by a
simple reduction.

Our main result consists of two parts which may be of independent interest. Our analysis goes via a detailed understanding of \emph{flips}, i.e.\@ local changes to the embedding, to be defined in Section~\ref{sec:results}. 
Firstly, we consider any algorithm for maintaining an embedding that lazily makes no changes to the embedding upon edge deletion, and that for each (attempted) insertion greedily only does the minimal (or close to minimal) number of flips necessary to accommodate the edge. We prove that any such algorithm will do amortized $\OO(\log n)$ flips. Secondly, we show how to find such a sufficiently small set of flips in worst case $\OO(\log^2 n)$ time per flip.%

The idea of focusing on flips is not new: 
In \cite{DBLP:journals/mst/HolmR17}, we use insights from Eppstein~\cite{DBLP:conf/soda/Eppstein03} to improve upon the data structure by Italiano, La Poutr\'{e}, and Rauch~\cite{DBLP:conf/esa/ItalianoPR93}, so that it also facilitates \emph{flips}, i.e.\@ local changes to the embedding, and, so that it may handle edge-insertions that only require one such flip.
In \cite{HR20}, we analyze these \emph{flips} further and show that there exists a class of embeddings where only $\Theta(\log n)$ flips are needed to accommodate any one edge insertion that preserves planarity.

The core idea of our analysis of the lazy greedy algorithm is to
define a potential function based on how far the current embedding is
from being in this class. This is heavily inspired by the analysis of
Brodal and Fagerberg's algorithm for fully-dynamic bounded outdegree
orientation~\cite{DBLP:conf/wads/BrodalF99}.

\subsection{Maintaining an embedding if it exists} \label{sec:results}
Before stating our results in detail, we will define some crucial but natural  terminology for describing changeable embeddings of dynamic graphs. 

\paragraph{Planar graphs} are graphs that can be drawn in the plane without edge crossings. A planar graph may admit many planar embeddings, and we use the term \emph{plane graph} to denote a planar graph equipped with a given planar embedding. 
Given a plane graph, its drawing in the plane defines \emph{faces},
and the faces together with the edges form its \emph{dual
  graph}. Related, one may consider the bipartite \emph{vertex-face
  (multi-)graph} whose nodes are the vertices and faces, and which has
an edge for each time a vertex is incident to a face. Through the
paper, we will use the term \emph{corner} to denote an edge in the
vertex-face graph, reflecting that it corresponds to a corner of the
face in the planar drawing of the graph.

If a planar graph has no vertex cut-sets of size $\le 2$, its embedding is unique up to reflection. On the other hand, if a plane graph has an articulation point (cut vertex) or a separation pair ($2$-vertex cut), then it may be possible to alter the embedding by \emph{flipping}~\cite{Whitney,DBLP:journals/mst/HolmR17,HR20} the embedding in that point or pair (see figure~\ref{fig:flips}). Given two embeddings of the same graph, the flip-distance between them is the minimal number of flips necessary to get from one to the other. Intuitively, a flip can be thought of as cutting out a subgraph by cutting along a $2$-cycle or $4$-cycle in the vertex-face graph, possibly mirroring its planar embedding, and then doing the inverse operation of cutting along a $2$- 
or $4$-cycle in the vertex-face graph. The initial cutting and the final gluing involve the same vertices but not necessarily the same faces, thus, the graph but not the embedding is preserved.

\begin{figure}[htb]
	\begin{minipage}[t]{.45\textwidth}
		\centering
\begin{tikzpicture}[
    vertex/.style={
      draw,
      circle,
      fill=white,
      minimum size=1mm,
      inner sep=0pt,
    },
  ]
  \begin{scope}[shift={(2.5,0)},rotate=-36.87,scale=.45]
    \path[use as bounding box] (0,-1.5) rectangle (4,2);
    \node[vertex,label={[label distance=-2mm]105:$y$}] (x) at (0,0) {};
    \node[vertex,label={[label distance=-2mm]-15:$z$}] (y) at (5,0) {};

    \coordinate (x1) at (0,2.5);
    \coordinate (y1) at (5,2.5);
    \coordinate (x2) at (-1,3);
    \coordinate (y2) at (6,3);

    \coordinate (c1) at (2,-.35) {};
    \coordinate (x3) at (1.25,-.5) {};
    \coordinate (y3) at (3.75,-.5) {};
    \coordinate (c2) at (2,.5) {};

    \coordinate (x4) at (2,1.5);
    \coordinate (y4) at (3,1.5);
    \coordinate (x5) at (1,2);
    \coordinate (y5) at (4,2);

    \begin{pgfonlayer}{background}
      \draw[fill=gray!20] (x) .. controls (x1) and (y1) .. (y)
      .. controls (y2) and (x2) .. (x);

      \draw[fill=gray!40,pattern=dots] (x) .. controls (x4) and (y4) .. (y)
      .. controls (y5) and (x5) .. (x);

      \draw[fill=gray!70] plot[smooth] coordinates {(x) (x3) (c1) (y3) (y)}
      plot [smooth] coordinates {(y) (c2) (x)};
    \end{pgfonlayer}
  \end{scope}

  \begin{scope}[shift={(0,0)},rotate=-36.87,scale=.45]
    \path[use as bounding box] (0,-1.5) rectangle (4,2);
    \node[vertex,label={[label distance=-2mm]105:$y$}] (x) at (0,0) {};
    \node[vertex,label={[label distance=-2mm]-15:$z$}] (y) at (5,0) {};

    \coordinate (x1) at (0,2.5);
    \coordinate (y1) at (5,2.5);
    \coordinate (x2) at (-1,3);
    \coordinate (y2) at (6,3);

    \coordinate (c1) at (2,.35) {};
    \coordinate (x3) at (1.25,.5) {};
    \coordinate (y3) at (3.75,.5) {};
    \coordinate (c2) at (2,-.5) {};

    \coordinate (x4) at (2,1.5);
    \coordinate (y4) at (3,1.5);
    \coordinate (x5) at (1,2);
    \coordinate (y5) at (4,2);

    \begin{pgfonlayer}{background}
      \draw[fill=gray!20] (x) .. controls (x1) and (y1) .. (y)
      .. controls (y2) and (x2) .. (x);

      \draw[fill=gray!40,pattern=dots] (x) .. controls (x4) and (y4) .. (y)
      .. controls (y5) and (x5) .. (x);

      \draw[fill=gray!70] plot[smooth] coordinates {(x) (x3) (c1) (y3) (y)}
      plot [smooth] coordinates {(y) (c2) (x)};
    \end{pgfonlayer}

  \end{scope}

  \begin{scope}[shift={(5,0.1)},rotate=-36.87,scale=.45]
    \path[use as bounding box] (0,-1.7) rectangle (4,1.8);
    \node[vertex,label={[label distance=-2mm]105:$y$}] (x) at (0,-.2) {};
    \node[vertex,label={[label distance=-2mm]-15:$z$}] (y) at (5,-.2) {};

    \coordinate (x1) at (0,2.5);
    \coordinate (y1) at (5,2.5);
    \coordinate (x2) at (-1,3);
    \coordinate (y2) at (6,3);

    \coordinate (c1) at (2,-.35) {};
    \coordinate (x3) at (1.25,-.5) {};
    \coordinate (y3) at (3.75,-.5) {};
    \coordinate (c2) at (2,.5) {};

    \coordinate (x4) at (1,1.5);
    \coordinate (y4) at (4,1.5);
    \coordinate (x5) at (0,2);
    \coordinate (y5) at (5,2);

    \coordinate (x6) at (2,-1.5);
    \coordinate (y6) at (3,-1.5);
    \coordinate (x7) at (1,-2);
    \coordinate (y7) at (4,-2);

    \begin{pgfonlayer}{background}
      \draw[fill=gray!20] (x) .. controls (x4) and (y4) .. (y)
      .. controls (y5) and (x5) .. (x);

      \draw[fill=gray!40,pattern=dots] (x) .. controls (x6) and (y6) .. (y)
      .. controls (y7) and (x7) .. (x);

      \draw[fill=gray!70] plot[smooth] coordinates {(x) (x3) (c1) (y3) (y)}
      plot [smooth] coordinates {(y) (c2) (x)};
    \end{pgfonlayer}
  \end{scope}

\end{tikzpicture}
 		\vspace{.1em}	
		\subcaption{Separation flips: reflect and  slide.}\label{fig:separation-flip}
	\end{minipage}  
	\hfill
	\begin{minipage}[t]{.45\textwidth}
		\centering
\begin{tikzpicture}[
    vertex/.style={
      draw,
      circle,
      fill=white,
      minimum size=1mm,
      inner sep=0pt,
    },
  ]
  \begin{scope}[shift={(2.5,0)},rotate=-36.87, scale=.45]
    \path[use as bounding box] (-2,0) rectangle (2,4);
    \node[vertex,label={180:$x$}] (x) at (0,0) {};

    \coordinate (a1) at (-3,2);
    \coordinate (a2) at (0,4);
    \coordinate (a3) at (3,2);
    \coordinate (b1) at (2.5,2);
    \coordinate (b2) at (1,3);
    \coordinate (c1) at (-1,3);
    \coordinate (c2) at (-2.5,2);

    \coordinate (d1) at (-1.2,2);
    \coordinate (d2) at (-1.5,2.8);
    \coordinate (d3) at (-2.1,1.9);

    \begin{pgfonlayer}{background}
      \draw[fill=gray!20]
      plot[smooth,tension=1] coordinates {
        (x) (a1) (a2) (a3) (x)
      }
      plot[smooth,tension=1] coordinates {
        (x) (b1) (b2) (x)
      }
      plot[smooth,tension=1] coordinates {
        (x) (c1) (c2) (x)
      };

      \draw[fill=gray!70]
      plot[smooth,tension=.5] coordinates {
        (x) (d1) (d2) (d3) (x)
      };
    \end{pgfonlayer}
  \end{scope}

  \begin{scope}[shift={(0,0)},rotate=-36.87, scale=.45]
    \path[use as bounding box] (-2,0) rectangle (2,4);
    \node[vertex,label={180:$x$}] (x) at (0,0) {};

    \coordinate (a1) at (-3,2);
    \coordinate (a2) at (0,4);
    \coordinate (a3) at (3,2);
    \coordinate (b1) at (2.5,2);
    \coordinate (b2) at (1,3);
    \coordinate (c1) at (-1,3);
    \coordinate (c2) at (-2.5,2);

    \begin{scope}[rotate=72,xscale=-1]
    \coordinate (d1) at (-1.2,2);
    \coordinate (d2) at (-1.5,2.8);
    \coordinate (d3) at (-2.1,1.9);
    \end{scope}

    \begin{pgfonlayer}{background}
      \draw[fill=gray!20]
      plot[smooth,tension=1] coordinates {
        (x) (a1) (a2) (a3) (x)
      }
      plot[smooth,tension=1] coordinates {
        (x) (b1) (b2) (x)
      }
      plot[smooth,tension=1] coordinates {
        (x) (c1) (c2) (x)
      };

      \draw[fill=gray!70]
      plot[smooth,tension=.5] coordinates {
        (x) (d1) (d2) (d3) (x)
      };
    \end{pgfonlayer}
  \end{scope}

  \begin{scope}[shift={(5,0)},rotate=-36.87,scale=.45]
    \path[use as bounding box] (-2,0) rectangle (2,4);
    \node[vertex,label={180:$x$}] (x) at (0,0) {};

    \coordinate (a1) at (-3,2);
    \coordinate (a2) at (0,4);
    \coordinate (a3) at (3,2);
    \coordinate (b1) at (2.5,2);
    \coordinate (b2) at (1,3);
    \coordinate (c1) at (-1,3);
    \coordinate (c2) at (-2.5,2);

    \begin{scope}[rotate=-75]
    \coordinate (d1) at (-1.2,2);
    \coordinate (d2) at (-1.5,2.8);
    \coordinate (d3) at (-2.1,1.9);
    \end{scope}

    \begin{pgfonlayer}{background}
      \draw[fill=gray!20]
      plot[smooth,tension=1] coordinates {
        (x) (a1) (a2) (a3) (x)
      }
      plot[smooth,tension=1] coordinates {
        (x) (b1) (b2) (x)
      }
      plot[smooth,tension=1] coordinates {
        (x) (c1) (c2) (x)
      };

      \draw[fill=gray!70]
      plot[smooth,tension=1] coordinates {
        (x) (d1) (d2) (d3) (x)
      };
    \end{pgfonlayer}
  \end{scope}
\end{tikzpicture}
 		\vspace{.1em}	
		\subcaption{Articulation flips: reflect and slide.}\label{fig:articulation-flip}
	\end{minipage}
	\caption{Local changes to the embedding of a graph~\cite{HR20}.\label{fig:flips}}
\end{figure}
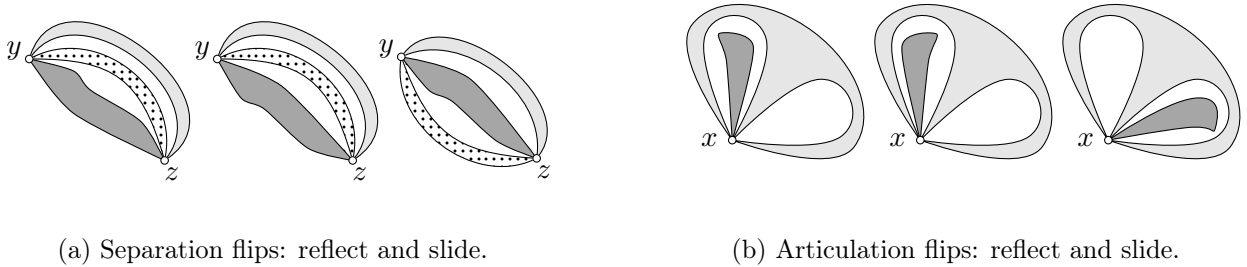

We use the following terminology: To distinguish between whether the
subgraph being flipped is connected to the rest of the graph by a
separation pair or an articulation point, we use the terms
\emph{separation flip} and \emph{articulation flip}, respectively. To
indicate whether the subgraph was mirrored, moved to a different location,
or both, we use the terms \emph{reflect}, \emph{slide}, and
\emph{reflect-and-slide}. Note that for articulation flips, only slide
and reflect-and-slide change which edges are insertable across
a face. For separation flips, note that any slide operation or
reflect-and-slide operation may be obtained by doing $3$ or $2$
reflect operations, respectively.  %

\paragraph{Results.} Let $n$ denote the number of vertices of our fully-dynamic graph. %
\begin{restatable}{theorem}{mainalggeneral}\label{thm:mainalggeneral}
  There is a data structure for fully-dynamic planarity testing that
  handles edge-insertions and edge-deletions in amortized $\OO(\log^3
  n)$ time, answers queries to planarity-compatibility of an edge in
  amortized $\OO(\log^3 n)$ time, and answers queries to whether the
  graph is presently planar in worst case $\OO(1)$ time, or to whether
  the component of a given vertex is presently planar in worst case
  $\OO(\log n / \log \log n)$ time. It maintains an implicit
  representation of an embedding that is planar on each planar 
  connected component, and may answer
  queries to the neighbors of a given existing edge in this current
  embedding, in $\OO(\log^2 n)$ time.
\end{restatable}

The result follows by applying a simple extension of the reduction by
Eppstein~et~al.~\cite[Corollary~1]{Eppstein:1996} to the following
theorem:
\begin{restatable}{theorem}{mainalgplanar}\label{thm:mainalgplanar}
  There is a data structure for maintaining a planar embedding of a
  fully-dynamic planar graph that handles edge-updates and
  planarity-compatibility queries in amortized $\OO(\log^3 n)$ time,
  edge deletions in worst-case $\OO(\log^2 n)$ time, and queries to the
  neighbors of a given existing edge in the current embedding in
  worst-case $\OO(\log^2 n)$ time.
\end{restatable}
The underlying properties in the data structure above include that the queries may change the embedding, but the deletions do not change anything aside from the mere deletion of the edge itself.

To arrive at these theorems, we prove some technical lemmas which may be of independent interest. A weak form of these that is easy to state is the following:
\begin{restatable}{lemma}{lemlazy}\label{lem:lazy}
  Any algorithm for maintaining a fully dynamic planar embedding that
  for each attempted edge-insertion greedily does the minimal number
  of flips, and that for each edge deletion lazily does nothing, will
  do amortized $\OO(\log n)$ flips per insertion when starting with an
  empty graph (or amortized over $\Omega(n/\log n)$ operations).
\end{restatable}

\subsection{Article outline}
In Section~\ref{sec:prelim}, we introduce some of the concepts and data structures that we use to prove our result. In Section~\ref{sec:numflips} we give the proof of Lemma~\ref{lem:lazy} conditioned on insights and details deferred to Section~\ref{sec:define-costs}. In Section~\ref{sec:findflips}, we prove Theorem~\ref{thm:mainalgplanar} by giving an algorithm (given an edge to insert and an embedded graph) for greedily finding flips that bring us closer to an embedding that is compatible with the edge we are trying to insert. In Section~\ref{sec:reduction}, we show how the reduction that extends this result to general graphs, proving Theorem~\ref{thm:mainalggeneral}.

\section{Preliminaries}\label{sec:prelim}

Since each $3$-connected component of a planar graph has a unique embedding up to reflection, the structure of $3$-connected components play an important role in the analysis of embeddings. Namely, the two-vertex cuts (also known as separation pairs) point to places where there is a choice in how to embed the graph. Similarly, any cutvertex (also known as articulation point) points to a freedom in the choice of embedding. 
In the following, we will define the BC tree and the SPQR tree which are structures that reflect the $2$-connected components of a connected graph and the $3$-connected components of a $2$-connected graph, respectively. Then, related to the understanding of a combinatorial embedding, we will define \emph{flips} which are local changes to the embedding, and the notion of \emph{flip distance} between embedded graphs.%

\paragraph{BC trees,} as described in~\cite[p. 36]{Harary69}, reflect the $2$-connected components and their relations. 
\begin{definition}
	Let $x$ be a vertex in a connected loopless multigraph
	$G$. Then $x$ is an \emph{articulation point} if $G-\set{x}$ is not
	connected.
\end{definition}
\begin{definition}
	A (strict) BC tree for a connected loopless multigraph
        $G=(V,E)$ with at least $1$ edge is a tree with nodes labelled
        $B$ and $C$, where each node $v$ has an associated
        \emph{skeleton graph} $\Gamma(v)$ with the following
        properties:
	\begin{itemize}
		\item For every node $v$ in the BC tree, $V(\Gamma(v))\subseteq V$.
		\item For every edge $e\in E$ there is a unique node $v=b(e)$ in the
		BC tree such that $e\in E(\Gamma(v))$.
		\item For every edge $(u,v)$ in the BC tree, $V(\Gamma(u))\cap
		V(\Gamma(v))\neq\emptyset$ and either $u$ or $v$ is a $C$ node.
		\item If $v$ is a $B$ node, $\Gamma(v)$ is either a single edge or a
		biconnected graph.
		\item If $v$ is a $C$ node, $\Gamma(v)$ consists of a single vertex,
		which is an articulation point in $G$.
		\item No two $B$ nodes are neighbors.
		\item No two $C$ nodes are neighbors.
	\end{itemize}
\end{definition}
The BC tree for a connected graph is unique. The (skeleton graphs
associated with) the $B$ nodes are sometimes referred to as $G$'s
biconnected components.
In this paper, we use the term \emph{relaxed} BC tree as defined in \cite{HR20} to denote a tree that satisfies all but the last condition.  Unlike the strict BC tree, the relaxed BC tree is not unique.

\paragraph{SPQR trees} Reflecting the structure of the $3$-connected components, we rely on the SPQR tree.

\begin{definition}[Hopcroft and Tarjan~{\cite[p. 6]{DBLP:journals/siamcomp/HopcroftT73}}]\label{def:separationpair}
	Let $\set{a,b}$ be a pair of vertices in a biconnected multigraph $G$. Suppose the edges of $G$ are divided into equivalence classes $E_1,E_2,\ldots,E_k$, such that two edges which lie on a common path not containing any vertex of $\set{a,b}$ except as an end-point are in the same class. The classes $E_i$ are called the \emph{separation classes} of $G$ with respect to $\set{a,b}$. If there are at least two separation classes, then $\set{a,b}$ is a \emph{separation pair} of $G$ unless (i) there are exactly two separation classes, and one class consists of a single edge\footnote{So in a triconnected graph, the endpoints of an edge do not constitute a separation pair.}, or (ii) there are exactly three classes, each consisting of a single edge\footnote{So the graph consisting of two vertices connected by $3$ parallel edges is triconnected.}
\end{definition}

\begin{definition}[\cite{DBLP:conf/esa/HolmIKLR18}]\label{def:SPQR}
	The \emph{(strict) SPQR tree} for a biconnected multigraph $G=(V,E)$ with at least $3$ edges is a tree with nodes labeled S, P, or R, where each node $x$ has an associated \emph{skeleton graph} $\Gamma(x)$ with the following properties:
	\begin{itemize}
		\item For every node $x$ in the SPQR tree, $V(\Gamma(x))\subseteq V$.
		\item For every edge $e\in E$ there is a unique node $x=b(e)$ in the SPQR tree such that $e\in E(\Gamma(x))$.
		\item For every edge $(x,y)$ in the SPQR tree, $V(\Gamma(x))\cap V(\Gamma(y))$ is a separation pair $\set{a,b}$ in $G$, and~there is a \emph{virtual edge} $ab$ in each of $\Gamma(x)$ and $\Gamma(y)$ that corresponds to $(x,y)$.
		\item For every node $x$ in the SPQR tree, every edge in $\Gamma(x)$ is either in $E$ or a virtual edge.
		\item If $x$ is an S node, $\Gamma(x)$ is a simple cycle with at least
		$3$ edges.
		\item If $x$ is a P node, $\Gamma(x)$ consists of a pair of vertices with at least $3$ parallel edges.
		\item If $x$ is an R node, $\Gamma(x)$ is a simple %
		triconnected graph.
		\item No two S nodes are neighbors, and no two P nodes are neighbors.
	\end{itemize}
\end{definition}
The SPQR tree for a biconnected graph is unique (see e.g.~\cite{DBLP:journals/algorithmica/BattistaT96}). 
In this paper, we use the term \emph{relaxed SPQR tree} as defined in \cite{HR20} to denote a tree that satisfies all but the last condition. Unlike the strict SPQR tree, the relaxed SPQR tree is not unique. 

\paragraph{Pre-split BC trees and SPQR trees.} \label{par:presplit}
Once the BC tree or the SPQR is rooted, one may form a path decomposition \cite{DBLP:journals/jcss/SleatorT83} over them. 
Given a connected component, one may form its BC tree and an SPQR tree for each block.
In \cite{HR20}, we show how to obtain a balanced combined tree for each component, inspired by \cite{DBLP:journals/siamcomp/BentST85}, where the heavy paths reflect not only the local SPQR tree of a block, but also the weight of the many other blocks.

Given a heavy path decomposition, \cite{HR20} introduced \emph{presplit} versions of BC trees and SPQR trees, in which cutvertices, P nodes and S nodes that lie internal on heavy paths have been split in two, thus transforming the strict BC tree or SPQR tree into a relaxed BC or SPQR tree.

\paragraph{Flip-finding.} In \cite{DBLP:journals/mst/HolmR17} we give a structure for maintaining a planar embedded graph subject to edge deletions, insertions across a face, and flips changing the embedding. It operates using the tree-cotree decomposition of a connected plane graph; for any spanning tree, the non-tree edges form a spanning tree of the dual graph where faces and vertices swap roles. The data structure allows the following interesting operation: Mark a constant number of faces, and search for vertices along a path in the spanning tree that are incident to all marked faces. Or, dually, mark a number of vertices and search for faces along a path in the cotree. This mark-and-search operation is supported in $\OO(\log^2 n)$ time. Originally, in \cite{DBLP:journals/mst/HolmR17}, we use the mark-and-search operation to detect single flips necessary to bring a pair of vertices to the same face. It turns out this operation is more powerful than previously assumed, and we will use it as part of the machinery that finds all the possibly many flips necessary to bring a pair of vertices to the same face.

\paragraph{Good embeddings}
In \cite{HR20}, a class of good embeddings are provided that share the property that any edge that can be added without violating planarity only requires $\OO(\log n)$ flips to the current embedding. The good embeddings relate to the dynamic balanced heavy path decompositions of BC trees and SPQR trees in the following way: For all the articulation points and separation pairs that lie internal on a heavy path, the embedding of the graph should be favorable to the possibility of an incoming edge connecting the endpoints of the heavy path, if possible. No other choices to the embedding matter; the properties of the heavy path decomposition ensure that only $\OO(\log n)$ such choices may be unfavorable to accommodating any planarity-preserving edge insertion.

\paragraph{Projections and meets}
For vertices $x,y$, and $z$ on a tree, we use $\meet(x,y,z)$ to denote
the unique common vertex on all $3$ tree paths between $x$, $y$, and
$z$. Alternatively $\meet(x,y,z)$ can be seen as the projection of $x$
on the tree path from $y$ to $z$.  For a vertex $x$ and a fundamental
cycle $C$ in a graph with some implied spannning tree, let $\pi_C(x)$
denote the projection of $x$ on $C$. Note that if $(y,z)$ is the non-tree
edge closing the cycle $C$ then $\pi_C(x)=\meet(x,y,z)$.

\section{Analyzing the number of flips in the lazy greedy algorithm}\label{sec:numflips}

This section is dedicated to the proof of Lemma~\ref{lem:lazy}. We define some distance measures and a concept of good embeddings, and single out exactly the properties of these that would be sufficient for the lazy analysis to go through. All proofs that these sufficient conditions indeed are met are deferred to Section \ref{sec:define-costs}, thus enabling us to give an overview within the first limited number of pages.

\begin{definition}
  We will consider only two kinds of flips:
  \begin{itemize}
  \item An articulation flip at $a$ takes a single contiguous
    subsequence of the edges incident to $a$ out, possibly reverses
    their order, and inserts them again, possibly in a different
    position.%
  \item A separation flip at $s,t$ reverses a contiguous subsequence
    of the edges incident to each of $s,t$.%
  \end{itemize}
\end{definition}

\begin{definition}
  We will distinguish between two types of separation flips at $s,t$:
  \begin{itemize}
  \item In a P flip, each of the two subgraphs consist of at least $2$
    $\set{s,t}$-separation classes.
  \item In an SR flip, at least one of the two subgraphs consist of a
    single $\set{s,t}$-separation class.
  \end{itemize}
\end{definition}

\begin{definition}
  A separation flip at $s,t$ is \emph{clean} if the first and last
  edges in each of the subsequences being reversed are biconnected,
  and \emph{dirty} otherwise. A clean separation flip preserves the
  sets of edges that could participate in an articulation flip.  All
  articulation flips are considered clean.
\end{definition}

\begin{definition}
  Given vertices $u,v$, a flip is called \emph{critical} (for $u,v$)
  if exactly one of $u$ and $v$ is in the subgraph being flipped.
\end{definition}

\begin{definition}
  For a planar graph $G$, let $\Emb(G)$ denote the graph whose nodes
  are planar embeddings of $G$, and $(H,H')$ is an edge if $H'$ is
  obtained from $H$ by applying a single flip.  As a slight abuse of
  notation we will also use $\Emb(G)$ to denote the set of all planar
  embeddings of $G$.  Furthermore, for vertices $u,v$ in $G$ let
  $\Emb(G;u,v)$ denote the (possibly empty) set of embeddings of $G$
  that admit insertion of $(u,v)$.
\end{definition}
\noindent{}We will often need to discuss distances in some
(pseudo)metric between a particular embedding $H\in\Emb(G)$ and some
particular subset of embeddings $S\subseteq\Emb(G)$. For this, we
define (for any metric or pseudometric $\flipdist$)
\begin{align*}
  \flipdist(H, S) = \flipdist(S, H) &:= \min_{H'\in S}\flipdist(H,H')
\end{align*}

\begin{definition}
  For any two embeddings $H,H'$ of the planar graph $G$, that is,  
  $H,H'\in\Emb(G)$, we say that a
  path from $H$ to $H'$ in $\Emb(G)$ is clean if every flip on the
  path is clean.
  \begin{itemize}
  \item Let $\flipdist_{\text{clean}}(H,H')$ be the length of a
    shortest clean path from $H$ to $H'$.
  \item Let $\flipdist_{\text{sep}}(H,H')$ be the minimum number of
    separation flips on a clean path from $H$ to $H'$.
  \item Let $\flipdist_P(H,H')$ be the minimum number of P flips on a
    clean path from $H$ to $H'$.
  \end{itemize}
\end{definition}
\begin{observation}
  $\flipdist_{\text{clean}}$ is a metric, and $\flipdist_{\text{sep}}$
  and $\flipdist_{\text{P}}$ are pseudometrics, on $\Emb(G)$.
\end{observation}

In Section~\ref{sec:define-costs}, we define two families of functions
related to these particular (pseudo) metrics. Intuitively, for each
$\tau\in\set{\text{clean},\text{sep},\text{P}}$, any planar graph $G$ containing vertices $u,v$, and any embedding $H\in\Emb(G)$:

\begin{itemize}
\item $\critcost_\tau(H; u,v)$ is the number of flips of type
  $\tau$ needed to accommodate $(u,v)$ (if possible).
\item $\solidcost_\tau(H; u,v)$ is the number of flips of type $\tau$
  needed to reach a ``good'' embedding that accommodates $(u,v)$ (if
  possible).
\end{itemize}
We will state the properties we need for these functions here, in the
form of Lemmas (to be proven once the actual definition has been given).
\begin{lemma}\label{lem:costs-nonneg}
  For any planar graph $G$ with vertices $u,v$, and any embedding
  $H\in\Emb(G)$,
  \begin{align*}
    \solidcost_\tau(H; u,v) &\geq \critcost_\tau(H; u,v)\geq 0
    \\
    \solidcost_{\text{clean}}(H; u,v) &\geq
    \solidcost_{\text{sep}}(H; u,v) \geq
    \solidcost_{\text{P}}(H; u,v)
    \\
    \critcost_{\text{clean}}(H; u,v) &\geq
    \critcost_{\text{sep}}(H; u,v) \geq
    \critcost_{\text{P}}(H; u,v)
  \end{align*}
  And if $G\cup(u,v)$ is planar,
  \begin{align*}
    \critcost_{\text{clean}}(H; u,v) = 0 \iff H\in\Emb(G; x,y)
  \end{align*}
\end{lemma}
\begin{lemma}\label{lem:costs-delta}
  Let $u,v$ be vertices in a planar graph $G$, let $H\in\Emb(G)$ and
  let $H'\in\Emb(G)$ be the result of a single flip $\sigma$ in $H$. Define
  \begin{align*}
    \Delta\critcost_\tau &:= \critcost_\tau(H'; u,v)-\critcost_\tau(H; u,v)
    \\
    \Delta\solidcost_\tau &:= \solidcost_\tau(H'; u,v)-\solidcost_\tau(H; u,v)
  \end{align*}
  then $\Delta\critcost_\tau \in \set{-1,0,1}$,
  $\Delta\solidcost_\tau \in \set{-1,0,1}$, and
  \begin{align*}
    \Delta\critcost_\tau&\neq 0 &&\implies&  \text{$\sigma$ is a critical flip}& &&\iff& \Delta\solidcost_\tau &= \Delta\critcost_\tau
  \end{align*}
\end{lemma}
\begin{lemma}\label{lem:costs-exist-decreasing-flip}
  Let $u,v$ be vertices in a planar graph $G$, and let $H\in\Emb(G)$.
  \begin{itemize}
  \item If $\critcost_\tau(H; u,v)>0$ then there exists a clean flip
    in $H$ such that the resulting $H'\in\Emb(G)$ has
    $\critcost_\tau(H'; u,v)<\critcost_\tau(H; u,v)$.
  \item If $\solidcost_\tau(H; u,v)>0$ then there exists a clean flip
    in $H$ such that the resulting $H'\in\Emb(G)$ has
    $\solidcost_\tau(H'; u,v)<\solidcost_\tau(H; u,v)$.
  \end{itemize}
\end{lemma}

\noindent{}The main motivation for defining $\critcost$ comes from the following
\begin{corollary}\label{cor:cost-is-distance}
  Let $G$ be a planar graph, let $u,v$ be vertices in $G$ such that
  $G\cup(u,v)$ is planar, and let $H\in\Emb(G)$. Then
  \begin{align*}
    \flipdist_\tau(H, \Emb(G; u,v)) &= \critcost_\tau(H; u,v).
  \end{align*}
\end{corollary}
\begin{proof}
  ``$\geq$'' follows from Lemmas~\ref{lem:costs-nonneg}
  and~\ref{lem:costs-delta}, and ``$\leq$'' follows from
  Lemma~\ref{lem:costs-exist-decreasing-flip}.
\end{proof}

With these properties in hand, we can now redefine what we mean by a \emph{good} embedding in a quantifiable way:
\begin{definition}
  Given a planar graph $G$ with vertices $u,v$, the \emph{good}
  embeddings of $G$ with respect to $u,v$ is the set
  \begin{align*}
    \EmbGood(G; u,v) &:= \set{
      H\in\Emb(G)\suchthat\solidcost_{\text{clean}}(H; u,v)=0
    }
    \intertext{And the set of all good embeddings of $G$ is}
    \EmbGood(G) &:= \bigcup_{u,v}\EmbGood(G; u,v) = \set{
      H\in\Emb(G)\suchthat\min_{u,v}\solidcost_{\text{clean}}(H; u,v)=0
    }
  \end{align*}
\end{definition}
Note that this definition is not the same as in~\cite{HR20}, although the underlying ideas are the same.

With this definition, we get
\begin{corollary}\label{cor:cost-is-distance2}
  Let $G$ be a planar graph, let $u,v$ be vertices in $G$, and let
  $H\in\Emb(G)$. Then
  \begin{align*}
    \flipdist_\tau(H, \EmbGood(G; u,v)) &= \solidcost_\tau(H; u,v)
    \\
    \flipdist_\tau(H, \EmbGood(G)) &= \min_{u,v}\solidcost_\tau(H; u,v)
  \end{align*}
\end{corollary}
\begin{proof}
  ``$\geq$'' follows from Lemmas~\ref{lem:costs-nonneg}
  and~\ref{lem:costs-delta}, and ``$\leq$'' follows from
  Lemma~\ref{lem:costs-exist-decreasing-flip}.
\end{proof}

The reason we call these embeddings good is the following property
\begin{lemma}\label{lem:embgood-logn}
  Given a planar graph $G$ with vertices $u,v$ and any
  $H\in\EmbGood(G)$, then
  \begin{align*}
    \flipdist_{\text{clean}}(H, \EmbGood(G; u,v)) &\in \OO(\log n)
  \end{align*}
\end{lemma}

\begin{corollary}\label{cor:embgood-sandwich}
  Given a planar graph $G$ with vertices $u,v$ and any
  $H\in\Emb(G)$, then
  \begin{align*}
    \flipdist_\tau(H, \EmbGood(G))
    \leq 
    \flipdist_\tau(H, \EmbGood(G; u,v))
    \leq
    \flipdist_\tau(H, \EmbGood(G)) + \OO(\log n)
  \end{align*}
\end{corollary}
\begin{proof}
  The first inequality follows from $\EmbGood(G)\supseteq\EmbGood(G;
  u,v)$. For the second, let $H'\in\EmbGood(G)$ minimize
  $\flipdist_\tau(H,H')$. By Corollary~\ref{cor:cost-is-distance2} and
  Lemmas~\ref{lem:costs-nonneg} and~\ref{lem:embgood-logn},
  \begin{align*}
    \flipdist_\tau(H',\EmbGood(G;u,v)) \leq
    \flipdist_{\text{clean}}(H',\EmbGood(G;u,v)) \in \OO(\log n)
  \end{align*}
  and the result follows by the triangle inequality.
\end{proof}
\begin{lemma}\label{lem:emb-insert}
  Given a planar graph $G$ with vertices $u,v$, such that $G\cup(u,v)$
  is planar, and any $H\in\Emb(G; u,v)$. Then
  \begin{align*}
    \solidcost_\tau(H; u,v) = \solidcost_\tau(H\cup(u,v); u,v)
  \end{align*}
\end{lemma}

Everything so far has been stated in terms of clean flips. Most of our
results do not depend on this, due to the following lemma (Proved in
Section~\ref{sec:struts-properties}).
\begin{lemma}\label{lem:dirty-ok}
  A dirty separation flip corresponds to a clean separation flip and
  at most $4$ articulation flips. At most one of these $5$ flips
  change $\solidcost_\tau$ or $\critcost_\tau$.
\end{lemma}

\begin{theorem}\label{thm:lazygreedy}
  Let $p,q,r\in\mathbb{N}_0$ be nonnegative integer constants.  Let
  $\mathcal{A}$ be a lazy greedy algorithm for maintaining a planar
  embedding with the following behavior:
  \begin{itemize}
  \item $\mathcal{A}$ does no flips during edge deletion; and
  \item during the attempted insertion of the edge $(u,v)$ into an
    embedded graph $H$, $\mathcal{A}$ only uses critical flips; and
  \item this sequence of flips can be divided into \emph{steps} of at
    most $r$ flips, such that each step (except possibly the last)
    decreases the following potential by at least $1$:
    \begin{align*}
      \critcost_{\text{clean}}(H;u,v) + p\cdot\critcost_{\text{sep}}(H;u,v) +
      q\cdot\critcost_{\text{P}}(H;u,v)
    \end{align*}
  \end{itemize}
  Then, $\mathcal{A}$ uses amortized $\OO(\log n)$ steps per
  attempted edge insertion.
\end{theorem}
\begin{proof}
  Let $G$ be a planar graph with vertices $u,v$, such that
  $G\cup(u,v)$ is planar, let $H_0\in\Emb(G)$ be the embedding before
  inserting $(u,v)$ and let $H_1,\ldots,H_k\in\Emb(G)$ be the embedded
  graphs after each ``step'' until finally $H_k\in\Emb(G; u,v)$.

  Define $\flipdist_{\text{alg}} = \flipdist_{\text{clean}} +
  p\cdot\flipdist_{\text{sep}} + q\cdot\flipdist_{\text{P}}$. Then
  $\flipdist_{\text{alg}}$ is a metric on $\Emb(G)$, and
  for each $H\in\Emb(G)$ we can define
  \begin{align*}
    \Phi(H)
    &= \flipdist_{\text{alg}}(H, \EmbGood(G))
    \\
    \Phi(H; u,v) &= \flipdist_{\text{alg}}(H, \EmbGood(G; u,v))
  \end{align*}
  By Corollary~\ref{cor:embgood-sandwich} $\Phi(H)\leq\Phi(H;
  u,v)\leq\Phi(H)+\OO(\log n)$.

  By assumption, $\flipdist_{\text{alg}}(H, \Emb(G; u,v))$ is strictly
  decreasing in each step, and by Lemma~\ref{lem:costs-delta},
  $\Phi(H; u,v)$ decreases by exactly the same amount. In particular,
  after $k$ steps it has decreased by at least $k$, so $\Phi(H_k;
  u,v)\leq \Phi(H_0; u,v) - k \leq \Phi(H_0) + \OO(\log n) - k$.
  By Lemma~\ref{lem:emb-insert}, $\Phi(H_k\cup(u,v);
  u,v)=\Phi(H_k; u,v)$, so
  \begin{align*}
    \Phi(H_k\cup(u,v))
    \leq \Phi(H_k\cup(u,v); u,v)
    = \Phi(H_k; u,v)
    \leq \Phi(H_0; u,v) - k
    \leq \Phi(H_0) + \OO(\log n) - k
  \end{align*}
  The same argument holds when an attempted insert stops after $k$
  steps because $G\cup(u,v)$ is not planar. 
  Since the potential $\Phi(H)$ increases by at most $\OO(\log n)$ and drops by
  at least the number of steps used, the amortized number of steps for
  each attempted insert is $\OO(\log n)$.

  For deletion it is even simpler, as
  \begin{align*}
    \Phi(H-(u,v))
    \leq \Phi(H-(u,v); u,v)
    = \Phi(H; u,v)
    \leq \Phi(H) + \OO(\log n)
  \end{align*}
  Thus, each deletion increases the potential by $\OO(\log
  n)$. However, as we start with an empty edge set, the number of
  deletions is upper bounded by the number of insertions, and so each
  edge can instead pay a cost of $\OO(\log n)$ steps when inserted to
  cover its own future deletion. In other words, deletions are
  essentially free.
\end{proof}

As a direct consequence, our main lemma holds.
\lemlazy*
\begin{proof}
Set $p=q=0$ and $r=1$ in the lemma above, and the result follows. 
\end{proof}

\begin{figure}
  \centering
\begin{tikzpicture}[
  embedding/.style = {
    circle,
    draw,
    fill,
    inner sep = 2pt,
  },
  plotlabel/.style={
    postaction={
      decorate,
      transform shape,
      decoration={pre length=1pt,post length=1pt,markings,mark=at position 0.5 with \node #1;},
    },
  },
  ]

  \node[embedding,label={180:$X_0$}] (X0) at (-1,5) {};
  \node[embedding,label={180:$X_k$}] (Xk) at (-1,3) {};
  \node[embedding,label={180:$X_k\cup(u,v)$}] (Gp) at (-1,2) {};

  \node[embedding,fill=white,label={0:$H\in\EmbGood(G)$}] (H0) at (8,5) {};
  \node[embedding,fill=white] (Y0) at (7.875,3.25) {};
  \node[embedding,label={0:$Y\in\EmbGood(G;u,v)$}] (Y) at (8,3) {};
  \node[embedding,label={0:$Y'\in\EmbGood(G\cup(u,v);u,v)$}] (Yp) at (8,2) {};
  \node[embedding,label={0:$H'\in\EmbGood(G\cup(u,v))$}] (Hp) at (8,0) {};

  \draw (X0) -- (H0) node[midway,above] {$\Phi(X_0)$};
  \draw plot[smooth] coordinates {(X0.south) (0,4) (7,4) (Y0.north)} [
    plotlabel={[above]{$\Phi(X_0; u,v)\leq\Phi(X_0)+\OO(\log n)$}}
  ];

  \draw (X0) -- (Xk) node[midway,left] {$k$};
  \draw (Xk) -- (Y) node[midway,above] {$\Phi(X_k; u,v)\leq\Phi(X_0; u,v)-k$};
  \draw (H0) -- (Y0) node[midway,right] {$\OO(\log n)$};
  \draw (Xk) -- (Gp) node[midway,left] {$0$};
  \node at (4,2.5) {$\Phi(X_k\cup(u,v); u,v)=\Phi(X_k; u,v)$};
  \draw (Y) -- (Yp) node[midway,right] {$0$};
  \draw (Gp) -- (Yp);
  \draw (Yp) -- (Hp) node[midway,right] {$\OO(\log n)$};
  \draw plot[smooth] coordinates {(Gp.south) (0,1) (7,1) (Hp.north)} [
    plotlabel={[above]{$\Phi(X_k\cup(u,v))\leq\Phi(X_k\cup(u,v); u,v)$}}
  ];

  \draw[dotted] (-4,3.5) -- (10,3.5);
  \node at (-4,5) {$\Emb(G)$};
  \node at (-4,3) {$\Emb(G; u,v)$};
  \node at (-4,0) {$\Emb(G\cup(u,v))$};

  \draw[dotted] (-5,2.5) -- (11,2.5);

  \draw[dotted] (9,4) ellipse (2.75cm and 1.5cm);
  \draw[dotted] (11,3) arc (0:180:2cm and .5cm);
  \draw[dotted] (7,2) arc (180:360:2cm and .5cm);
  \draw[dotted] (9,1) ellipse (2.75cm and 1.5cm);

\end{tikzpicture}
   \caption{\label{fig:lazygreedy-proof}Illustration of the proof of Theorem~\ref{thm:lazygreedy}.}
\end{figure}
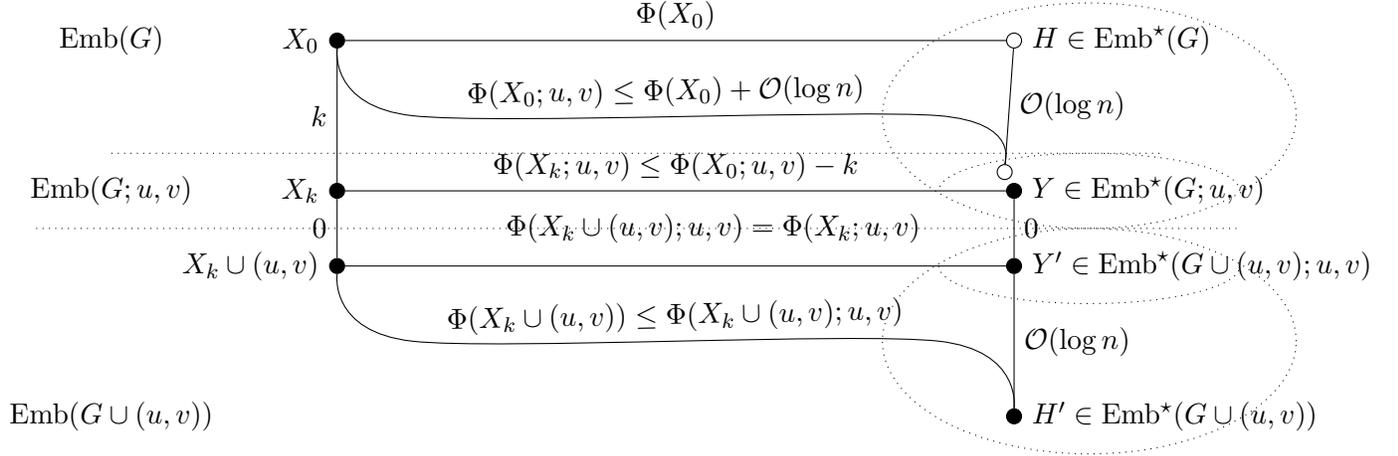

Our algorithm may not decrease $\critcost$ in every step. To simplify
the description of how and when it changes, we associate each type of
flip with one of the $\critcost_\tau$ as in Corollary~\ref{cor:lazygreedy-simple}:
\begin{itemize}
\item A P flip is associated with $\critcost_{\text{P}}$
\item An SR flip is associated with $\critcost_{\text{sep}}$
\item An articulation flip is associated with $\critcost_{\text{clean}}$
\end{itemize}
A flip is potential-decreasing/potential-neutral/potential-increasing
if it changes its associated cost by $-1/0/1$ respectively.

Using the following lemma (proved in Section~\ref{sec:define-costs})
\begin{lemma}\label{lem:fliptypes-unchanged}
  Any SR flip or articulation flip leaves $\critcost_{\text{P}}$ and
  $\solidcost_{\text{P}}$ unchanged, and any articulation flip leaves
  $\critcost_{\text{sep}}$ and $\solidcost_{\text{sep}}$ unchanged.
\end{lemma}
we can simplify Theorem~\ref{thm:lazygreedy} to.

\begin{corollary}\label{cor:lazygreedy-simple}
  Let $\mathcal{A}$ be a lazy greedy algorithm for planar embeddings
  that does no flips during edge deletion, and that during attempted
  edge insertion only uses critical flips such that
  \begin{itemize}
  \item Each flip (except possibly the last) is potential-decreasing
    or potential-neutral.
  \item For some constant $r$, no sequence of $r$ consecutive flips
    are potential-neutral.
  \end{itemize}
  Then $\mathcal{A}$ uses amortized $\OO(\log n)$ steps per attempted
  edge insertion.
\end{corollary}
\begin{proof}
  Simply use $p=r+1,q=r^2+2r+1,r=r$ in Theorem~\ref{thm:lazygreedy},
  and note that Lemma~\ref{lem:fliptypes-unchanged} guarantees the
  resulting potential is strictly decreasing in each round of at most
  $r$ flips.
\end{proof}

\section{A Greedy Flip-Finding Algorithm}\label{sec:findflips}

We use the data structure from~\cite{DBLP:journals/mst/HolmR17} to
represent the current embedding. This structure maintains
\emph{interdigitating spanning trees} (also known as the tree co-tree
decomposition) for the primal and dual graphs under flips, admissible
edge insertions, and edge deletions in worst case $\OO(\log^2 n)$ time
per operation.  In particular it supports the $\Call{linkable}{u,v}$
operation, which in worst case $\OO(\log^2 n)$ time either determines
that $u$ and $v$ has no face in common and returns ``no'', or returns
some pair of corners $((u,f),(v,f))$ where $f$ is a common face.

Furthermore, the structure allows for a mark-and-search operation, in which a constant number of faces may be ``marked'', and vertices along a path on the spanning tree that are incident to all marked faces may be sought after in $\OO(\log ^2 n)$ time. (Dually, one may mark vertices and search for faces in the same time, ie. $\OO(\log^2 n)$.)

\subsection{Algorithm overview}

We want to use this structure to search for the flips needed to insert
a new edge $(u,v)$ that is not admissible in the current
embedding. For simplicity, we will present an algorithm that fits the
framework in Corollary~\ref{cor:lazygreedy-simple}, rather than
insisting on finding a shortest sequence of clean flips. This is
sufficient to get amortized $\OO(\log n)$ flips, and with $\OO(\log^2
n)$ overhead per flip, amortized $\OO(\log^3 n)$ time for (attempted)
edge insertion. If desired, the algorithm can be made to detect if it
has made non-optimal flips and backtrack to use an optimal sequence of
flips without affecting the asymptotic amortized running time.

At the highest possible level of abstraction, our algorithm is just
the \textproc{multi-flip-linkable} routine from
Algorithm~\ref{alg:multi-flip-linkable}. In the following we will go
into more detail and provide detailed proofs. We will assume full
knowledge of how to use the mark-and-search features
from~\cite{DBLP:journals/mst/HolmR17} to e.g.\@ search a path in the
dual tree for the first face containing a given pair of vertices.

\begin{algorithm}[htb!]
  \caption{}
  \label{alg:multi-flip-linkable}
  \begin{algorithmic}[1]

    \Function{multi-flip-linkable}{$u,v$}

      \State $u'\gets u$

      \While{$u'\neq v$}

        \If{$u',v$ biconnected}

          \State $v'\gets v$

        \Else{}

          \State $v'\gets$ first articulation point on $u'\cdots v$.

        \EndIf

        \If{not \Call{do-separation-flips}{$u',v'$}}\label{line:call-do-sep}
           \Comment{Do all separation flips needed in $B_{i+1}$.}

           \State \Return ``no''

        \EndIf

        \State\label{line:call-do-art}
        $\Call{do-articulation-flips}{u,u',v',v}$
        \Comment{Do articulation flips required by $B_{i+1}$.}

        \LeftComment{Now $u$ shares a face with $v'$, and (if $v'\neq
        v$) with at least one edge in $B_{i+2}$.}

        \State\label{line:next-block} $u'\gets
        \Call{find-next-flip-block}{u,u',v',v}$
        \Comment{Skip to next relevant $a_i$}

      \EndWhile

      \State \Return ``yes''
      \Comment{$u,v$ are now in same face}

    \EndFunction

  \end{algorithmic}
\end{algorithm}

Let $a_1,\ldots,a_{k-1}$ be the articulation points on $u\cdots v$, and
let $a_0=u$ and $a_k=v$. For $1\leq i\leq k$ let $B_i$ be the
biconnected component (or bridge) containing $a_{i-1}\cdots a_i$.
Our algorithm ``cleans up'' this path by sweeping from $a_0=u$ to
$a_k=v$. At all times the algorithm keeps track of a latest
articulation point $u'=a_i$ seen on $u\cdots v$ (initially $u'=a_0=u$)
such that either $u'=u$ or $(u,u')$ is admissible in the current
embedding. We will further maintain the invariant that (unless
$i=k$) there is a common face of $u$ and $a_i$ that contains at
least one edge from $B_{i+1}$.
In the round where $u'=a_i$ the algorithm sets $v'=a_{i+1}$ and does
the following:
\begin{enumerate}
\item It finds and applies all separation flips in $B_{i+1}$ needed to
  make $(a_i,a_{i+1})$ admissible, or detects (after some number of
  flips) that $B_{i+1}\cup(a_i,a_{i+1})$ --- and therefore $G\cup(u,v)$
  --- is nonplanar.

\item It finds at most one articulation flip at $u'$ and at most one
  articulation flip at $v'$, such that afterwards $u$ shares a face
  with $v'=a_{i+1}$, and (if $v'\neq v$) with at least one edge from
  $B_{i+2}$.

\item It finds the first $a_j$ with $j\geq i+1$ such that either: the
  next iteration of the loop finds at least one flip; or no more flips
  are needed and $a_j=v$. It then sets $u'\gets a_j$.
\end{enumerate}
The algorithm stops when $u'=v$. By our invariant $(u,v)$ is
admissible if it reaches this point.

\begin{lemma}\label{lem:multiflip-iterations}
  If $G\cup(u,v)$ is planar and $H_0\in\Emb(G)$, this algorithm finds
  a sequence of graphs $H_1,\ldots,H_k\in\Emb(G)$ such that
  $H_k\in\Emb(G;u,v)$. The main loop performs $\OO(k)$ iterations.
\end{lemma}
\begin{proof}
  If $G\cup(u,v)$ is planar, then for every block $B_{i+1}$ there
  exists some (possibly empty) set of separation flips such that
  $H\cup(a_i,a_{i+1})$ is planar. Thus, after
  line~\ref{line:call-do-sep} we know that $a_i$ and $a_{i+1}$ share a
  face. And by our invariant, we also know that $u=a_0$ and $u'=a_i$
  share a face, incident to at least one edge from $B_{i+1}$.

  Now the call to \textproc{do-articulation-flips} in
  line~\ref{line:call-do-art} uses at most $2$ articulation flips to
  update the invariant so $u$ shares a face with $v'$ and (if $v\neq
  v'$) with at least one edge in $B_{i+2}$.

  Finally, the call to \textproc{find-next-flip-block} in
  line~\ref{line:next-block}, updates $u'$ to the largest $a_j$ so the
  invariant still holds. In particular, either a separation flip is
  needed in $B_{j+1}$, or an articulation flip is needed in $a_{j+1}$
  before the face shared by $u=a_0$ and $a_{j+1}$ is incident to an
  edge in $B_{j+2}$.

  Thus, in every iteration after (possibly) the first, at least one
  flip is performed. Thus if we stop after $k$ flips, the number of
  iterations is at most $k+1$.

  We stop with $H\not\in\Emb(G;u,v)$ only if there is a block $B_{i+1}$ where
  \Call{do-articulation-flips}{$a_i,a_{i+1}$} returns ``no'' because we
  have detected that $G\cup(u,v)$ is nonplanar.

  Otherwise we keep making progress, and will eventually have $u'=v$.
  By our invariant $(u,v)$ are now in the same face, and thus
  $H_k\in\Emb(G;u,v)$.
\end{proof}

\FloatBarrier
\subsection{\textproc{find-next-flip-block}}
The simplest part of our algorithm is the
\textproc{find-next-flip-block} function in
Algorithm~\ref{alg:find-next-flip-block} we use to move to the next
``interesting'' articulation point, or to vertex $v$ if we are done.

By interesting is meant the following: it is an articulation point
$a_j$ on the BC-path from $u$ to $v$ such that a flip in either $a_j$,
$B_{j+1}$, or $a_{j+1}$ is necessary in order to bring $u$ and $v$ to
the same face.

\begin{algorithm}[htb!]
  \caption{}
  \label{alg:find-next-flip-block}
  \begin{algorithmic}[1]

    \Function{find-next-flip-block}{$u,u',v',v$}

        \LeftComment{$u$ shares a face with $v'$, and
        (if $v'\neq v$) with at least one edge in $B_{i+2}$.}

        \If{$v'\neq v$}
          \State $f_u,c_u^1,c_u^2\gets\Call{find-bounding-face}{u,v',v}$
          \Comment{$f_u$ is incident to $u$.}
          \State $f_v,c_v^1,c_v^2\gets\Call{find-bounding-face}{v,v',u}$
          \If{$f_u=f_v$}
            \If{$v$ incident to $f_u$}
              \State \Return $v$
            \EndIf
            \State\label{line:find-next-art} \Return last internal node on $u\cdots v$ touching
            $f_u$ on both sides.
          \EndIf
        \EndIf
        \State \Return $v'$
    \EndFunction

    \Statex

    \Function{find-bounding-face}{$u,a,v$}

      \State $c_u\gets$ any corner incident to $u$; $c_v\gets$ any
      corner incident to $v$

      \State $f\gets$ first face on the dual path $c_u\cdots c_v$
      touching $a$ on both sides.

      \State $c_L,c_R\gets$ first corners on left and right side of
      $c_u\cdots c_v$ that are incident to both $a$ and $f$.

      \State \Return $f,c_L,c_R$

    \EndFunction

  \end{algorithmic}
\end{algorithm}

\begin{lemma}\label{lem:next-flip}
  If $a_0$ and $a_{i+1}$ share a face containing at least one edge
  from $B_{i+2}$, then in worst case $\OO(\log^2 n)$ time
  $\Call{find-next-flip-block}{a_0,a_i,a_{i+1},a_k}$ either returns
  the last $a_j$ ($i<j<k$) such that $a_0$ and $a_j$ share a face
  containing at least one edge from $B_{j+1}$; or $a_k$ if $a_0$ and
  $a_k$ share a face.
\end{lemma}
\begin{proof}
First note that the running time is $\OO(\log ^2 n)$ because the dominating subroutine is the mark-and-search algorithm from~\cite{DBLP:journals/mst/HolmR17} called a constant number of times in $\Call{Find-bounding-face}{\ldots}$, and once on line~\ref{line:find-next-art}.

For correctness,
assume $a=a_{i+1}$ is an articulation point separating $u=a_0$ from $v=a_k$. Then there is at least one articulation point in the dual graph incident to $a$, and all such articulation points lie on a path in the dual tree. By assumption, $u$ lies in the bounding face of $a_i$ and $a_{i+1}$, that is, in the face $f$ returned by $\Call{Find-bounding-face}{u,a_i,a_{i+1}}$. 
Now there are four basic cases: 
\begin{enumerate}
\item\label{case:next-flip-1} If $a_{i+2}$ does not lie in $f$, then we have indeed reached a block where flips are necessary, and the algorithm returns $a_{i+1}$ as desired. (See Figure~\ref{fig:next-flip-1})
\item\label{case:next-flip-2} On the other hand, if $a_{i+2}$ lies in $f$ but only has corners incident to $f$ on one side, then we are in a case where a flip in $a_{i+2}$ is necessary to bring $u$ to the same face as $v$, and the algorithm returns $a_{i+1}$ as desired. (See Figure~\ref{fig:next-flip-2})
\item\label{case:next-flip-3} Thirdly, it may be the case that $a_{i+2}$ is incident to $f$ on both sides of the path, in which case no flips in $a_{i+1},B_{i+1}$, or $a_{i+2}$ are necessary. In this case there is a non-trivial segment of articulation points $a_{i+1},a_{i+2},\ldots$ that lie in $f$. Our algorithm will now return the last such $a_j$ incident to $f$ on both sides of the path. Here, we have two sub-cases. If $a_{j+1}$ is not incident to $f$, this indicates that either we have reached a point $a_j$ where an articulation flip is needed, or, we have reached a block $B_j$ where flips are needed. On the other hand, if $a_{j+1}$ is incident to $f$ but only on one side, we are in the case where an articulation flip in $a_{j+1}$ is needed to bring $u$ and $v$ to the same face. In both cases, our algorithm returns $a_j$ as desired. (See Figure~\ref{fig:next-flip-3})
\item\label{case:next-flip-4} Finally, we may be in the case where $v$ lies in the same face as $u$ and we are done, but in this case, the face $f$ must be the shared face: Namely, since $u$ and $v$ are separated by at least one articulation point $a$ in the primal graph with $a$ incident to some face $f$ on both sides of the tree-path from $u$ to $v$, then there is a $2$-cycle through $f$ and $a$ in the vertex-face graph separating $u$ from $v$, and thus, $f$ must be the unique face shared by all articulation points on any path $u\cdots v$. (See Figure~\ref{fig:next-flip-4})\qedhere
\end{enumerate}
\end{proof}

\begin{figure}[H]
  \centering
\begin{tikzpicture}[
  vertex/.style = {
    circle, draw, fill=white,
    inner sep = 2pt,
  },
  block/.style = {
  },
]
  \node[vertex, label={$a_0$}] (a0) at (0,0) {};
  \node[block] (B1) at (1,0) {$B_1$};
  \node[vertex, label={$a_1$}] (a1) at (2,0) {};

  \node[vertex, label={$a_i$}] (ai) at (3,0) {};
  \node[block] (Bi1) at (4,0) {$B_{i+1}$};
  \node[vertex, label={$a_{i+1}$}] (ai1) at (5,0) {};
  \node[block] (Bi2) at (6,0) {$B_{i+2}$};
  \node[vertex, label={$a_{i+2}$}] (ai2) at (7,0) {};

  \begin{pgfonlayer}{background}
    \draw[fill=gray!20] (a0) to[bend right=30] (a1.center) to[bend right=30] (a0);
    \draw[thick,dotted] (a1) -- (ai);
    \draw[fill=gray!20] (ai) to[bend right=30] (ai1.center) to[bend right=30] (ai);

    \draw[fill=gray!20] plot[smooth,tension=1] coordinates {(ai1) (7,1) (8,0) (7,-1) (ai1)};

  \end{pgfonlayer}
\end{tikzpicture}
   \caption{Lemma~\ref{lem:next-flip} case~\ref{case:next-flip-1}:
    $a_{i+2}$ does not lie in $f$. %
    $\protect\Call{find-next-flip-block}{a_0,a_i,a_{i+1},a_k}$ returns
    $a_{i+1}$.}
  \label{fig:next-flip-1}
\end{figure}
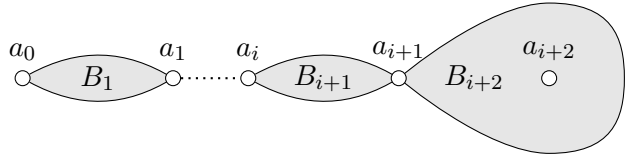

\begin{figure}[H]
  \centering
\begin{tikzpicture}[
  vertex/.style = {
    circle, draw, fill=white,
    inner sep = 2pt,
  },
  block/.style = {
  },
]
  \node[vertex, label={$a_0$}] (a0) at (0,0) {};
  \node[block] (B1) at (1,0) {$B_1$};
  \node[vertex, label={$a_1$}] (a1) at (2,0) {};

  \node[vertex, label={$a_i$}] (ai) at (3,0) {};
  \node[block] (Bi1) at (4,0) {$B_{i+1}$};
  \node[vertex, label={$a_{i+1}$}] (ai1) at (5,0) {};
  \node[block] (Bi2) at (5.6,0) {$B_{i+2}$};
  \node[vertex, label={right:$a_{i+2}$}] (ai2) at (8,0) {};

  \node[vertex, label={$a_{i+3}$}] (ai3) at (7,0) {};

  \begin{pgfonlayer}{background}
    \draw[fill=gray!20] (a0) to[bend right=30] (a1.center) to[bend right=30] (a0);
    \draw[thick,dotted] (a1) -- (ai);
    \draw[fill=gray!20] (ai) to[bend right=30] (ai1.center) to[bend right=30] (ai);

    \draw[fill=gray!20] plot[smooth,tension=1] coordinates {(ai1) (6.5,1) (8,0) (6.5,-1) (ai1)};

    \draw[fill=white] plot[smooth,tension=1] coordinates {(ai2) (6.75,.5) (6,0) (6.75,-.5) (ai2)};

    \draw[fill=gray!20] (ai2) to[bend right=30] (ai3.center) to[bend right=30] (ai2);

  \end{pgfonlayer}
\end{tikzpicture}
   \caption{Lemma~\ref{lem:next-flip} case~\ref{case:next-flip-2}: the
    path $a_0\cdots a_k$ touches $f$ on only one side in
    $a_{i+2}$. %
    $\protect\Call{find-next-flip-block}{a_0,a_i,a_{i+1},a_k}$ returns
    $a_{i+1}$.}
  \label{fig:next-flip-2}
\end{figure}
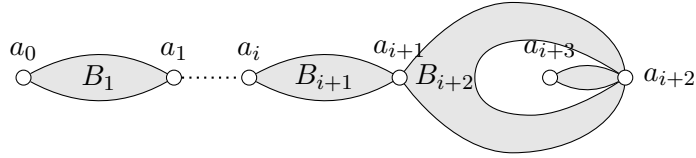

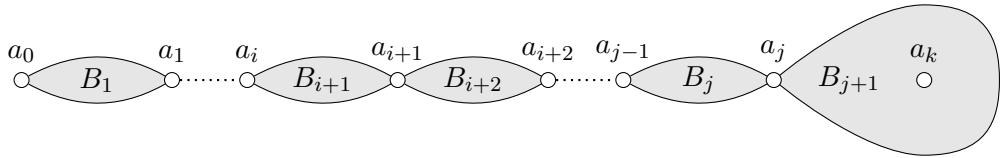
\begin{figure}[H]
  \centering
\begin{tikzpicture}[
  vertex/.style = {
    circle, draw, fill=white,
    inner sep = 2pt,
  },
  block/.style = {
  },
]
  \node[vertex, label={$a_0$}] (a0) at (0,0) {};
  \node[block] (B1) at (1,0) {$B_1$};
  \node[vertex, label={$a_1$}] (a1) at (2,0) {};

  \node[vertex, label={$a_i$}] (ai) at (3,0) {};
  \node[block] (Bi1) at (4,0) {$B_{i+1}$};
  \node[vertex, label={$a_{i+1}$}] (ai1) at (5,0) {};
  \node[block] (Bi2) at (6,0) {$B_{i+2}$};
  \node[vertex, label={$a_{i+2}$}] (ai2) at (7,0) {};

  \node[vertex, label={$a_{j-1}$}] (ajm) at (8,0) {};
  \node[block] (Bj) at (9,0) {$B_j$};
  \node[vertex, label={$a_j$}] (aj) at (10,0) {};
  \node[block] (Bj1) at (11,0) {$B_{j+1}$};

  \node[vertex, label={$a_k$}] (ak) at (12,0) {};

  \begin{pgfonlayer}{background}
    \draw[fill=gray!20] (a0) to[bend right=30] (a1.center) to[bend right=30] (a0);
    \draw[thick,dotted] (a1) -- (ai);
    \draw[fill=gray!20] (ai) to[bend right=30] (ai1.center) to[bend right=30] (ai);
    \draw[fill=gray!20] (ai1) to[bend right=30] (ai2.center) to[bend right=30] (ai1);

    \draw[thick,dotted] (ai2) -- (ajm);
    \draw[fill=gray!20] (ajm) to[bend right=30] (aj.center) to[bend right=30] (ajm);

    \draw[fill=gray!20] plot[smooth,tension=1] coordinates {(aj) (12,1) (13,0) (12,-1) (aj)};

  \end{pgfonlayer}
\end{tikzpicture}
   \caption{Lemma~\ref{lem:next-flip} case~\ref{case:next-flip-3}:
    $\protect\Call{find-next-flip-block}{a_0,a_i,a_{i+1},a_k}$ returns
    $a_j$ with $i+1<j<k$.}
  \label{fig:next-flip-3}
\end{figure}

\begin{figure}[H]
  \centering
\begin{tikzpicture}[
  vertex/.style = {
    circle, draw, fill=white,
    inner sep = 2pt,
  },
  block/.style = {
  },
]
  \node[vertex, label={$a_0$}] (a0) at (0,0) {};
  \node[block] (B1) at (1,0) {$B_1$};
  \node[vertex, label={$a_1$}] (a1) at (2,0) {};

  \node[vertex, label={$a_i$}] (ai) at (3,0) {};
  \node[block] (Bi1) at (4,0) {$B_{i+1}$};
  \node[vertex, label={$a_{i+1}$}] (ai1) at (5,0) {};
  \node[block] (Bi2) at (6,0) {$B_{i+2}$};
  \node[vertex, label={$a_{i+2}$}] (ai2) at (7,0) {};

  \node[vertex, label={$a_{k-1}$}] (akm) at (8,0) {};
  \node[block] (Bk) at (9,0) {$B_k$};
  \node[vertex, label={$a_k$}] (ak) at (10,0) {};

  \begin{pgfonlayer}{background}
    \draw[fill=gray!20] (a0) to[bend right=30] (a1.center) to[bend right=30] (a0);
    \draw[thick,dotted] (a1) -- (ai);
    \draw[fill=gray!20] (ai) to[bend right=30] (ai1.center) to[bend right=30] (ai);
    \draw[fill=gray!20] (ai1) to[bend right=30] (ai2.center) to[bend right=30] (ai1);

    \draw[thick,dotted] (ai2) -- (akm);
    \draw[fill=gray!20] (akm) to[bend right=30] (ak.center) to[bend right=30] (akm);

  \end{pgfonlayer}
\end{tikzpicture}
   \caption{Lemma~\ref{lem:next-flip} case~\ref{case:next-flip-4}:
    $\protect\Call{find-next-flip-block}{a_0,a_i,a_{i+1},a_k}$ returns
    $a_k$.}
  \label{fig:next-flip-4}
\end{figure}
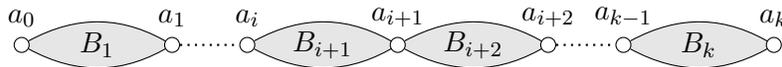

\FloatBarrier
\subsection{\textproc{do-articulation-flips}}

The next piece of our algorithm is the \textproc{do-articulation-flips}
function in Algorithm~\ref{alg:do-art}.

\begin{algorithm}[htb!]
  \caption{}
  \label{alg:do-art}
  \begin{algorithmic}[1]
    \Require $u',v'$ linkable%
    \Function{do-articulation-flips}{$u,u',v',v$}
      \If{$u=u'$}
        \If{$v'=v$}
          \LeftComment{Nothing to do}
        \Else
          \State $f_v,c_v^1,c_v^2\gets \Call{find-bounding-face}{v,v',u}$
          \If{$f_v$ not incident to $u'$}
            \State $c_u,c_v\gets\Call{linkable}{u',v'}$
            \State\label{line:aflip-v1} \Call{articulation-flip}{$c_v^1,c_v^2,c_v$}
          \EndIf
        \EndIf
      \Else
        \If{$v'=v$}
          \State $f_u,c_u^1,c_u^2\gets \Call{find-bounding-face}{u,u',v}$
          \If{$f_u$ not incident to $v'$}
            \State $c_u,c_v\gets\Call{linkable}{u',v'}$
            \State\label{line:aflip-u1} \Call{articulation-flip}{$c_u^1,c_u^2,c_u$}
          \EndIf
        \Else{ $u\neq u'$ and $v'\neq v$}
          \State $f_v,c_v^1,c_v^2\gets \Call{find-bounding-face}{v,v',u}$
          \State $f_u,c_u^1,c_u^2\gets \Call{find-bounding-face}{u,u',v}$
          \If{$f_u=f_v$}
            \LeftComment{Nothing to do}
          \ElsIf{$u'$ incident to $f_v$}
            \State $c_u\gets$ any corner between $u'$ and $f_v$
            \State\label{line:aflip-u2} \Call{articulation-flip}{$c_u^1,c_u^2,c_u$}
          \ElsIf{$v'$ incident to $f_u$}
            \State $c_v\gets$ any corner between $v'$ and $f_u$
            \State\label{line:aflip-v2} \Call{articulation-flip}{$c_v^1,c_v^2,c_v$}
          \Else{}
            \State $c_u,c_v\gets\Call{linkable}{u',v'}$
            \State\label{line:aflip-uv1} \Call{articulation-flip}{$c_u^1,c_u^2,c_u$}
            \State\label{line:aflip-uv2} \Call{articulation-flip}{$c_v^1,c_v^2,c_v$}
          \EndIf
        \EndIf
      \EndIf
    \EndFunction

  \end{algorithmic}
\end{algorithm}

\begin{lemma}\label{lem:do-art}
  If $a_0$ and $a_i$ share a face containing at least one edge from
  $B_{i+1}$, and $a_i$ and $a_{i+1}$ share a face, then in $\OO(\log^2 n)$
  time $\Call{do-articulation-flips}{a_0,a_i,a_{i+1},a_k}$ does at
  most one articulation flip at each of $a_i$ and $a_{i+1}$, each of
  which are critical. Afterwards, $a_0$ shares a face with $a_{i+1}$,
  and (if $i\leq k-2$) with at least one edge in $B_{i+2}$.
\end{lemma}
\begin{proof}
  The running time of $\Call{do-articulation-flips}{\cdots}$ is
  $\OO(\log^2 n)$ because it does a constant number of calls to
  $\Call{find-bounding-face}{\cdots}$, $\Call{linkable}{\cdots}$, and
  $\Call{articulation-flip}{\cdots}$, which each take worst case
  $\OO(\log^2 n)$ time.
  The number and location of the articulation flips done is likewise
  clear from the definition, and since each of the flips move exactly
  one of $u$ and $v$, any flips done are critical.
  \begin{itemize}
  \item If no articulation flips are done: $u$ and $v$ are already in
    the same face. (See Figure~\ref{fig:do-art-1})

  \item If an articulation flip is done at $v'=a_{i+1}$ but not at
    $u'=a_i$: $B_{i+2}$ is flipped into a face containing $u$ in
    line~\ref{line:aflip-v1} if $u=u'$, or a face containing $u'$ and
    $B_i$ (and hence $u$) in line~\ref{line:aflip-v2} otherwise. (See
    Figure~\ref{fig:do-art-2})

  \item If an articulation flip is done at $u'=a_i$ but not at
    $v'=a_{i+1}$: $B_i$ (and hence $u$) is flipped into a face
    containing $v'$ in line~\ref{line:aflip-u1} if $v'=v$, or a face
    containing $v'$ and $B_{i+2}$ in line~\ref{line:aflip-u2} otherwise.
    (See Figure~\ref{fig:do-art-3})

  \item If articulation flips are done at both $u'=a_i$ and
    $v'=a_{i+1}$: Both $B_i$ (and hence $u$) and $B_{i+2}$ are flipped
    into the same face in
    lines~\ref{line:aflip-uv1}--\ref{line:aflip-uv2}.
    (See Figure~\ref{fig:do-art-4})
  \end{itemize}
  In each case the postcondition is satisfied.
\end{proof}

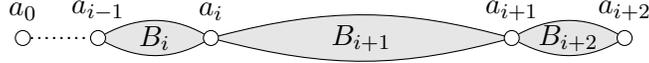
\begin{figure}[H]
  \centering
\begin{tikzpicture}[
  vertex/.style = {
    circle, draw, fill=white,
    inner sep = 2pt,
  },
  block/.style = {
  },
]
  \node[vertex, label={$a_0$}] (a0) at (0.5,0) {};
  \node[vertex, label={$a_{i-1}$}] (aim) at (1.5,0) {};
  \node[block] (Bi) at (2.25,0) {$B_i$};

  \node[vertex, label={$a_i$}] (ai) at (3,0) {};
  \node[block] (Bi1) at (5,0) {$B_{i+1}$};
  \node[vertex, label={$a_{i+1}$}] (ai1) at (7,0) {};
  \node[block] (Bi2) at (7.75,0) {$B_{i+2}$};
  \node[vertex, label={$a_{i+2}$}] (ai2) at (8.5,0) {};

  \begin{pgfonlayer}{background}

    \draw[fill=gray!20] (ai) to[bend right=15] (ai1.center) to[bend right=15] (ai);

    \draw[thick,dotted] (a0) -- (aim);

    \draw[fill=gray!20] (aim) to[bend right=30] (ai.center) to[bend right=30] (aim);
    \draw[fill=gray!20] (ai1) to[bend right=30] (ai2.center) to[bend right=30] (ai1);

  \end{pgfonlayer}
\end{tikzpicture}
   \caption{
    $\protect\Call{do-articulation-flips}{a_0,a_i,a_{i+1},a_k}$: If
    $f_u=f_v$ there is nothing to do. }
  \label{fig:do-art-1}
\end{figure}
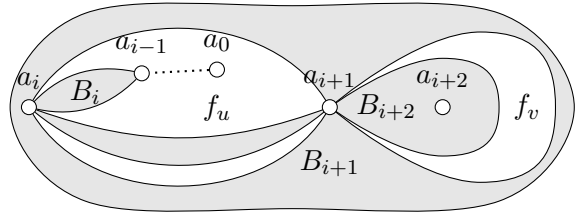
\begin{figure}[H]
  \centering
\begin{tikzpicture}[
  vertex/.style = {
    circle, draw, fill=white,
    inner sep = 2pt,
  },
  block/.style = {
  },
]

  \node[vertex, label={$a_0$}] (a0) at (5.5,.5) {};
  \node[vertex, label={$a_{i-1}$}] (aim) at (4.5,.45) {};
  \node[block] (Bi) at (3.75,.2) {$B_i$};

  \node[vertex, label={$a_i$}] (ai) at (3,0) {};
  \node[block] (Bi1) at (7,-.75) {$B_{i+1}$};
  \node[vertex, label={$a_{i+1}$}] (ai1) at (7,0) {};
  \node[block] (Bi2) at (7.75,0) {$B_{i+2}$};
  \node[vertex, label={$a_{i+2}$}] (ai2) at (8.5,0) {};

  \node[below=.1 of a0] {$f_u$};
  \node[right=.7 of ai2] {$f_v$};

  \begin{pgfonlayer}{background}

    \draw[fill=gray!20] plot[smooth cycle,tension=.6] coordinates {
      (6.5,1.35) (3.75,1.25) (2.75,0) (3.75,-1.25)
      (6.5,-1.35)
      (9.25,-1.25) (10.25,0) (9.25,1.25) };

    \draw[fill=white] (ai) to[bend right=60] (ai1.center) to[bend right=60] (ai);
    \draw[fill=gray!20] (ai) to[bend right=40] (ai1.center) to[bend left=20] (ai);

    \draw[fill=white] plot[smooth,tension=1] coordinates {(ai1) (9,1) (10,0) (9,-1) (ai1)};

    \draw[thick,dotted] (a0) -- (aim);

    \draw[fill=gray!20] (aim) to[bend right=30] (ai.center) to[bend right=30] (aim);
    \draw[fill=gray!20] (ai1) to[bend right=30] (ai2.center) to[bend right=30] (ai1);

    \draw[fill=gray!20] plot[smooth,tension=1] coordinates {(ai1) (8.5,.65) (9.25,0) (8.5,-.65) (ai1)};

  \end{pgfonlayer}
\end{tikzpicture}
   \caption{
    $\protect\Call{do-articulation-flips}{a_0,a_i,a_{i+1},a_k}$: When
    $f_u\neq f_v$ and $a_{i+1}\in f_u$, we may use any corner between
    $f_u$ and $a_{i+1}$ to flip $B_{i+2}$ into.}
  \label{fig:do-art-2}
\end{figure}
\begin{figure}[H]
  \centering
\begin{tikzpicture}[
  vertex/.style = {
    circle, draw, fill=white,
    inner sep = 2pt,
  },
  block/.style = {
  },
]

  \node[vertex, label={$a_0$}] (a0) at (0.5,0) {};
  \node[vertex, label={$a_{i-1}$}] (aim) at (1.5,0) {};
  \node[block] (Bi) at (2.25,0) {$B_i$};

  \node[vertex, label={$a_i$}] (ai) at (3,0) {};
  \node[block] (Bi1) at (3,-.75) {$B_{i+1}$};
  \node[vertex, label={$a_{i+1}$}] (ai1) at (7,0) {};
  \node[block] (Bi2) at (6.25,.2) {$B_{i+2}$};
  \node[vertex, label={$a_{i+2}$}] (ai2) at (5.5,.3) {};

  \node[below right=.4 of a0] {$f_u$};
  \node[left=.7 of ai2] {$f_v$};

  \begin{pgfonlayer}{background}

    \draw[fill=gray!20] plot[smooth cycle,tension=.6] coordinates {
      (3.5,1.35) (0.75,1.25) (-.25,0) (0.75,-1.25)
      (3.5,-1.35)
      (6.25,-1.25) (7.25,0) (6.25,1.25) };

    \draw[fill=white] plot[smooth,tension=1] coordinates {(ai) (1,1) (0,0) (1,-1) (ai)};

    \draw[fill=white] (ai) to[bend right=60] (ai1.center) to[bend right=60] (ai);
    \draw[fill=gray!20] (ai) to[bend right=40] (ai1.center) to[bend left=20] (ai);

    \draw[thick,dotted] (a0) -- (aim);

    \draw[fill=gray!20] (aim) to[bend right=30] (ai.center) to[bend right=30] (aim);

    \draw[fill=gray!20] plot[smooth,tension=1] coordinates {(ai1) (5.5,.85) (4.75,.4) (5.65,-.1) (ai1)};

  \end{pgfonlayer}
\end{tikzpicture}
   \caption{
    $\protect\Call{do-articulation-flips}{a_0,a_i,a_{i+1},a_k}$: When
    $f_u\neq f_v$ and $a_i\in f_v$, we use any corner between
    $f_v$ and $a_i$ to flip $B_i$ into.}
  \label{fig:do-art-3}
\end{figure}
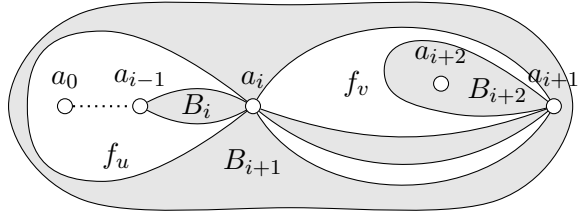
\begin{figure}[H]
  \centering
\begin{tikzpicture}[
  vertex/.style = {
    circle, draw, fill=white,
    inner sep = 2pt,
  },
  block/.style = {
  },
]

  \node[vertex, label={$a_0$}] (a0) at (0.5,0) {};
  \node[vertex, label={$a_{i-1}$}] (aim) at (1.5,0) {};
  \node[block] (Bi) at (2.25,0) {$B_i$};

  \node[vertex, label={$a_i$}] (ai) at (3,0) {};
  \node[block] (Bi1) at (3,-.75) {$B_{i+1}$};
  \node[vertex, label={$a_{i+1}$}] (ai1) at (7,0) {};
  \node[block] (Bi2) at (7.75,0) {$B_{i+2}$};
  \node[vertex, label={$a_{i+2}$}] (ai2) at (8.5,0) {};

  \node[below right=.4 of a0] {$f_u$};
  \node[right=.7 of ai2] {$f_v$};

  \draw[dashed] (ai) to[bend left=30] (ai1);

  \begin{pgfonlayer}{background}

    \draw[fill=gray!20] plot[smooth cycle,tension=.6] coordinates {
      (5,1.35) (0.75,1.25) (-.25,0) (0.75,-1.25)
      (5,-1.35)
      (9.25,-1.25) (10.25,0) (9.25,1.25) };

    \draw[fill=white] plot[smooth,tension=1] coordinates {(ai) (1,1) (0,0) (1,-1) (ai)};

    \draw[fill=white] (ai) to[bend right=60] (ai1.center) to[bend right=60] (ai);
    \draw[fill=gray!20] (ai) to[bend right=10] (ai1.center) to[bend right=10] (ai);

    \draw[fill=white] plot[smooth,tension=1] coordinates {(ai1) (9,1) (10,0) (9,-1) (ai1)};

    \draw[thick,dotted] (a0) -- (aim);

    \draw[fill=gray!20] (aim) to[bend right=30] (ai.center) to[bend right=30] (aim);

    \draw[fill=gray!20] plot[smooth,tension=1] coordinates {(ai1) (8.5,.65) (9.25,0) (8.5,-.65) (ai1)};

  \end{pgfonlayer}
\end{tikzpicture}
   \caption{
    $\protect\Call{do-articulation-flips}{a_0,a_i,a_{i+1},a_k}$: When
    $a_{i+1}\not\in f_u$ and $a_i\not\in f_v$, use a
    $\protect\Call{linkable}{a_i,a_{i+1}}$ query to find the corners
    where $(a_i,a_{i+1})$ could be inserted, and flip into those corners.}
  \label{fig:do-art-4}
\end{figure}
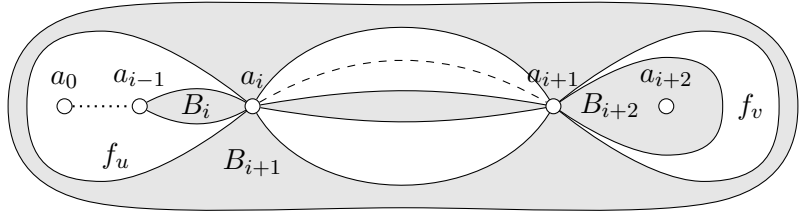

\begin{lemma}\label{lem:art-flip-pattern}
  During execution of $\Call{multi-flip-linkable}{u,v}$, the pattern
  of articulation flips is as follows:

  \begin{itemize}
  \item If $\Call{do-articulation-flips}{a_0,a_i,a_{i+1},a_k}$ does a
    flip at $a_i$, it is potential-decreasing.

  \item If $\Call{do-articulation-flips}{a_0,a_i,a_{i+1},a_k}$ does a
    flip at $a_{i+1}$ that is not potential-decreasing, then it is
    potential-neutral, the following
    $\Call{find-next-flip-block}{a_0,a_i,a_{i+1},a_k}$ returns
    $a_{i+1}$, and in the next iteration either:
    \begin{itemize}
    \item $\Call{do-separation-flips}{a_{i+1},a_{i+2}}$ returns ``no''; or
    \item $\Call{do-separation-flips}{a_{i+1},a_{i+2}}$ does at least
      one flip; or
    \item $\Call{do-articulation-flips}{a_0,a_{i+1},a_{i+2},a_k}$ does
      a potential-decreasing articulation flip at $a_{i+1}$.
    \end{itemize}
  \end{itemize}
\end{lemma}
\begin{proof}
  If $\Call{do-articulation-flips}{a_0,a_i,a_{i+1},a_k}$ does a flip
  at $u'=a_i$, it is because $u$ and $v'$ does not yet share a face
  and at least one flip at $a_i$ is needed to bring them together. By
  our invariants, $u$ does share a face with both $a_i$ and at least
  one edge of $B_{i+1}$. Observe that no amount of clean separation
  flips can help bringing them together, and thus any clean path in
  $\Emb(G)$ from $H$ to $H'\in\Emb(G;u,v)$ contains at least one
  separation flip at $a_i$. In particular, any shortest path contains
  exactly one such flip, and thus this flip is potential-decreasing.

  On the other hand, if
  $\Call{do-articulation-flips}{a_0,a_i,a_{i+1},a_k}$ does a flip at
  $u'=a_{i+1}$ because some other block $B'$ is incident to $a_{i+1}$
  on both sides of $u\cdots v$, then it can happen that (after all
  separation flips in $B_{i+2}$ are done) $u$ is still not sharing a
  face with $a_{i+2}$ even though $a_{i+1}$ is (See
  Figure~\ref{fig:neutral-art-flip}). In this case, a second
  articulation flip is needed at $a_{i+1}$ (which will be
  potential-decreasing).  However, if $a_i$ shares a face with $B'$,
  the first flip at $a_{i+1}$ could have been skipped, and we would
  reach the invariant state with one fewer flips.  In this case, and
  only this case, the first flip the algorithm does at $a_{i+1}$ is
  potential-neutral. Now observe that if we do such a
  potential-neutral flip, the following
  $\Call{find-next-flip-block}{a_0,a_i,a_{i+1},a_k}$ must return
  $a_{i+1}$, because either $a_{i+1}$ does not share a face with
  $a_{i+2}$ and $\Call{do-separation-flips}{a_{i+1},a_{i+2}}$ will
  behave as described, or $a_{i+1}$ does share a face with $a_{i+2}$
  but $a_0$ does not share a face with $a_{i+2}$ so the second
  articulation flip at $a_{i+1}$ is done as described.
\end{proof}

\begin{figure}[H]
  \centering
\begin{tikzpicture}[
  vertex/.style = {
    circle, draw, fill=white,
    inner sep = 2pt,
  },
  block/.style = {
  },
]
  \node[vertex, label={$a_0$}] (a0) at (0,0) {};
  \node[block] (B1) at (1,0) {$B_1$};
  \node[vertex, label={$a_1$}] (a1) at (2,0) {};

  \node[vertex] at (3.5,0) {};

  \node[vertex, label={$a_i$}] (ai) at (5,0) {};
  \node[block] (Bi1) at (6,0) {$B_{i+1}$};
  \node[vertex, label={$a_{i+1}$}] (ai1) at (7,0) {};

  \node[block] (Bi2) at (9.15,-.7) {$B_{i+2}$};

  \node[vertex, label={right:$a_{i+2}$}] (ai2) at (9,0) {};

  \node[block] (Bp) at (10.75,0) {$B'$};

  \begin{pgfonlayer}{background}
    \draw[fill=gray!20] (a0) to[bend right=30] (a1.center) to[bend right=30] (a0);
    \draw[thick,dotted] (a1) -- (ai);
    \draw[fill=gray!20] (ai) to[bend right=30] (ai1.center) to[bend right=30] (ai);

    \draw[fill=gray!20] plot[smooth,tension=1] coordinates {(ai1) (9,1.75)
      (11,0) (9,-1.75) (ai1)};
    \draw[fill=white] plot[smooth,tension=1] coordinates {(ai1) (9,1.5)
      (10.5,0) (9,-1.5) (ai1)};

    \draw[fill=gray!20] plot[smooth,tension=1] coordinates {(ai1) (9,1) (10,0) (9,-1) (ai1)};

    \draw[fill=white] plot[smooth,tension=1] coordinates {(ai1) (8.5,.5) (ai2) (8.5,-.5) (ai1)};

  \end{pgfonlayer}
\end{tikzpicture}
   \caption{Example where \textproc{do-articulation-flips} makes a
    potential-neutral flip at $a_{i+1}$, moving $B_{i+2}$ into the outer face.
    With or without this flip, we still need to flip $B_1\cdots
    B_{i+1}$ into a face containing $a_{i+2}$.}
  \label{fig:neutral-art-flip}
\end{figure}
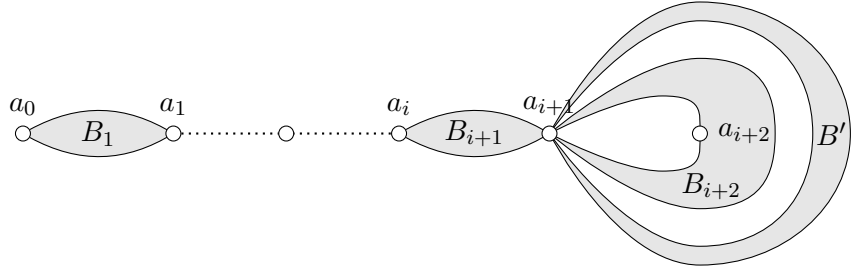

\FloatBarrier
\subsection{If $u,v$ are biconnected (\textproc{do-separation-flips})}
\label{subsec:do-sep}

The remaining piece of our algorithm, and by far the most complicated, consists of the \textproc{do-separation-flips} function in Algorithm~\ref{alg:do-sep}, and its subroutines.

The general idea in $\Call{do-separation-flips}{u,v}$ is to repeatedly
find and apply a separation flip $\sigma$ (uniquely described by a
tuple of $4$ corners forming a $4$-cycle in the vertex-face graph),
such that
\begin{itemize}
\item The vertices incident to $\sigma$ form a separation pair
  $\set{s,t}$, separating $u$ from $v$.
\item $u$ is incident to at least one of the faces $f_u,f_v$ incident
  to $\sigma$.
\item $\sigma$ partitions (the edges of) $H$ into two subgraphs
  $H_u,H_v$, with $u\in H_u\setminus\set{s,t}$ and $v\in
  H_v\setminus\set{s,t}$.
\end{itemize}

We call such a flip a \emph{$u$-flip} (w.r.t.\@ $v$), and the
corresponding $H_u$ a \emph{$u$-flip-component} (w.r.t.\@ $v$).  We
call a $u$-flip \emph{maximal} if it maximizes the size (e.g. edges
plus vertices) of $H_u$.

A given $u$-flip-component $H_u$ remains a $u$-flip-component if we
flip it, so we may require that each step flips a strictly larger
subgraph than the previous. If no strictly larger $H_u$ exists
and $u$ and $v$ still do not share a face, we conclude that
$G\cup(u,v)$ is nonplanar and stop.

\begin{algorithm}[htb!]
  \caption{}
  \label{alg:do-sep}
  \begin{algorithmic}[1]

    \Function{do-separation-flips}{$u,v$}
        \State $s\gets 0$
        \While{not \Call{linkable}{$u,v$}}
          \Comment{Separation flip needed in block bounded by $u,v$}
          \State $s',\sigma\gets \Call{find-first-separation-flip}{u,v}$
          \If{$s'\leq s$}
            \Return ``no''
          \EndIf
          \State Execute separation flip $\sigma$
          \State $s\gets s'$
        \EndWhile
        \State \Return ``yes''
    \EndFunction

  \end{algorithmic}
\end{algorithm}

\begin{lemma}\label{lem:do-sep}
  Assume that $\Call{find-first-separation-flip}{u,v}$ runs in worst case
  $\OO(\log^2 n)$ time, and in each step:
  \begin{itemize}
  \item If $G\cup(u,v)$ is planar it finds a maximal $u$-flip.
  \item If $G\cup(u,v)$ is non-planar it either:
    \begin{itemize}
    \item finds a maximal $u$-flip; or
    \item finds a $u$-flip such that the \emph{next}
      $u$-flip-component found has the same size; or
    \item finds no $u$-flip.
    \end{itemize}
  \end{itemize}
  Then $\Call{do-separation-flips}{a_i,a_{i+1}}$ does only critical
  separation flips in worst case $\OO(\log^2 n)$ time per flip, and:
  \begin{itemize}
  \item If $G\cup(u,v)$ is planar, every flip is potential-decreasing.
  \item If $G\cup(u,v)$ is nonplanar every flip except the last is
    potential-decreasing.
  \end{itemize}
\end{lemma}
\begin{proof}
  The running time is worst case $\OO(\log^2 n)$ per flip, by our
  assumption, and because the $\Call{linkable}{u,v}$ query and
  actually executing the flip takes $\OO(\log^2 n)$ in the underlying
  data structure from~\cite{DBLP:journals/mst/HolmR17}.
  By definition, any $u$-flip is a critical flip, so all the flips
  performed are critical.
  By assumption, each tuple of corners $\sigma$ that we consider
  (except the last) bound a maximal $u$-flip-component. If the flip
  described by $\sigma$ is a P flip, there may be two possible choices
  if the node containing $u$ on the $u,v$-critical path in the SPQR
  tree is an S node. However, either choice is potential-decreasing,
  as it brings the two neighbors to the involved P node that are on
  the critical path together.  If the flip described by $\sigma$ is
  not a P flip, it is potential-decreasing unless the first node $X$
  that is not included in $H_u$ on the $u,v$-critical path in the SPQR
  tree is an $R$ node that is \emph{cross} (i.e.\@ the virtual edges
  in $\Gamma(X)$ corresponding to the path do not share a face). If
  $X$ is cross, $G\cup(u,v)$ is nonplanar, and the flip will be the
  last we execute, since no larger $H_u$ exists after the flip.
\end{proof}

Let $B=(V_B,E_B)$ be the biconnected component of $G$ that contains
$u,v$, and suppose $u,v$ do not share a face in $G$ (otherwise $u$ and
$v$ would be linkable and $\Call{find-first-separation-flip}{u,v}$
would not be called).  Consider the $u,v$-critical path $X_1\cdots
X_k$ with $u\in X_1,v\in X_k$ in the SPQR tree for $B$.  Observe that
if $k=1$ our assumption that $u,v$ do not share a face means that
$X_1$ is an R node, and no flip exists that can make $u$ and $v$
linkable.  Assume therefore that $k>1$. Then $X_1$ and $X_k$ are
distinct, and by definition of $u,v$-critical path none of them are P
nodes.

To aid in our discussion, we need to name certain subsets of the edges
of $E_B$ based on their relationship with the SPQR nodes on the
$u,v$-critical path $X_1\cdots X_k$ in the SPQR tree for $B$.
\begin{definition}\label{def:E-partition}
  Let $e_u$ and $e_v$ be the edges incident to $u$ and $v$ on the primal
  spanning tree path $u\cdots v$. For each $i$ with $1\leq i\leq k$
  define:
  \begin{align*}
    E_{<i} &:=
    \begin{cases}
      \emptyset
      & \text{if $i=1$}
      \\
      \parbox[t]{9em}{the separation class of $e_u$ w.r.t. $X_{i-1}\cap X_i$}
      & \text{otherwise}
    \end{cases}
    &
    E_{\geq i} &:= E_B\setminus E_{<i}
    \\
    E_{>i} &:=
    \begin{cases}
      \emptyset & \text{if $i=k$}
      \\
      \parbox[t]{9em}{the separation class of $e_v$ w.r.t. $X_i\cap X_{i+1}$}
      & \text{otherwise}
    \end{cases}
    &
    E_{\leq i} &:= E_B\setminus E_{>i}
    \\
    E_i &:= E_{\leq i}\cap E_{\geq i}%
    &
    E_{\neq i} &:= E_B\setminus E_i
  \end{align*}
  Now $E_1,\ldots,E_k$ is a partition of $E_B$, so for each $e\in E_B$
  there is a unique \emph{index} $i=\idx(e)$ such that $e\in E_i$.
  Furthermore, for each $1\leq i\leq k$ the set $E_i$ is associated with
  the node $X_i$
\end{definition}

Since $u$ and $v$ are biconnected, there exists two internally
vertex-disjoint paths between $u$ and $v$. Let $p^s$, $p^t$ be an
arbitrary pair of internally vertex-disjoint paths from $u$ to $v$.
We will use $p^s$ and $p^t$ to define some further concepts, and then
argue (e.g. in Lemmas~\ref{lem:define-r-firstcross}
and~\ref{lem:fu-blocking-welldefined}) that these definitions do
not depend on the particular choice of $p^s,p^t$.

For $1\leq i\leq k-1$ let $\set{s_i,t_i}=X_i\cap X_{i+1}$ such that
$s_i\in p^s$ and $t_i\in p^t$, and let $s_0=t_0=u$ and $s_k=t_k=v$.
Let $f_u$ be any face maximizing the maximum $i$ such that $f_u$
contains all of $u$, $s_i$, and $t_i$. Note that the candidates for
$f_u$ do not depend on the specific choice of $p^s$ and $p^t$, but
only on the structure of the SPQR tree and the current embedding.
Together, $p^s\cup p^t$ form a simple cycle in $G$ which we call a
\emph{$u,v$-critical cycle}.  This cycle partitions the plane into two
regions.  Call the region containing $f_u$ the \emph{$f_u$-side}
region and the other the \emph{anti-$f_u$-side} region.

For each node $X_i$ on the $u,v$-critical path in the SPQR tree, any
$u,v$-critical cycle corresponds to a unique cycle in $\Gamma(X_i)$,
and the partition into $f_u$-side and anti-$f_u$-side regions carry
over into $\Gamma(X_i)$.
\begin{itemize}
\item If $X_i$ is a P node, we say that $X_i$ is \emph{$f_u$-blocking}
  if the $f_u$-side region of $\Gamma(X_i)$ contains any edges, and
  \emph{anti-$f_u$-blocking} if the anti-$f_u$-side region contains
  any edges.
\item If $X_i$ is an R node, we say that $X_i$ is
  \emph{$f_u$-blocking} (resp.\@ \emph{anti-$f_u$-blocking}) if the
  $f_u$-side (resp. anti-$f_u$-side) region of $\Gamma(X_i)$ does not
  contain a face incident to all of $s_{i-1},t_{i-1},s_i,t_i$. (Note
  that this holds even for $i=1$ and $i=k$).
\item If $X_i$ is an S node it is neither $f_u$-blocking nor
  anti-$f_u$-blocking.
\end{itemize}
A node that is both $f_u$-blocking and anti-$f_u$-blocking is simply
called \emph{blocking}. A blocking R node is also called a
\emph{cross} node.

\begin{lemma}\label{lem:define-r-firstcross}
  If $G\cup(u,v)$ is planar let $r=k$, otherwise let $r$ be the
  minimum index such that $X_r$ is a cross node.  Then $r$ is well-defined
  and depends only on $G$ and the vertices $u,v$.
\end{lemma}
\begin{proof}
  This is trivial if $G\cup(u,v)$ is planar, so suppose not. Then
  $G\cup(u,v)$ contains either a $K_5$ subdivision containing $(u,v)$
  or a $K_{3,3}$ subdivision containing $(u,v)$. In particular, the
  SPQR tree for $G\cup(u,v)$ contains an $R$ node whose skeleton graph
  contains such a subdivision.  When deleting $(u,v)$ from
  $G\cup(u,v)$, this $R$ node splits into the $u,v$-critical path
  $X_1\cdots X_k$ in $G$. This path contains at least one $R$ node
  containing a $K_4$ subdivision, and since $G\cup(u,v)$ is nonplanar,
  for at least one such R node $X_i$ the vertices
  $\set{s_{i-1},t_{i-1},s_i,t_i}=(X_{i-1}\cap X_i)\cup(X_i\cap
  X_{i+1})$ do not all share a face in $\Gamma(X_i)$.
\end{proof}

\begin{lemma}\label{lem:fu-blocking-welldefined}
  Let $r$ be defined as in Lemma~\ref{lem:define-r-firstcross}.  If
  $u$ and $v$ do not share a face, then for $1\leq i\leq r$ whether
  $X_i$ is (anti-)$f_u$-blocking depends only on the choice of $f_u$
  and the current embedding, and is independent of the particular
  choice of $u,v$-critical cycle.
\end{lemma}
\begin{proof}
  \todo{Do we even need to argue that the $s_i$ and $t_i$ are
    unique?}For any S or P node $X_i$ with $2\leq i\leq k$, the
  vertices $s_i$ and $t_i$ are completely determined by $s_{i-1}$ and
  $t_{i-1}$, and do not depend on the particular choice of
  $p^s,p^t$. The same is true for every R node $X_i$ with $2\leq i\leq
  r-1$, since by definition of $r$ these are not both $f_u$- and
  anti-$f_u$-blocking. Thus $s_2,\ldots,s_{r-1}$ and
  $t_2,\ldots,t_{r-1}$ are completely determined by $s_1,t_1$. There
  are only two possible ways to map $\set{s_1,t_1}$ to the two
  vertices in $X_1\cap X_2$, so $s_1,\ldots,s_{r-1}$ and
  $t_1,\ldots,t_{r-1}$ are uniquely defined up to an exchange of every
  $s$ and $t$ with its opposite.
  Now consider any two $u,v$-critical cycles $C_1,C_2$, and a node
  $X_i$ with $1\leq i\leq r$.
  \begin{itemize}
  \item If $X_i$ is an S node, it is by definition neiher $f_u$- nor
    anti-$f_u$-blocking, and this does not depend on the choice of
    critical cycle.

  \item If $X_i$ is a $P$ node, then $1<i<r$ and the cycle in
    $\Gamma(X_i)$ does not depend on which critical cycle is used,
    only on which neighbors $X_i$ has on the $u,v$-critical path in
    the SPQR tree. Thus the definition of $f_u$-side and
    anti-$f_u$-side region, and hence $X_i$s status as $f_u$- or
    anti-$f_u$-blocking does not depend on the choice of critical
    cycle, but only on the choice of $f_u$ and the current embedding.

  \item If $X_i$ is an $R$ node, and is not cross there is a unique
    face of $\Gamma(X_i)$ that contains all of
    $s_{i-1},s_i,t_{i-1},t_i$, and this face \todo{why?}is in the
    $f_u$-side region of $C_1$ if and only if it is in the $f_u$-side
    region of $C_2$. Thus, $X_i$ is $f_u$- or anti-$f_u$-blocking with
    respect to $C_1$ if and only it is with respect to $C_2$.

  \item If $X_i$ is a cross node, this does not depend on the choice
    of $u,v$-critical cycle but only on $G,u,v$. By definition $X_i$
    is both $f_u$- and anti-$f_u$-blocking in both $C_1$ and
    $C_2$.\qedhere
  \end{itemize}
\end{proof}

\begin{observation}\label{obs:define-b-firstblock}
  Suppose $u$ and $v$ do not share a face, let $f_u$ be given, and let
  $r$ be defined as in Lemma~\ref{lem:define-r-firstcross}.  Let
  $b\in\set{1,\ldots,r}$ be the minimum index such that $X_b$ is
  $f_u$-blocking.  Then $b$ depends only on $u,v$, the current
  embedding, and the choice of $f_u$, and is independent of the
  particular choice of $u,v$-critical cycle.
\end{observation}

Since we assume $k>1$, if $b=1$ then by definition $X_1$ is cross and no $u$-flip exists. For the remainder of this section, we will therefore assume that $b>1$ (and thus $1<b\leq r\leq k$).

In each step of our algorithm, to maximize the size of $H_u$ we must
find the face $f_u$ and the separation pair
$\set{s_{b-1},t_{b-1}}$. In addition, we must find a second face
$f_v$, such that $H_u$ corresponds to the region on one side of the
$4$-cycle $f_u,s_{b-1},f_v,t_{b-1}$ in the vertex-face graph.

\FloatBarrier
\subsubsection{\textproc{choose-best-flip}}

The first thing we observe is that if we can somehow guess (or
compute) the faces $f_u$ and $f_v$ bounding (a candidate to) the
maximal $u$-flip-component, then the function
$\Call{choose-best-flip}{u,v,f_u,f_v}$ from
Algorithm~\ref{alg:choose-best-flip} can compute the size and corners
of the proposed flip.  In addition, if a given pair of faces are
``obviously not'' bounding the maximal $u$-flip-component (either
because no $u$-flip-component is bounded by these faces, or because a
larger component can be easily found), this function can detect and
report it. \todo{ER: obviously not ... do you mean not bounding a flip at
  all?}\todo{JH: Either that, or not locally maximal as defined next.}
\begin{definition}\label{def:locally-maximal}
  Define a $u$-flip-component $H_u$ to be \emph{locally maximal} if,
  given the corners $\sigma$, and the separation pair $\set{s_j,t_j}$,
  and faces $f_u',f_v'$ incident to $\sigma$:
  \begin{itemize}
  \item No other $u$-flip-component bounded by $f_u',f_v'$ is larger
    than $H_u$; and
  \item No other $u$-flip-component incident to $s_j,t_j$ and $f_u'$
    is larger than $H_u$.
  \end{itemize}
\end{definition}
Any maximal $u$-flip-component is also locally maximal, so any pair of faces that do \emph{not} bound a locally maximal $u$-flip-component can be rejected.
\begin{lemma}\label{lem:locally-maximal}
  Given vertices $u,v$ that do not share a face and a $4$-tuple of
  corners $\sigma$ we can in worst case $\OO(\log^2 n)$ time determine
  if $\sigma$ bounds a locally maximal $u$-flip-component.
\end{lemma}
\begin{proof}
  Let $f_u,f_v$ be the faces, and $\set{s_x,s_y}$ the separation pair
  incident to $\sigma$.  If $u\not\in f_u$ or
  $\set{u,v}\cap\set{s_x,s_y}\neq\emptyset$, $\sigma$ does not bound a
  $u$-flip-component.

  Let $H_u\ni u,H_v$ be the subgraphs of $H$ on the two sides of
  $\sigma$.  If $v\not\in H_v$ then $\sigma$ does not bound a
  $u$-flip-component, otherwise $H_u$ is a $u$-flip-component and we
  must determine if it is locally maximal.

  Let $C$ be the fundamental cycle in $H$ closed by the first primal
  edge on the dual path $f_u\cdots f_v$. Now $C\cap H_v$ is a path
  from $s_x$ to $s_y$. If any internal node on this path is incident
  to both $f_u$ and $f_v$, then $H_u$ is not locally maximal.

  Let $\dual{e}$ be the dual of the first edge on the primal tree path
  $s_x\cdots s_y$, and let $\dual{C}$ be the fundamental cycle closed
  by $\dual{e}$ in the dual tree.  Then $\sigma$ cuts $\dual{C}$ into
  an $f_u\cdots f_v$ path through each of $\dual{H_u}$ and
  $\dual{H_v}$. Let $f_v'$ be the first face on the path from $f_v$ to
  $f_u$ in $\dual{C}\cap \dual{H_v}$ that is incident to both $s_x$
  and $s_y$, and let $\sigma'$ be any $4$ corners between
  $(f_v,s_x),(f_v,s_y),(f_v',s_x),(f_v',s_y)$.  Let $H_u'\ni u,H_v'$
  be the subgraphs of $H$ on the two sides of $\sigma'$. If $v\not\in
  H_v'$ then $H_u$ is not locally maximal.

  Otherwise $H_u$ is locally maximal.  Each step of this test can be
  done in worst case $\OO(\log^2 n)$ time using the data structure
  from~\cite{DBLP:journals/mst/HolmR17}.
\end{proof}

\begin{algorithm}[htb!]
  \caption{}
  \label{alg:choose-best-flip}
  \begin{algorithmic}[1]

    \Function{choose-best-flip}{$u,v,f_u,f_v$}

    \LeftComment{Determine maximum $u$-flip in $f_u,f_v$}

    \If{$u\in f_u$ and $u\not\in f_v$}

      \State $e\gets$ the first primal edge on the dual tree path from
      $f_u$ to $f_v$.

      \State $C\gets$ fundamental cycle of $e$

      \If{$\exists w\in C$ incident to both $f_u$ and $f_v$}

        \State $p_v\gets\pi_C(v)$; $p_x,p_y\gets$ neighbors of $p_v$ on $C$

        \For{$(x,y)\in\set{(p_x,p_v),(p_v,p_y)}$}

          \State $s_x,s_y\gets$ vertices on $C$ incident to $f_u$ and
          $f_v$ closest to $(x,y)$ in each direction

          \State $c_x^u,c_y^u\gets$ corners at $s_x,s_y$ nearest to
          $(x,y)$ on the $f_u$ side

          \State $c_x^v,c_y^v\gets$ corners at $s_x,s_y$ nearest to
          $(x,y)$ on the $f_v$ side

          \State $\sigma\gets(c_x^u,c_y^u,c_y^v,c_x^v)$

          \If{$\sigma$ bounds a locally maximal
            $u$-flip-component}\Comment{$\sigma$ is a valid
            $u,v$-flip}

            \State $s\gets$ size of $u$-flip-component of $\sigma$

            \State \Return $(s,\sigma)$

          \EndIf

        \EndFor

      \EndIf

    \EndIf

    \State \Return $(0,\bot)$

    \EndFunction

  \end{algorithmic}
\end{algorithm}

\begin{lemma}\label{lem:choose-best-flip}\sloppy
  Let $f_u,f_v$ be faces. If $f_u,f_v$ are incident to the corners
  $\sigma$ bounding a locally maximal $u$-flip-component $H_u$, then
  $\Call{choose-best-flip}{u,v,f_u,f_v}$ in
  Algorithm~\ref{alg:choose-best-flip} returns a tuple $(s,\sigma)$
  where $s$ is the size of $H_u$. Otherwise $(0,\bot)$ is returned.
  In either case, the total time is worst case $\OO(\log^2
  n)$.\todo{ER: måske standardisere vores brug af "sigmaen har fire
    hjørner".}  \todo{JH: Jeg har prøvet, se starten af
    sektion~\ref{subsec:do-sep}}
\end{lemma}
\begin{proof}
  If $f_u,f_v$ are valid then by definition $u\in f_u$ and $u\not\in
  f_v$ (and thus $f_u\neq f_v$).  Let $e$ be any primal edge on the
  dual tree path from $f_u$ to $f_v$, and consider the fundamental
  cycle $C$ closed by $e$.  By construction, $C$ has $f_u$ and $f_v$
  on opposite sides, so any common vertex of $f_u$ and $f_v$ must be
  on $C$. In particular, we must have $s_i,t_i\in C$, and if no vertex
  in $C$ is incident to both $f_u$ and $f_v$ then $f_u$ and $f_v$ do
  not bound a valid $u$-flip-component.  We want the flip-component
  containing $u$ to be as large as possible, which means the
  flip-component containing $v$ must be as small as possible.  Let
  $p_v=\pi_C(v)$ and let $p_x,p_y$ be the neighbors of $p_v$ on
  $C$. Then at least one of the edges $\set{(p_x,p_v),(p_v,p_y)}$ is
  in the flip-component of $v$.  Suppose $(x,y)$ is this edge, then
  the separation pair $\set{s_x,s_y}$ we want will consist of the first
  vertex incident to both $f_u$ and $f_v$ in each direction on $C$
  away from $(x,y)$.  If $G$ is not biconnected, there may be multiple
  corners between each of $s_x,s_y$ and each of $f_u,f_v$.  To truly
  minimize the size of the flip-component of $v$, we must choose the
  nearest $4$ corners to $(x,y)$.
  This gives us $4$ corners $(c_x^u,c_y^u,c_y^v,c_x^v)$, and by
  Lemma~\ref{lem:locally-maximal} we can test if these bound a locally
  maximal $u$-flip-component and compute its size in worst case
  $\OO(\log^2 n)$ time. Note that if both candidates for $(x,y)$ are
  valid, they will give the same answer, so it is ok to stop as soon
  as we find one that is valid.
  If none of the two candidates for $(x,y)$ give a valid flip, then
  $f_u,f_v$ did not bound a locally maximal $u$-flip-component, and we
  return $(0,\bot)$.
\end{proof}

\FloatBarrier
\subsubsection{\textproc{find-single-flip-candidates}}

Our task is thus to compute (candidates for) $f_u$ and $f_v$.
In~\cite{DBLP:journals/mst/HolmR17} we considered the special case
where $G\cup(u,v)$ is planar and a single flip is needed to admit
$(u,v)$.  The function \Call{find-single-flip-candidates}{} in
Algorithm~\ref{alg:find-single} is essentially the same as is
described in~\cite{DBLP:journals/mst/HolmR17}, but instrumented to
return a bit more information which will become important later.

\begin{algorithm}[htb!]
  \caption{}
  \label{alg:find-single}
  \begin{algorithmic}[1]

    \Function{find-single-flip-candidates}{$u,v$}

    \State $e_u,e_v\gets$ the first and last edge on the primal tree
    path from $u$ to $v$.

    \State $u_L,u_R\gets$ left and right face indident to the first
    edge on the primal tree path from $u$ to $v$.

    \State $v_L,v_R\gets$ left and right face indident to the first
    edge on the primal tree path from $v$ to $u$.

    \State $f_u^L\gets\meet(u_L,u_R,v_L)$; $f_u^R\gets\meet(u_L,u_R,v_R)$

    \State $f_v^L\gets\meet(v_L,v_R,u_L)$; $f_v^R\gets\meet(v_L,v_R,u_R)$

    \State result $\gets\emptyset$

    \If{$f_u^L\neq f_u^R$}

      \LeftComment{Handles case of a fundamental cycle through $u$ and
        $v$.}\label{line:single-11-start}

      \State $e\gets$ primal edge corresponding to first edge on dual
      tree path $f_u^L\cdots f_u^R$

      \State $C\gets$ fundamental cycle closed by $e$

      \State $e_u^1,e_u^2,e_v^1,e_v^2\gets$ the edges incident to $u,v$ on $C$

      \For{$f\in\set{u_L,u_R}$}

        \State $f_u^1,f_u^2,f_v^1,f_v^2\gets$ faces incident to
        $e_u^1,e_u^2,e_v^1,e_v^2$ on the same side of $C$ as $f$.

        \State $\bar{f_u^1},\bar{f_u^2},\bar{f_v^1},\bar{f_v^2}\gets$
        faces incident to $e_u^1,e_u^2,e_v^1,e_v^2$ on the opposite side
        of $C$ from $f$.

        \State result $\gets$ result $\cup\set{ (
          \meet(f_u^1,f_u^2,f_v^1),
          \meet(\bar{f_v^1},\bar{f_v^2},\bar{f_u^1}),
          C,
          e_u,
          e_v
          ) }$

      \EndFor\label{line:single-11-end}

    \Else{} $f_u^L=f_u^R$ and $f_v^L=f_v^R$

       \State $f_u^\star\gets f_u^L$; $f_v^\star\gets f_v^L$

       \LeftComment{Handles case of a fundamental separating cycle
         through $u$ but none through $u,v$.}\label{line:single-10-start}

       \For{$f\in\set{u_L,u_R}\setminus\set{f_u^\star}$}

         \State $e\gets$ first primal edge on dual tree path from
         $f_u^\star$ to $f$

         \State $C\gets$ fundamental cycle closed by $e$

         \State $e_u^1,e_u^2,e_v^1,e_v^2\gets$ the edges incident to
         $u,\pi_C(v)$ on $C$

         \For{$e_v'\in\set{e_v^1,e_v^2}$}

           \For{$f_v'\in$ faces incident to $e_v'$}

             \State $f_u^1,f_u^2\gets$ faces incident to $e_u^1,e_u^2$
             on the same side of $C$ as $f_v'$.

             \State result $\gets$ result $\cup\set{ (
               \meet(f_u^1,f_u^2,f_v'),
               f_v^\star,
               C,
               e_u,
               e_v'
               ) }$

           \EndFor

         \EndFor

       \EndFor\label{line:single-10-end}

       \LeftComment{Handles case of a fundamental separating cycle
         through $v$ but none through $u,v$.}\label{line:single-01-start}

       \For{$f\in\set{v_L,v_R}\setminus\set{f_v^\star}$}

         \State $e\gets$ first primal edge on dual tree path from
         $f_v^\star$ to $f$

         \State $C\gets$ fundamental cycle closed by $e$

         \State $e_u^1,e_u^2,e_v^1,e_v^2\gets$ the edges incident to
         $\pi_C(u),v$ on $C$

         \For{$e_u'\in\set{e_u^1,e_u^2}$}

           \For{$f_u'\in$ faces incident to $e_u'$}

             \State $f_v^1,f_v^2\gets$ faces incident to $e_v^1,e_v^2$
             on the same side of $C$ as $f_u'$.

             \State result $\gets$ result $\cup\set{ (
               f_u^\star,
               \meet(f_v^1,f_v^2,f_u'),
               C,
               e_u',
               e_v
               ) }$

           \EndFor

         \EndFor

       \EndFor\label{line:single-01-end}

       \LeftComment{Handles case of no fundamental separating cycle
         through $u$ or $v$.}\label{line:single-00-start}

       \If{$f_u^\star\neq f_v^\star$}

         \State $e\gets$ first primal edge on dual tree path from
         $f_u^\star$ to $f_v^\star$

         \State $C\gets$ fundamental cycle closed by $e$

         \State $e_u^1,e_u^2,e_v^1,e_v^2\gets$ the edges incident to
         $\pi_C(u),\pi_C(v)$ on $C$

         \For{$e_u'\in\set{e_u^1,e_u^2}$}

           \For{$e_v'\in\set{e_v^1,e_v^2}$}

             \State result $\gets$ result $\cup\set{ (
               f_u^\star,
               f_v^\star,
               C,
               e_u',
               e_v'
               ) }$

           \EndFor

         \EndFor\label{line:single-00-end}

      \EndIf

    \EndIf

    \State \Return result

    \EndFunction

  \end{algorithmic}
\end{algorithm}

\begin{definition}\label{def:goodcycle}
  A fundamental cycle $C$ is \emph{separating} if $f_u$ and $f_v$ are
  on opposite sides of $C$.  If $X_b$ is the first $f_u$-blocking node
  on $X_1\ldots X_r$ then a fundamental cycle $C$ is \emph{good} when
  either $C$ is separating or $C$ intersects both $E_{\leq b}$ and
  $E_{>b}$.
\end{definition}

\begin{definition}\label{def:candidate-tuple}\sloppy
  A tuple $(f_u',f_v',C,e_u',e_v')$ is a \emph{candidate tuple} if
  $f_u'$ and $f_v'$ are faces, $C$ is a fundamental cycle, and
  $e_u',e_v'\in C$.
  A candidate tuple is \emph{good} if $f_u' = f_u$, $C$ is good as
  defined in Definition~\ref{def:goodcycle}, $\idx(e_u')=\min_{e\in
    C}\idx(e)$, and $\idx(e_v')=\max_{e\in C}\idx(e)$.
  A good candidate tuple is \emph{correct} if furthermore $f_v' = f_v$.
\end{definition}

\begin{lemma}\label{lem:find-single-time}
  $\Call{find-single-flip-candidates}{u,v}$ from
  Algorithm~\ref{alg:find-single} runs in worst case $\OO(\log^2 n)$
  time and returns a set of at most $20$ candidate tuples.
\end{lemma}
\begin{proof}
  The result set starts empty and changes only by addition of
  elements, and simple counting shows that either exactly $2$, or at
  most $2\cdot2\cdot2+2\cdot2\cdot2+2\cdot2=20$ elements are added,
  and by construction each element added is trivially a candidate
  tuple.

  For each element added, we use a constant number of elementary
  operations each taking worst case constant time, and a constant
  number of the following operations:
  \begin{itemize}
  \item Find the first or last edge on a path in the primal or dual tree.
  \item Find the faces incident to an edge in the primal tree.
  \item Compute the $\meet$ of $3$ faces in the dual tree.
  \item Determine if two faces are the same.
  \item Find the primal edge corresponding to an edge in the dual tree.
  \item Given a fundamental cycle $C$ and a vertex $w\in C$, find the
    edges on $C$ incident to $w$.
  \item Given a fundamental cycle $C$ and faces $f_1,f_2$ determine
    if they are on the same side of $C$.
  \item Given a fundamental cycle $C$ and a vertex $w\not\in C$, find the
    projection $\pi_C(w)\in C$.
  \end{itemize}
  Using the data structure from~\cite{DBLP:journals/mst/HolmR17}, each
  of these operations can be done in worst case $\OO(\log n)$ time.
\end{proof}

\begin{lemma}\label{lem:find-single-findsflip}
  If $G\cup(u,v)$ is planar and only one flip is needed at $f_u,f_v$
  to admit $(u,v)$, then at least one
  $(f_u',f_v',\cdot,\cdot,\cdot)\in\Call{find-single-flip-candidates}{u,v}$
  from Algorithm~\ref{alg:find-single} is correct.
\end{lemma}
\begin{proof}
  We consider the $4$ cases for how the fundamental separating cycles
  may relate to $u$ and $v$, and show that in each case we find at
  least one correct candidate tuple.
  \begin{description}\sloppy
  \item[If some fundamental separating cycle contains both $u$ and
    $v$,] then any such cycle $C$ separates $u_L$ from $u_R$ and $v_L$
    from $v_R$.  In particular a primal nontree edge $e$ closes such a
    cycle if and only if it is on both $u_L\cdots u_R$ and $v_L\cdots
    v_R$.  But then we also have $f_u^L\neq f_u^R$ and $e\in
    f_u^L\cdots f_u^R$. We therefore execute
    Line~\ref{line:single-11-start}--\ref{line:single-11-end} exactly
    when such a separating cycle exists, and $C$ is such a separating
    cycle.
    By definition, $e_u,e_u^1,e_u^2\in E_1$ and $e_v,e_v^1,e_v^2\in
    E_k$, and the cycle $C$ separates $u_L$ from $u_R$, so there
    exists an $f\in\set{u_L,u_R}$ that is on the same side of $C$ as
    $f_u$.  For this value of $f$, we must have $f_u\in f_u^1\cdots
    f_u^2$, $f_u\in f_u^1\cdots f_v^1$, and $f_u\in f_u^2\cdots
    f_v^1$, and thus $f_u=\meet(f_u^1,f_u^2,f_v^1)$ and
    $(\meet(f_u^1,f_u^2,f_v^1),\cdot,C,e_u,e_v)$ is a good candidate
    tuple.
    Similarly, if only one flip is needed to admit $(u,v)$ then $f_v$
    is on the opposite side of $C$ from $f$, and we must have that
    $f_v\in \bar{f_v^1}\cdots \bar{f_v^2}$, $f_v\in
    \bar{f_v^1}\cdots \bar{f_u^1}$, and $f_v\in \bar{f_v^2}\cdots
    \bar{f_u^1}$, and thus
    $f_v=\meet(\bar{f_v^1},\bar{f_v^2},\bar{f_u^1})$ and
    $(\meet(f_u^1,f_u^2,f_v^1), \meet(\bar{f_v^1},\bar{f_v^2},
    \bar{f_u^1}),C,e_u,e_v)$ is a correct candidate tuple.

  \item[If no fundamental separating cycle contains both $u$ and $v$,
    but some contains $u$,] then $f_u^L=f_u^R$ and $f_v^L=f_v^R$ and
    Line~\ref{line:single-10-start}--\ref{line:single-10-end} gets
    executed. Any such cycle must separate either $u_L$ or $u_R$ from
    the rest of $\set{u_L,u_R,v_L,v_R}$, and in particular it must
    separate $u_L$ or $u_R$ from $f_u^\star$.  In this case there must
    exist an $f\in\set{u_L,u_R}\setminus\set{f_u^\star}$ that is on
    the opposite side of such a cycle from $f_u^\star$, and in fact
    the first primal edge on the dual tree path from $f_u^\star$ to
    $f$ closes such a separating cycle $C$.
    By definition, $e_u,e_u^1,e_u^2\in E_1$ and there exists
    $e_v'\in\set{e_v^1,e_v^2}$ with $\idx(e_v')=\max_{e\in C}\idx(e)$.
    Exactly one face $f_v'$ incident to $e_v'$ is on the same side of
    $C$ as $f_u$.  Given this face $f_v'$ and the corresponding faces
    $f_u^1,f_u^2$ incident to $e_u^1,e_u^2$ on the same side of $C$ as
    $f_v'$, we have $f_u\in f_u^1\cdots f_u^2$, $f_u\in f_u^1\cdots
    f_v'$, and $f_u\in f_u^2\cdots f_v'$, and thus
    $f_u=\meet(f_u^1,f_u^2,f_v')$ and
    $(\meet(f_u^1,f_u^2,f_v'),\cdot,C,e_u,e_v')$ is a good candidate
    tuple.
    If only one flip is needed to admit $(u,v)$ then $f_v$ is on all
    of $v_L\cdots v_R$, $v_L\cdots f_u^\star$, $v_R\cdots f_u^\star$,
    and thus $f_v=f_v^\star$ and
    $(\meet(f_u^1,f_u^2,f_v'),f_v^\star,C,e_u,e_v')$ is a correct
    candidate tuple.

  \item[If no fundamental separating cycle contains both $u$ and $v$,
    but some contains $v$,] then $f_u^L=f_u^R$ and $f_v^L=f_v^R$ and
    Line~\ref{line:single-01-start}--\ref{line:single-01-end} gets
    executed. Any such cycle must separate either $v_L$ or $v_R$ from
    the rest of $\set{u_L,u_R,v_L,v_R}$, and in particular it must
    separate $v_L$ or $v_R$ from $f_v^\star$.  In this case there must
    exist an $f\in\set{v_L,v_R}\setminus\set{f_v^\star}$ that is on
    the opposite side of such a cycle from $f_v^\star$, and in fact
    the first primal edge on the dual tree path from $f_v^\star$ to
    $f$ closes such a separating cycle $C$.
    By definition, $e_v,e_v^1,e_v^2\in E_k$ and there exists
    $e_u'\in\set{e_u^1,e_u^2}$ with $\idx(e_u')=\min_{e\in C}\idx(e)$.
    In this case, $f_u$ is on all of $u_L\cdots u_R$, $u_L\cdots
    f_v^\star$, $u_R\cdots f_v^\star$, and thus $f_u=f_u^\star$ and
    $(f_u^\star,\cdot,C,e_u',e_v)$ is a good candidate tuple.
    Exactly one face $f_u'$ incident to $e_u'$ is on the same side of
    $C$ as $f_v$.  Given this face $f_u'$ and the corresponding faces
    $f_v^1,f_v^2$ incident to $e_v^1,e_v^2$ on the same side of $C$ as
    $f_u'$, we have $f_v\in f_v^1\cdots f_v^2$,
    $f_v\in f_v^1\cdots f_u'$, and $f_v\in f_v^2\cdots f_u'$, and
    thus $f_v=\meet(f_v^1,f_v^2,f_u')$ and
    $(f_u^\star,\meet(f_v^1,f_v^2,f_u'),C,e_u',e_v)$ is a correct candidate
    tuple.

  \item[If no fundamental separating cycle contains $u$ or $v$,] then
    $f_u^L=f_u^R$ and $f_v^L=f_v^R$ and
    Line~\ref{line:single-00-start}--\ref{line:single-00-end} gets
    executed. Assuming only one flip is needed to admit $(u,v)$, then
    $f_u$ and $f_v$ are simply the $f_u^\star$ and $f_v^\star$ found
    in the algorithm. Specifically, since $u$ and $v$ are not already
    linkable, $f_u^\star\neq f_v^\star$.  Since there exists
    $e_u'\in\set{e_u^1,e_u^2}$ with $\idx(e_u')=\min_{e\in C}\idx(e)$
    and $e_v'\in\set{e_v^1,e_v^2}$ with $\idx(e_v')=\max_{e\in
      C}\idx(e)$, at least one of the candidates is correct.
  \end{description}
\end{proof}

\begin{lemma}\label{lem:find-single-findsgood}
  If $X_{b}$ is the first $f_u$-blocking node on $X_1\ldots X_r$ and
  $1<b<r$, then at least one of the candidates
  $(f_u',\cdot,C,e_u',e_v')\in\Call{find-single-flip-candidates}{u,v}$
  from Algorithm~\ref{alg:find-single} is good.
\end{lemma}
\begin{proof}
  Since, for every cycle $C$, we try all relevant edges near
  $\pi_C(u)$ and $\pi_C(v)$, we need only argue that $f_u'$ is $f_u$
  and $C$ is good; if one of the tuples contains a good cycle and the
  correct $f_u$, then one of the tuples will be good.

  If there exists a fundamental separating cycle through at least one
  of $u,v$ then (by the same arguments as in
  Lemma~\ref{lem:find-single-findsflip}) for some such cycle $C$,
  $(f_u,\cdot,C,e_u',e_v')$ is among the returned candidates.  Since
  $C$ is separating, it is by definition good, and thus this tuple is
  clearly good.

  Otherwise no fundamental separating cycle contains any of $u$ or
  $v$. Let $f_u^\star =f_u^L=f_u^R$ and $f_v^\star
  =f_v^L=f_v^R$. Then, because $1<b<r$, it must be the case that
  $f_u^\star\neq f_v^\star$. Thus, let $e$ be the first primal edge on
  the dual tree path from $f_u^\star$ to $f_v^\star$.
  Note that if $f_u^\star\not\in f_v\cdots f_v^\star$, then that edge
  $e$ closes a fundamental cycle that is separating and therefore
  good. Suppose therefore that $f_u^\star\in f_v\cdots
  f_v^\star$. Then $e\in E_{b}$. Furthermore the path $f_v\cdots
  f_v^\star$ cuts $E_{\leq b}$ such that not all of
  $s_{b-1},t_{b-1},s_{b}$, and $t_{b}$ are connected in $E_{\leq
    b}\setminus (f_v\cdots f_v^\star)$.  Thus, the fundamental cycle
  closed by any primal edge on $f_v\cdots f_v^\star$ must go through
  $E_{>b}$.  In particular the edge $e$ closes a cycle $C$
  intersecting $E_{>b}$. Thus $C$ is good, since it intersects both
  $E_{b}$ and $E_{>b}$-
\end{proof}

\FloatBarrier
\subsubsection{\textproc{find-first-separation-flip}}

\begin{algorithm}[htb!]
  \caption{}
  \label{alg:find-first-sep}
  \begin{algorithmic}[1]

    \Function{find-first-separation-flip}{$u,v$}

      \LeftComment{Assumes $u,v$ are biconnected and do not share a
      face.}

      \State candidates $\gets\Call{find-single-flip-candidates}{u,v}$

      \LeftComment{Test for single flip}

      \For{$(f_u,f_v,\cdot,\cdot,\cdot)\in$ candidates}

        \If{$v\in f_v$}

          \State $(s,\sigma)\gets\Call{choose-best-flip}{u,v,f_u,f_v}$

          \If{$s>0$}

            \State \Return $(s,\sigma)$

          \EndIf

        \EndIf

      \EndFor

      \LeftComment{Not single flip. If first flip exists, then first
      $f_u$-blocking node $X_{b}$ has $1<b<r$.}

      \State result $\gets(0,\bot)$

      \For{$(f_u,\cdot,C,e_u',e_v')\in$ candidates}

        \LeftComment{Assume $f_u$ is correct, $C$ is good, and
        $e_u',e_v'\in C$ are in the first/last $E_j$ with
        $E_j\cap C\neq\emptyset$.}\label{line:find-first-sep:assume-good}

        \Statex

        \LeftComment{Handle cases where $E_{>b}\cap C\neq\emptyset$.}

        \State $f_v'\gets$ face
        incident to $e_v'$ on same side of $C$ as $f_u$.

        \If{$f_u\neq f_v'$}\Comment{If $f_u=f_v'$ we are not in this
          case.}
        \label{line:find-first-sep:fvprime-test}

          \State $(x,y)\gets$ first primal edge on dual tree path
          $f_u\cdots f_v'$.

          \State\parbox[t]{\linewidth-\algorithmicindent*3}{
            \setlength{\abovedisplayskip}{-\baselineskip}
            \setlength{\jot}{0pt}
            \setlength{\belowdisplayskip}{0pt}
            \begin{flalign*}
              \text{result} \gets \max \{
              &\text{result},
              &&\\
              &\Call{find-sep-P11}{u,v,f_u,C,e_u',e_v',x,y},
              &&\tag*{$\triangleright$
                \textit{$C$ separating and P}
              }\\
              &\Call{find-sep-R11}{u,v,f_u,C,e_u',e_v',x,y},
              &&\tag*{$\triangleright$
                \textit{$C$ separating and R}
              }\\
              &\Call{find-sep-P10}{u,v,f_u,C,e_u',e_v',x,y},
              &&\tag*{$\triangleright$
                \textit{$C$ not separating and P}
              }\\
              &\Call{find-sep-R10}{u,v,f_u,C,e_u',e_v',x,y}
              &&\tag*{$\triangleright$
                \textit{$C$ not separating and R}
              }\\
              \}
            \end{flalign*}
          }

        \EndIf

        \Statex

        \LeftComment{Handle cases where $E_{>b}\cap C=\emptyset$.}

        \State $e_v^1,e_v^2\gets$ the edges on $C$ incident to $\pi_C(v)$

        \State $\bar{f_v^1},\bar{f_v^2}\gets$ the faces incident to
        $e_v^1,e_v^2$ on the opposite side of $C$ from $f_u$.

        \State $\bar{f_v^3}\gets$ any face incident to $v$ e.g.\@ $v_L$ or $v_R$

        \State $\bar{f_1}\gets\meet(\bar{f_v^1},\bar{f_v^2},\bar{f_v^3})$
        \label{line:find-first-sep:f1bar}

        \State result $\gets\max\set{
          \text{result},
          \Call{choose-best-flip}{u,v,f_u,\bar{f_1}}
        }$

        \For{$\bar{f_2}\in\set{\bar{f_v^1},\bar{f_v^2}}\setminus\set{\bar{f_1}}$}
        \label{line:find-first-sep:f2bar}
        \Comment{Otherwise $\bar{f_2}$ can not be the right face}

          \State $(x,y)\gets$ first primal edge on the
          dual tree path from $\bar{f_1}$ to $\bar{f_2}$

          \State\parbox[t]{\linewidth-\algorithmicindent*4}{
            \setlength{\abovedisplayskip}{-\baselineskip}
            \setlength{\jot}{0pt}
            \setlength{\belowdisplayskip}{0pt}
            \begin{flalign*}
              \text{result} \gets \max \{
              &\text{result},
              &&\\
              &\Call{find-sep-P0x}{u,v,f_u,C,e_u',e_v',x,y},
              &&\tag*{$\triangleright$
                \textit{P}
              }\\
              &\Call{find-sep-R01}{u,v,f_u,C,e_u',e_v',x,y}
              &&\tag*{$\triangleright$
                \textit{$C$ separating and R}
              }\\
              \}
              &&&\tag*{$\triangleright$
                \textit{In this case it is impossible to have $C$ not separating and R}
              }
            \end{flalign*}
          }

        \EndFor

      \EndFor

      \State \Return result

    \EndFunction

  \end{algorithmic}
\end{algorithm}

\begin{lemma}\label{lem:find-first-sep-11}
  Consider \Call{find-first-separation-flip}{} in
  Algorithm~\ref{alg:find-first-sep}. Suppose
  $(f_u,\cdot,C,e_u',e_v')$ is good in
  Line~\ref{line:find-first-sep:assume-good}, and $E_{>b}\cap
  C\neq\emptyset$ and $C$ is separating. Then in
  Line~\ref{line:find-first-sep:fvprime-test}, $f_u\neq f_v'$ and the
  first primal edge $(x,y)$ on the dual tree path from $f_u$ to $f_v'$
  is in $E_{b}$ and closes a good fundamental cycle $C'$ such that
  $C'\setminus C$ is a path $\pi_C(x)\cdots\pi_C(y)$ in
  $E_{b}$. Furthermore, if $X_{b}$ is a P node,
  $\set{\pi_C(x),\pi_C(y)}=\set{s_{b-1},t_{b-1}}$.
\end{lemma}
\begin{proof}
  Since $(f_u,\cdot,C,e_u',e_v')$ is good and $E_{>b}\cap
  C\neq\emptyset$ we specifically have $e_v'\in E_{>b}$. Since the
  first $f_u$-blocking node on $X_1\cdots X_k$ has index
  $b<\idx(e_v')$, that means $f_u\neq f_v'$.  Let $(x,y)$ be the
  first primal edge on the dual tree path from $f_u$ to $f_v'$. Since
  $C$ is a fundamental cycle, the path $f_u \cdots f_v'$ will stay on
  one side of $C$, thus $(x,y)\in E_{b}$ and
  \begin{itemize}
  \item if $X_{b}$ is a P node the dual tree path $f_u \cdots f_v'$
    cuts every separation class w.r.t.\@ $\set{s_{b-1},t_{b-1}}$ on the $f_u$
    side of $C$, separating $s_{b-1}=s_{b}$ from
    $t_{b-1}=t_{b}$. Thus, $(x,y)$ is in a different separation class
    w.r.t.\@ $\set{s_{b-1},t_{b-1}}$ from any edge on $C$ and we must have
    $\set{s_{b-1},t_{b-1}}=\set{\pi_C(x),\pi_X(y)}$. The cycle $C'$ closed by
    $(x,y)$ is either separating (and contains $e_u'\in E_{\leq
      b}$), or contains $e_v'\in E_{>b}$.
  \item if $X_{b}$ is an R node the dual tree path $f_u \cdots f_v'$
    cuts $E_{b}$, separating $s_{b-1}$ and $s_{b}$ from $t_{b-1}$ and
    $t_{b}$. Thus, the fundamental cycle $C'$ closed by $(x,y)$ must
    connect $x$ to $y$ via primal tree paths $x\cdots
    \pi_C(x)\subseteq E_{b}$ and $y\cdots \pi_C(y)\subseteq
    E_{b}$, and the primal tree path $\pi_C(x)\cdots \pi_C(y)$ which
    (assuming $X_{b}$ is not a cross node) must go via either
    $s_{b-1}\cdots t_{b-1}\subseteq E_{<b}$ (thus $C'$ is separating and
    contains $e_u'$) or $s_{b}\cdots t_{b}\subseteq E_{>b}$ (and
    contains $e_v'$).
  \end{itemize}
  In either case, $C'$ is good.
\end{proof}

\begin{lemma}\label{lem:find-first-sep-10}
  Consider \Call{find-first-separation-flip}{} in
  Algorithm~\ref{alg:find-first-sep}. Suppose
  $(f_u,\cdot,C,e_u',e_v')$ is good in
  Line~\ref{line:find-first-sep:assume-good}, and $E_{>b}\cap
  C\neq\emptyset$ and $C$ is not separating. Then in
  Line~\ref{line:find-first-sep:fvprime-test}, $f_u\neq f_v'$ and the
  first primal edge $(x,y)$ on the dual tree path from $f_u$ to $f_v'$
  is in $E_{\le b}$ and closes a good fundamental cycle $C'$
  intersecting $E_{<b}$.
\end{lemma}
\begin{proof}
  Since $f_u$ and $f_v'$ are on the same side of $C$, by construction,
  the dual tree path $f_u \cdots f_v'$ must contain $f_v$. But then,
  the first edge $(x,y)$ closes a separating fundamental cycle. If
  $(x,y)$ is in $E_{<b}$ we are done, so suppose $(x,y)\in E_{b}$
  which is the only remaining case. Then, either $x$ or $y$ is
  separated from $C$ in $E_{b}\setminus (f_u\cdots f_v')$. But then,
  the primal tree path $x\cdots y$ must go via $s_{b-1}\cdots
  t_{b-1}\subseteq E_{<b}$.
\end{proof}

\begin{lemma}\label{lem:find-first-sep-f1f2-x0}%
  Let $C$ be a good cycle that is not separating, let $e_1,e_2$ be the
  edges incident to $\pi_C(u)$ on $C$, and let $f_v^1,f_v^2$ be the
  faces incident to $e_1,e_2$ on the same side of $C$ as $f_u$, then
  $f_v^1\neq f_v^2$ and $f_v\in f_v^1\cdots f_v^2$.
\end{lemma}
\begin{proof}
  Since $C$ is good and not separating, then by definition $C$
  intersects both $E_{\leq b}$ and $E_{>b}$ but not $E_{<b}$,
  and in particular $E_{b}\cap C$ consists of a path $s_{b}\cdots
  t_{b}$.
  Since $u,v$ are assumed to be biconnected, $f_v^1$ and $f_v^2$ are
  distinct. Furthermore let $s'$ be the first vertex on the primal
  tree path from $\pi_C(u)$ to $u$ such that the remaining path
  $s'\cdots u\subseteq E_{<b}$, then $s'\in\set{s_{b-1},t_{b-1}}$.

  Since $s_{b}$ and $t_{b}$ are distinct, at least one of them is
  different from $\pi_C(u)$, and thus at least one of $f_v^1$ and
  $f_v^2$ is incident to an edge on a path in $E_{b}\cap(C\cup T)$
  from $s_{b}$ or $t_{b}$ to $s'$.  Suppose without loss of
  generality that $t_{b}\neq\pi_C(u)$ and that $f_v^2$ is incident
  to the first edge on $\pi_C(u)\cdots t_{b}$.  Since $t_{b}\in
  f_v$ and $s'\in f_v$, the path $s'\cdots t_{b}$ in $C\cup T$ together
  with (an imaginary edge $(t_{b},s')$ through) $f_v$ form a closed
  curve separating $f_v^1$ from $f_v^2$. Since $s'\cdots t_{b}$ is in
  $C\cup T$, the path $f_v^1\cdots f_v^2$ does not cross it, and thus it
  must contain $f_v$.
\end{proof}

\begin{lemma}\label{lem:find-first-sep-f1f2-0x}%
  Let $C$ be a good cycle with $E_{>b}\cap C=\emptyset$, let
  $e_1,e_2$ be the edges incident to $\pi_C(v)$ on $C$, and let
  $\bar{f_v^1},\bar{f_v^2}$ be the faces incident to $e_1,e_2$ on the
  opposite side of $C$ from $f_u$, then $\bar{f_v^1}\neq\bar{f_v^2}$
  and $f_v\in\bar{f_v^1}\cdots\bar{f_v^2}$.
\end{lemma}
\begin{proof}
  Since $C$ is good and does not intersect $E_{>b}$, we must have
  that $C$ is separating and $E_{b}\cap C$ contains a path
  $s_{b-1}\cdots t_{b-1}$.
  Since $u,v$ are assumed to be biconnected $\bar{f_v^1}$ and
  $\bar{f_v^2}$ are distinct.  Furthermore let $s'$ be the first
  vertex on the primal tree path from $\pi_C(v)$ to $v$ such that the
  remaining path $s'\cdots v\subseteq E_{>b}$, then
  $s'\in\set{s_{b},t_{b}}$.

  Since $s_{b-1}$ and $t_{b-1}$ are distinct, at least one of them is
  different from $\pi_C(v)$, and thus at least one of $\bar{f_v^1}$
  and $\bar{f_v^2}$ is incident to an edge on a path in
  $E_{b}\cap(C\cup T)$ from $s_{b-1}$ or $t_{b-1}$ to $s'$.  Suppose without
  loss of generality that $t_{b-1}\neq\pi_C(v)$ and that $\bar{f_v^2}$ is
  incident to the first edge on $\pi_C(v)\cdots t_{b-1}$.  Since $t_{b-1}\in
  f_v$ and $s'\in f_v$, the path $s'\cdots t_{b-1}$ in $C\cup T$ together
  with (an imaginary edge $(t_{b-1},s')$ through) $f_v$ forms a closed
  curve separating $\bar{f_v^1}$ from $\bar{f_v^2}$. Since $s'\cdots
  t_{b-1}$ is in $C\cup T$, the path $\bar{f_v^1}\cdots\bar{f_v^2}$ does not
  cross it and so must contain $f_v$.
\end{proof}

\begin{lemma}\label{lem:find-first-sep-0x}
  Consider \Call{find-first-separation-flip}{} in
  Algorithm~\ref{alg:find-first-sep}. If $(f_u,\cdot,C,e_u',e_v')$ is
  good in Line~\ref{line:find-first-sep:assume-good}, and
  $E_{>b}\cap C=\emptyset$. Then in
  Line~\ref{line:find-first-sep:f1bar} if $\bar{f_1}\neq f_v$ then for
  at least one
  $\bar{f_2}\in\set{\bar{f_v^1},\bar{f_v^2}}\setminus\set{\bar{f_1}}$,
  the first primal edge $(x,y)$ on the dual tree path from $\bar{f_1}$
  to $\bar{f_2}$ closes a good fundamental cycle $C'$ intersecting
  $E_{>b}$.
\end{lemma}
\begin{proof}
  By Lemma~\ref{lem:find-first-sep-f1f2-0x}
  $\bar{f_v^1}\neq\bar{f_v^2}$ and
  $f_v\in\bar{f_v^1}\cdots\bar{f_v^2}$. Furthermore, assuming
  $v\not\in f_v$ also $f_v$ and $\bar{f_v^3}$ are distinct.
  If $\bar{f_1}=f_v$ we are done, so suppose $\bar{f_1}\neq f_v$. Then
  either $f_v\not\in\bar{f_v^1}\cdots\bar{f_1}$ or
  $f_v\not\in\bar{f_v^2}\cdots\bar{f_1}$. Choose $j\in\set{0,1}$ such
  that $f_v\not\in\overline{f_v^{1+j}}\cdots\bar{f_1}$, then the first
  primal edge $(x,y)$ on the dual tree path from $\bar{f_1}$ to
  $\overline{f_v^{2-j}}$ closes a cycle $C'$ separating
  $\overline{f_v^{1+j}}\cdots\bar{f_1}\cdots\bar{f_v^3}$ from
  $f_v\cdots\overline{f_v^{2-j}}$.

  If $(x,y)\in E_{\leq b}$ then the path
  $\bar{f_v^3}\cdots\bar{f_1}\cdots f_v$ cuts $E_{\leq b}$, separating
  $s_{b}$ from $t_{b}$.  But then the primal tree path from $s_{b}$ to
  $t_{b}$ is contained in $E_{>b}$, and is part of $C'$. Thus $C'$
  intersects both $E_{\leq b}$ and $E_{>b}$ and is therefore good.

  Otherwise $(x,y)\in E_{>b}$, and $\bar{f_v^1}\cdots\bar{f_v^2}$ cuts
  $E_{>b}$, separating $s_{b}$ from $t_{b}$.  But then the primal tree
  path from $s_{b}$ to $t_{b}$ is contained in $E_{\leq b}$, and is
  part of $C'$. Thus $C'$ intersects both $E_{\leq b}$ and $E_{>b}$
  and is therefore good.
\end{proof}

\begin{algorithm}[htb!]
  \caption{}
  \label{alg:find-sep-P11}
  \begin{algorithmic}[1]

    \Function{find-sep-P11}{$u,v,f_u,C,e_u',e_v',x,y$}

    \LeftComment{Handle cases where $X_{b}$ is a P node and
      $E_{>b}\cap C\neq\emptyset$ and $C$ is separating.}

    \State $\bar{f_u},\bar{f_v}\gets$ face incident to $e_u',e_v'$ on
    opposite side of $C$ from $f_u$.

    \State $p_x\gets\pi_C(x)$; $p_y\gets\pi_C(y)$

    \If{$p_x\neq p_y$}\Comment{If $p_x=p_y$ we are not in this case.}

      \State $f_v\gets$ first face on $\bar{f_v}\cdots \bar{f_u}$
      incident to both $p_x$ and $p_y$

      \State \Return $\Call{choose-best-flip}{u,v,f_u,f_v}$

    \EndIf

    \State \Return $(0,\bot)$

    \EndFunction

  \end{algorithmic}
\end{algorithm}

\begin{lemma}\label{lem:find-sep-P11}\sloppy
  If $(f_u,\cdot,C,e_u',e_v')$ is good, $E_{>b}\cap C\neq\emptyset$,
  $C$ is separating, $(x,y)$ closes a fundamental cycle $C'$ such that
  $C'\setminus C$ is a path $\pi_C(x)\cdots\pi_C(y)$ in $E_{b}$ with
  $\set{\pi_C(x),\pi_C(y)}=\set{s_{b-1},t_{b-1}}$, and $X_{b}$ is a P node,
  then $\Call{find-sep-P11}{u,v,f_u,C,e_u',e_v',x,y}$ in
  Algorithm~\ref{alg:find-sep-P11} returns the size and corners of a
  maximal $u$-flip.
\end{lemma}
\begin{proof}
  Since $C$ is separating, $f_v$ is on the opposite side of $C$ from
  $f_u$, and thus by construction $\bar{f_u}$ and $\bar{f_v}$ are on
  the same side of $C$ as $f_v$.  The path $s_{b-1}\cdots t_{b-1}$ together
  with (an imaginary edge $(s_{b-1},t_{b-1})$ through) $f_v$ forms a closed
  curve separating $\bar{f_u}$ from $\bar{f_v}$, so $f_v\in
  \bar{f_u}\cdots\bar{f_v}$.  Since $f_v$ maximizes the size of the
  $u$-flip-component, it must be the first face on this path that
  defines a valid flip, which it does if and only if it contains
  $s_{b-1},t_{b-1}$.
  Given the correct $f_v$, the rest follows from
  Lemma~\ref{lem:choose-best-flip}.
\end{proof}

\begin{algorithm}[htb!]
  \caption{}
  \label{alg:find-sep-P10}
  \begin{algorithmic}[1]

    \Function{find-sep-P10}{$u,v,f_u,C,e_u',e_v',x,y$}

    \LeftComment{Handle cases where $X_{b}$ is a P node and
      $E_{>b}\cap C\neq\emptyset$ and $C$ is not separating.}

    \State $f_v'\gets$ the face incodent to $e_v'$ on the same side of
    $C$ as $f_u$.

    \State $p_x\gets\pi_C(x)$; $p_y\gets\pi_C(y)$

    \If{$p_x\neq p_y$}\Comment{If $p_x=p_y$ we are not in this case.}

      \State $f_v\gets$ first face on $f_v'\cdots f_u$
      incident to both $p_x$ and $p_y$

      \State \Return $\Call{choose-best-flip}{u,v,f_u,f_v}$

    \EndIf

    \State \Return $(0,\bot)$

    \EndFunction

  \end{algorithmic}
\end{algorithm}

\begin{lemma}\label{lem:find-sep-P10}\sloppy
  If $(f_u,\cdot,C,e_u',e_v')$ is good, $E_{>b}\cap C\neq\emptyset$,
  $C$ is not separating, $(x,y)$ closes a good fundamental cycle
  intersecting $E_{<b}$, and $X_{b}$ is a P node, then
  $\Call{find-sep-P10}{u,v,f_u,C,e_u',e_v',x,y}$ in
  Algorithm~\ref{alg:find-sep-P10} returns the size and corners of a
  maximal $u$-flip.
\end{lemma}
\begin{proof}
  Since $C$ is good but not separating, $C\cap E_{<b}=\emptyset$,
  $C\cap E_{b}\neq\emptyset$, and $C\cap
  E_{>b}\neq\emptyset$. Thus, $C$ contains $s_{b-1},t_{b-1}$ and since $C$
  is a fundamental cycle $C$ contains the tree path $s_{b-1}\cdots
  t_{b-1}$. Since $C'$ is a good fundamental cycle through $E_{<b}$ it
  must also contain that tree path, and we must have $(x,y)\in
  E_{<b}$ and thus $\set{\pi_C(x),\pi_C(y)}=\set{s_{b-1},t_{b-1}}$.

  Furthermore, since $C$ is not separating $f_v$ is on the same side
  of $C$ as $f_u$, and thus by construction $f_v'$ is on the same side
  of $C$ as $f_v$.  The path $s_{b-1}\cdots t_{b-1}$ together with (an
  imaginary edge $(s_{b-1},t_{b-1})$ through) $f_v$ forms a closed curve
  separating $f_u$ from $f_v'$, so $f_v\in f_u\cdots f_v'$.  Since
  $f_v$ maximizes the size of the $u$-flip-component, it must be the
  first face on this path that defines a valid flip, which it does if
  and only if it contains $s_{b-1},t_{b-1}$.
  Given the correct $f_v$, the rest follows from
  Lemma~\ref{lem:choose-best-flip}.
\end{proof}

\begin{algorithm}[htb!]
  \caption{}
  \label{alg:find-sep-R11}
  \begin{algorithmic}[1]

    \Function{find-sep-R11}{$u,v,f_u,C,e_u',e_v',x,y$}

    \LeftComment{Handle cases where $X_{b}$ is an R node and
      $E_{>b}\cap C\neq\emptyset$ and $C$ is separating.}

    \State result $\gets(0,\bot)$

    \State $\bar{f_u},\bar{f_v}\gets$ face incident to $e_u',e_v'$ on
    opposite side of $C$ from $f_u$.

    \State $p_x\gets\pi_C(x)$; $p_y\gets\pi_C(y)$

    \If{$p_x\neq p_y$}\Comment{If $p_x=p_y$ we are not in this case.}

      \State $e_x^1,e_x^2,e_y^1,e_y^2\gets$ edges on $C$ incident to
      $p_x,p_y$, with $e_x^1,e_y^1$ closest to $e_u'$

      \State $\bar{f_x^1},\bar{f_x^2},\bar{f_y^1},\bar{f_y^2}\gets$
      faces incident to $e_x^1,e_x^2,e_y^1,e_y^2$ on the same side of
      $C$ as $\bar{f_u}$

      \For{$\bar{f_x}\in\set{\bar{f_x^1},\bar{f_x^2}}$}

        \For{$\bar{f_y}\in\set{\bar{f_y^1},\bar{f_y^2}}$}

          \State $f_v\gets\meet(\bar{f_u},\bar{f_x},\bar{f_y})$

          \State result $\gets\max\set{
            \text{result},
            \Call{choose-best-flip}{u,v,f_u,f_v}
          }$

        \EndFor

      \EndFor

    \EndIf

    \State \Return result

    \EndFunction

  \end{algorithmic}
\end{algorithm}

\begin{lemma}\label{lem:find-sep-R-cycles-common}
  Suppose $X_{b}$ is an R node.  Let $C$ and $C'$ be good
  fundamental cycles such that $C\cup C'$ consists of $3$ internally
  vertex-disjoint paths $P_{<},P_{=},P_{>}$ between a pair of distinct
  vertices $p_x,p_y$, where $P_{<}\cap E_{<b}\neq\emptyset$,
  $P_{=}\subseteq E_{b}$, and $P_{>}\cap E_{>b}\neq\emptyset$.
  Let $C''=P_{<}\cup P_{>}$, let $e\in C''\setminus E_{b}$, and let
  $e_x^1,e_x^2,e_y^1,e_y^2$ be the edges incident to $p_x,p_y$ on
  $C''$.  Let $f,f_x^1,f_x^2,f_y^1,f_y^2$ be the faces incident to
  $e,e_x^1,e_x^2,e_y^1,e_y^2$ on the same side of $C''$ as $f_v$.
  Then there exists $f_x\in\set{f_x^1,f_x^2}$ and
  $f_y\in\set{f_y^1,f_y^2}$ such that $f_v=\meet(f,f_x,f_y)$.
\end{lemma}
\begin{proof}
  Suppose that $e\in E_{<b}$ (the case $e\in E_{>b}$ is
  symmetric).  Since $X_{b}$ is an R node, at least one of $s_{b-1}\neq
  s_{b}$ and $t_{b-1}\neq t_{b}$ holds.  We may assume without loss of
  generality that $s_{b-1}\neq s_{b}$ and that $p_x$ is on the
  $s_{b-1}\cdots s_{b}$ path in $C''\cap E_{b}$.  Then there exists at
  least one $e_x\in\set{e_x^1,e_x^2}\cap E_{b}$, and at least one
  $e_y\in\set{e_y^1,e_y^2}\cap E_{\geq b}$.  Let
  $f_x\in\set{f_x^1,f_x^2},f_y\in\set{f_y^1,f_y^2}$ be the
  corresponding faces.

  Now $f_v\in f\cdots f_x$ and $f_v\in f\cdots f_y$ because $C''\cap
  E_{<b}$ consists of a path $s_{b-1}\cdots t_{b-1}$ which together with (an
  imaginary edge $(s_{b-1},t_{b-1})$ through) $f_v$ forms a closed curve
  separating $f$ from $f_x$ and $f_y$.

  Similarly $f_v\in f_x\cdots f_y$ because the path $s_{b-1}\cdots
  s_{b}$ in $C\cap E_{b}$  together with (an
  imaginary edge $(s_{b-1},s_{b})$ through) $f_v$ forms a closed curve
  separating $f_x$ from $f_y$.

  Since $f_v$ is on all $3$ paths $f\cdots f_x$, $f\cdots f_y$, and
  $f_x\cdots f_y$, we have $f_v=\meet(f,f_x,f_y)$.
\end{proof}

\begin{lemma}\label{lem:find-sep-R-cycles-separate}
  Suppose $X_{b}$ is an R node.  Let $C_\leq\subseteq E_{\leq b}$
  and $C_\geq\subseteq E_{\geq b}$ be good fundamental cycles with
  at most one vertex in common.  Let $e\in (C_\leq\cup
  C_\geq)\setminus E_{b}$, and let $e_x^1,e_x^2,e_y^1,e_y^2$ be the
  edges incident to $p_x=\pi_{C_\leq}(v),p_y=\pi_{C_\geq}(u)$ on
  $C_\leq,C_\geq$.  Let $f,f_x^1,f_x^2,f_y^1,f_y^2$ be the faces
  incident to $e,e_x^1,e_x^2,e_y^1,e_y^2$ on the same side of
  $C_\leq,C_\geq$ as $f_v$.  Then there exists
  $f_x\in\set{f_x^1,f_x^2}$ and $f_y\in\set{f_y^1,f_y^2}$ such that
  $f_v=\meet(f,f_x,f_y)$.
\end{lemma}
\begin{proof}
  Suppose that $e\in E_{<b}$ (the case $e\in E_{>b}$ is
  symmetric). Then there exists at least one
  $e_x\in\set{e_x^1,e_x^2}\cap E_{b}$. Let $f_x\in\set{f_x^1,f_x^2}$
  be the corresponding face.  By
  Lemma~\ref{lem:find-first-sep-f1f2-x0}
  or~\ref{lem:find-first-sep-f1f2-0x} $f_y^1\neq f_y^2$ and $f_v\in
  f_y^1\cdots f_y^2$, so there exists $f_y\in\set{f_y^1,f_y^2}$ such
  that $f_v\in f_x\cdots f_y$.

  Now $f_v\in f\cdots f_x$ and $f_v\in f\cdots f_y$ because $C_\leq\cap
  E_{<b}$ consists of a path $s_{b-1}\cdots t_{b-1}$ which together with (an
  imaginary edge $(s_{b-1},t_{b-1})$ through) $f_v$ forms a closed curve
  separating $f$ from $f_x$ and $f_y$.

  Since $f_v$ is on all $3$ paths $f\cdots f_x$, $f\cdots f_y$, and
  $f_x\cdots f_y$, we have $f_v=\meet(f,f_x,f_y)$.
\end{proof}

\begin{lemma}\label{lem:find-sep-R11}\sloppy
  If $(f_u,\cdot,C,e_u',e_v')$ is good, $E_{>b}\cap C\neq\emptyset$,
  $C$ is separating, $(x,y)$ closes a fundamental cycle $C'$ such that
  $C'\setminus C$ is a path $\pi_C(x)\cdots\pi_C(y)$ in $E_{b}$, and
  $X_{b}$ is an R node, then
  $\Call{find-sep-R11}{u,v,f_u,C,e_u',e_v',x,y}$ in
  Algorithm~\ref{alg:find-sep-R11} returns the size and corners of a
  maximal $u$-flip.
\end{lemma}
\begin{proof}
  By Lemma~\ref{lem:choose-best-flip} it is sufficient to show that at
  least one of the candidates to $f_v$ used in the calls to
  $\Call{choose-best-flip}{u,v,f_u,f_v}$ is correct.  This follows
  directly from Lemma~\ref{lem:find-sep-R-cycles-common} with $e=e_u'$
  and the fact that $\bar{f_u}$ is on the same side of $C$ as $f_v$.
\end{proof}

\begin{algorithm}[htb!]
  \caption{}
  \label{alg:find-sep-R10}
  \begin{algorithmic}[1]

    \Function{find-sep-R10}{$u,v,f_u,C,e_u',e_v',x,y$}

    \LeftComment{Handle cases where $X_{b}$ is an R node and
      $E_{>b}\cap C\neq\emptyset$ and $C$ is not separating.}

    \State result $\gets(0,\bot)$

    \State $f_v'\gets$ the face incodent to $e_v'$ on the same side of
    $C$ as $f_u$.

    \State $C'\gets$ fundamental cycle of $(x,y)$

    \State $p_x\gets\pi_C(x)$; $p_y\gets\pi_C(y)$

    \If{$p_x\neq p_y$} \Comment{The two cycles have common edges (case I)}

      \State $e_x^1,e_y^1\gets$ edges on $C'\setminus C$ incident to
      $p_x,p_y$

      \If{$e_u'\in C'$}

        \State $e_x^2,e_y^2\gets$ edges on $C\setminus C'$ incident to
        $p_x,p_y$

      \Else{ $e_u'\not\in C'$}

        \State $e_x^2,e_y^2\gets$ edges on $C\cap C'$ incident to
        $p_x,p_y$

      \EndIf

      \State $f_x^1,f_x^2,f_y^1,f_y^2\gets$
      faces incident to $e_x^1,e_x^2,e_y^1,e_y^2$ on the same side of
      $C,C'$ as $f_v'$

      \For{$f_x\in\set{f_x^1,f_x^2}$}

        \For{$f_y\in\set{f_y^1,f_y^2}$}

          \State $f_v\gets\meet(f_v',f_x,f_y)$

          \State result $\gets\max\set{
            \text{result},
            \Call{choose-best-flip}{u,v,f_u,f_v}
          }$

        \EndFor

      \EndFor

    \Else{ $p_x=p_y$} \Comment{The two cycles have no common edges (case X
      or H)}

      \State $e_u^1,e_u^2\gets$ edges incident to $\pi_{C'}(v)$ on $C'$

      \State $e_v^1,e_v^2\gets$ edges incident to $\pi_{C}(u)$ on $C$

      \State $f_u^1,f_u^2,f_v^1,f_v^2\gets$ faces incident to
      $e_u^1,e_u^2,e_v^1,e_v^2$ on the same side of $C,C'$ as $f_v'$

      \For{$f_x\in\set{f_u^1,f_u^2}$}

        \For{$f_y\in\set{f_v^1,f_v^2}$}

          \State $f_v\gets\meet(f_v',f_x,f_y)$

          \State result $\gets\max\set{
            \text{result},
            \Call{choose-best-flip}{u,v,f_u,f_v}
          }$

        \EndFor

      \EndFor

    \EndIf

    \State \Return result

    \EndFunction

  \end{algorithmic}
\end{algorithm}

\begin{lemma}\label{lem:find-sep-R10}\sloppy
  If $(f_u,\cdot,C,e_u',e_v')$ is good, $E_{>b}\cap C\neq\emptyset$,
  $C$ is not separating, $(x,y)$ closes a good fundamental cycle
  intersecting $E_{<b}$, and $X_{b}$ is an R node, then
  $\Call{find-sep-R10}{u,v,f_u,C,e_u',e_v',x,y}$ in
  Algorithm~\ref{alg:find-sep-R10} returns the size and corners of a
  maximal $u$-flip.
\end{lemma}
\begin{proof}
  By Lemma~\ref{lem:choose-best-flip} it is sufficient to show that at
  least one of the candidates to $f_v$ used in the calls to
  $\Call{choose-best-flip}{u,v,f_u,f_v}$ is correct.

  If $p_x\neq p_y$ then $C\cup C'$ consists of $3$ internally
  vertex-disjoint paths from $p_x$ to $p_y$, and by
  Lemma~\ref{lem:find-sep-R-cycles-common} with $e=e_v'$ there exists
  $f_x\in\set{f_x^1,f_x^2}$ and $f_y\in\set{f_y^1,f_y^2}$ such that
  $f_v=\meet(f_v',f_x,f_y)$.

  Otherwise $p_x=p_y$, and $C$ and $C'$ has at most one vertex in
  common, so by Lemma~\ref{lem:find-sep-R-cycles-separate} with
  $e=e_v'$ there exists $f_x\in\set{f_x^1,f_x^2}$ and
  $f_y\in\set{f_y^1,f_y^2}$ such that $f_v=\meet(f_v',f_x,f_y)$.
\end{proof}

\begin{algorithm}[htb!]
  \caption{}
  \label{alg:find-sep-P0x}
  \begin{algorithmic}[1]

    \Function{find-sep-P0x}{$u,v,f_u,C,e_u',e_v',x,y$}

    \LeftComment{Handle cases where $X_{b}$ is a P node and
      $E_{>b}\cap C=\emptyset$.}

    \State $\bar{f_u},\bar{f_v}\gets$ face incident to $e_u',e_v'$ on
    opposite side of $C$ from $f_u$.

    \State $C'\gets$ fundamental cycle of $(x,y)$

    \State $p_x\gets\pi_C(x)$; $p_y\gets\pi_C(y)$

    \If{$p_x\neq p_y$}\Comment{If $p_x=p_y$ we are not in this case.}

      \State $\bar{f}\gets$ the face incident to
      $(x,y)$ on the same side of $C'$ as $\bar{f_u}$

      \State $f_v\gets$ first face on $\bar{f}\cdots \bar{f_u}$
      incident to both $p_x$ and $p_y$

      \State \Return $\Call{choose-best-flip}{u,v,f_u,f_v}$

    \EndIf

    \State \Return $(0,\bot)$

    \EndFunction

  \end{algorithmic}
\end{algorithm}

\begin{lemma}\label{lem:find-sep-P0x}\sloppy
  If $(f_u,\cdot,C,e_u',e_v')$ is good, $E_{>b}\cap C=\emptyset$,
  $(x,y)$ closes a good fundamental cycle intersecting $E_{>b}$, and
  $X_{b}$ is a P node, then
  $\Call{find-sep-P0x}{u,v,f_u,C,e_u',e_v',x,y}$ in
  Algorithm~\ref{alg:find-sep-P0x} returns the size and corners of a
  maximal $u$-flip.
\end{lemma}
\begin{proof}
  Since $C$ is good and $E_{>b}\cap C=\emptyset$, $C$ is separating
  and contains $s_{b-1},t_{b-1}$ as well as the tree path $s_{b-1}\cdots
  t_{b-1}$. Since $C'$ is a good fundamental cycle through $E_{>b}$ it
  must also contain that tree path, and we must have $(x,y)\in
  E_{>b}$ and thus $\set{\pi_C(x),\pi_C(y)}=\set{s_{b-1},t_{b-1}}$.

  Furthermore, since $C$ is separating $f_v$ is on the opposite side
  of $C$ from $f_u$, and thus by construction $\bar{f_u}$ and
  $\bar{f}$ are on the same side of $C$ as $f_v$.  The path $s_{b-1}\cdots
  t_{b-1}$ together with (an imaginary edge $(s_{b-1},t_{b-1})$ through) $f_v$
  forms a closed curve separating $\bar{f}$ from $\bar{f_u}$, so
  $f_v\in \bar{f}\cdots\bar{f_u}$.  Since $f_v$ maximizes the size of
  the $u$-flip-component, it must be the first face on this path that
  defines a valid flip, which it does if and only if it contains
  $s_{b-1},t_{b-1}$.
  Given the correct $f_v$, the rest follows from
  Lemma~\ref{lem:choose-best-flip}.
\end{proof}

\begin{algorithm}[htb!]
  \caption{}
  \label{alg:find-sep-R01}
  \begin{algorithmic}[1]

    \Function{find-sep-R01}{$u,v,f_u,C,e_u',e_v',x,y$}

    \LeftComment{Handle cases where $X_{b}$ is an R node and
      $E_{>b}\cap C=\emptyset$ and $C$ is separating.}

    \State result $\gets(0,\bot)$

    \State $\bar{f_u},\bar{f_v}\gets$ face incident to $e_u',e_v'$ on
    opposite side of $C$ from $f_u$.

    \State $C'\gets$ fundamental cycle of $(x,y)$

    \State $p_x\gets\pi_C(x)$; $p_y\gets\pi_C(y)$

    \If{$p_x\neq p_y$} \Comment{The two cycles have common edges (case I)}

      \If{$e_v'\in C'$}

        \State $e_x^1,e_y^1\gets$ edges on $C\setminus C'$ incident to
        $p_x,p_y$

      \Else{ $e_v'\not\in C'$}

        \State $e_x^1,e_y^1\gets$ edges on $C\cap C'$ incident to
        $p_x,p_y$

      \EndIf

      \State $e_x^2,e_y^2\gets$ edges on $C'\setminus C$ incident to
      $p_x,p_y$

      \State $\bar{f_x^1},\bar{f_x^2},\bar{f_y^1},\bar{f_y^2}\gets$
      faces incident to $e_x^1,e_x^2,e_y^1,e_y^2$ on the same side of
      $C,C'$ as $\bar{f_u}$

      \For{$\bar{f_x}\in\set{\bar{f_x^1},\bar{f_x^2}}$}

        \For{$\bar{f_y}\in\set{\bar{f_y^1},\bar{f_y^2}}$}

          \State $f_v\gets\meet(\bar{f_u},\bar{f_x},\bar{f_y})$

          \State result $\gets\max\set{
            \text{result},
            \Call{choose-best-flip}{u,v,f_u,f_v}
          }$

        \EndFor

      \EndFor

    \Else{ $p_x=p_y$} \Comment{The two cycles have no common edges (case X
      or H)}

      \State $e_u^1,e_u^2\gets$ edges incident to $\pi_{C}(v)$ on $C$

      \State $e_v^1,e_v^2\gets$ edges incident to $\pi_{C'}(u)$ on $C'$

      \State $\bar{f_u^1},\bar{f_u^2},\bar{f_v^1},\bar{f_v^2}\gets$
      faces incident to $e_u^1,e_u^2,e_v^1,e_v^2$ on the same side of
      $C,C'$ as $\bar{f_u}$

      \For{$\bar{f_x}\in\set{\bar{f_u^1},\bar{f_u^2}}$}

        \For{$\bar{f_y}\in\set{\bar{f_v^1},\bar{f_v^2}}$}

          \State $f_v\gets\meet(\bar{f_u},\bar{f_x},\bar{f_y})$

          \State result $\gets\max\set{
            \text{result},
            \Call{choose-best-flip}{u,v,f_u,f_v}
          }$

        \EndFor

      \EndFor

    \EndIf

    \State \Return result

    \EndFunction

  \end{algorithmic}
\end{algorithm}

\begin{lemma}\label{lem:find-sep-R01}\sloppy
  If $(f_u,\cdot,C,e_u',e_v')$ is good, $E_{>b}\cap C=\emptyset$,
  $E_{<b}\cap C\neq\emptyset$, $(x,y)$ closes a good fundamental
  cycle intersecting $E_{>b}$, and $X_{b}$ is an R node, then
  $\Call{find-sep-R01}{u,v,f_u,C,e_u',e_v',x,y}$ in
  Algorithm~\ref{alg:find-sep-R01} returns the size and corners of a
  maximal $u$-flip.
\end{lemma}
\begin{proof}
  By Lemma~\ref{lem:choose-best-flip} it is sufficient to show that at
  least one of the candidates to $f_v$ used in the calls to
  $\Call{choose-best-flip}{u,v,f_u,f_v}$ is correct.

  If $p_x\neq p_y$ then $C\cup C'$ consists of $3$ internally
  vertex-disjoint paths from $p_x$ to $p_y$, and by
  Lemma~\ref{lem:find-sep-R-cycles-common} with $e=e_u'$ there exists
  $\bar{f_x}\in\set{\bar{f_x^1},\bar{f_x^2}}$ and
  $\bar{f_y}\in\set{\bar{f_y^1},\bar{f_y^2}}$ such that
  $f_v=\meet(\bar{f_u},\bar{f_x},\bar{f_y})$.

  Otherwise $p_x=p_y$, and $C$ and $C'$ has at most one vertex in
  common, so by Lemma~\ref{lem:find-sep-R-cycles-separate} with
  $e=e_u'$ there exists $\bar{f_x}\in\set{\bar{f_x^1},\bar{f_x^2}}$
  and $\bar{f_y}\in\set{\bar{f_y^1},\bar{f_y^2}}$ such that
  $f_v=\meet(\bar{f_u},\bar{f_x},\bar{f_y})$.
\end{proof}

\begin{lemma}\label{lem:find-sep-R00}
  If $(f_u,\cdot,C,e_u',e_v')$ is good, $E_{>b}\cap C=\emptyset$,
  and $E_{<b}\cap C=\emptyset$, then $X_{b}$ is a P node.
\end{lemma}
\begin{proof}
  Follows directly from the definition of $f_u,C,e_u',e_v'$ being correct.
\end{proof}

\begin{lemma}\label{lem:find-first-sep}
  $\Call{find-first-separation-flip}{u,v}$ in
  Algorithm~\ref{alg:find-first-sep} runs in worst case $\OO(\log^2
  n)$ time, and:
  \begin{itemize}
  \item If $G\cup(u,v)$ is planar it always finds a maximal $u$-flip.
  \item If $G\cup(u,v)$ is non-planar it either:
    \begin{itemize}
    \item finds a maximal $u$-flip; or
    \item finds a $u$-flip $\sigma$ such that immediately calling
      $\Call{find-first-separation-flip}{u,v}$ again after executing
      $\sigma$ will return a $u$-flip $\sigma'$ of the same size; or
    \item finds no $u$-flip.
    \end{itemize}
  \end{itemize}
\end{lemma}
\begin{proof}
  If $G\cup(u,v)$ is planar and only one flip is needed to admit
  $(u,v)$, then by Lemma~\ref{lem:find-single-findsflip} our
  \Call{find-single-flip-candidates}{} algorithm will find a correct
  candidate tuple, and then by Lemma~\ref{lem:choose-best-flip}
  \Call{choose-best-flip}{} will select the corresponding maximal
  $u$-flip, and this will be returned.

  Otherwise, let $b$ be the minimum index of an $f_u$-blocking node.
  If $b=1$ then no $u$-flip exists and so no $u$-flip is found.  If
  $1<b<r$, by Lemma~\ref{lem:find-single-findsgood} at least one of
  the candidate tuples is good.  Let $(f_u,\cdot,C,e_u',e_v')$ be the
  good candidate tuple.  Let $E_1,\ldots,E_k$ be as in
  Definition~\ref{def:E-partition}. We will case by how $C$ intersects
  $E_{>b}$, and then by whether $C$ separates $f_u$ from $f_v$, and
  then, finally, by whether the $X_b$ is a P or an R node.
  \begin{itemize}
  \item If $C\cap E_{>b}\neq \emptyset$
    \begin{itemize}
    \item if $C$ is separating, then by
      Lemma~\ref{lem:find-first-sep-11}, $(x,y)$ closes a good fundamental
      cycle $C'$, such that $C'\setminus C$ is a path in $E_{b}$, and:
      \begin{itemize}
      \item If $X_{b}$ is a P node,
        $\set{s_i,t_i}=\set{\pi_C(x),\pi_C(y)}$ so by
        Lemma~\ref{lem:find-sep-P11}, we return the maximal $u$-flip.
      \item If $X_{b}$ is an R node, then by
        Lemma~\ref{lem:find-sep-R11}, we return the maximal $u$-flip.
      \end{itemize}
    \item if $C$ is not separating, then by
      Lemma~\ref{lem:find-first-sep-10}, $(x,y)$ closes a good
      fundamental cycle $C'$ intersecting $E_{<b}$, and:
      \begin{itemize}
      \item If $X_{b}$ is a P node, then by
        Lemma~\ref{lem:find-sep-P10}, we return the maximal $u$-flip.
      \item If $X_{b}$ is an R node, then by
        Lemma~\ref{lem:find-sep-R10}, we return the maximal $u$-flip.
      \end{itemize}
    \end{itemize}
  \item If $C\cap E_{>b} = \emptyset$ then by
    Lemma~\ref{lem:find-first-sep-0x} at least one of the $(x,y)$ we
    try closes a cycle $C'$ intersecting $E_{>b}$, and:
    \begin{itemize}
    \item If $X_{b}$ is a P node, then by
      Lemma~\ref{lem:find-sep-P0x}, we return the maximal $u$-flip.
    \item If $X_{b}$ is an R node, then, by
      Lemma~\ref{lem:find-sep-R00}, $C'$ is separating. Then, by
      Lemma~\ref{lem:find-sep-R01}, we return the maximal $u$-flip.
    \end{itemize}
  \end{itemize}
  Finally, if $b=r$ and $G\cup(u,v)$ is nonplanar then we may find a
  non-maximal $u$-flip $\sigma$. This flip will be in some separation
  pair $\set{s_j,t_j}$ with $j<b-1$, and since every flip is chosen by
  a call $\Call{choose-best-flip}{u,v,f_u,f_v'}$ to be maximal for
  $f_u,f_v'$, $X_j$ is not an S node.  But then after performing the
  flip, $X_j$ will be the first $f_u$-blocking node, and since
  $\sigma$ was locally maximal it will not be anti-$f_u$-blocking.  By
  the previous argument we are guaranteed that the next call to
  $\Call{find-first-separation-flip}{u,v}$ finds the unique (since
  $X_j$ is not anti-$f_u$-blocking) maximal $u$-flip, which is the
  inverse of $\sigma$ and therefore has the same size.
\end{proof}

\FloatBarrier

We are finally ready to prove a main theorem:
\mainalgplanar*
\begin{proof}
  Consider the function $\Call{multi-flip-linkable}{u,v}$ from
  Algorithm~\ref{alg:multi-flip-linkable}.  By
  Lemma~\ref{lem:find-first-sep} and Lemma~\ref{lem:do-sep} each
  $\Call{do-separation-flips}{u,v}$ runs in $\OO(\log^2 n)$ time per
  flip, and only does critical potential-decreasing flips (and at most
  one critical flip that is not potential-decreasing). Similarly, by
  Lemma~\ref{lem:do-art} each
  $\Call{do-articulation-flips}{u,u',v',v}$ runs in $\OO(\log^2 n)$
  time and only does critical potential-decreasing or
  potential-neutral flips. By Lemma~\ref{lem:art-flip-pattern} any
  such potential-neutral flip is immediately followed by either a
  potential-decreasing flip (either a separation flip or an
  articulation flip) or by the final flip. Thus by
  Corollary~\ref{cor:lazygreedy-simple} with $r=2$,
  $\Call{multi-flip-linkable}{u,v}$ does amortized $\OO(\log n)$
  flips.

  By Lemma~\ref{lem:next-flip}, each
  $\Call{find-next-flip-block}{u,u',v',v}$ call also runs in
  $\OO(\log^2 n)$ time. By Lemma~\ref{lem:multiflip-iterations} the
  main loop in $\Call{multi-flip-linkable}{u,v}$ iterates amortized
  $\OO(\log n)$ times, and we have shown that each iteration takes
  $\OO(\log^2 n)$ time (in addition to the time taken by the
  separation flips). Thus the total amortized time for
  $\Call{multi-flip-linkable}{u,v}$ is $\OO(\log^3 n)$.

  Using that, the remaining edge insertion and deletion is
  trivial. And the queries to the embedding are handled directly by
  the underlying data structure from~\cite{DBLP:journals/mst/HolmR17}.
\end{proof}

\section{Allowing non-planar insertions}\label{sec:reduction}

In \cite[p.12, proof of Corollary~1]{Eppstein:1996}, Eppstein et al. give a reduction from any data structure that maintains a planar graph subject to deletions and planarity-preserving insertions and answers queries to the planarity-compatibility of edges, to a data structure that allows the graph to be non-planar and furthermore maintains whether %
the graph is presently planar, at the same amortized time. 
The reduction uses the following simple and elegant argument: If some component of the graph is non-planar, keep a pile of not-yet inserted edges, and upon a deletion, add edges from the pile until either the pile is empty or a new planarity-blocking certificate edge is found. 

To maintain not only whether the graph is planar but maintain for each component whether it is planar, it becomes necessary to keep a pile of not-yet inserted edges for each nonplanar connected component. To maintain the connected components, we run an auxiliary fully-dynamic connectivity structure for the entire graph~\cite{Wulff-Nilsen16}, maintaining a spanning forest and the non-tree edges. We may mark the edges indicating whether they are inserted or not-yet inserted in the planar subgraph, and we may mark vertices indicating whether they are incident to not-yet inserted edges. When an edge deletion causes a non-planar component to break into two, a spanning tree breaks into two, say, $T_u$ and $T_v$. For each $i\in\set{u,v}$, we may efficiently find not-yet inserted edges incident to $T_i$, and insert them into the planar subgraph of the component spanned by $T_i$. We may continue until either all edges for that component are handled or until we find the first planarity violating edge. The method for efficiently finding not-yet inserted edges follows the exact same outline as the method for finding candidate replacement edges in the connectivity structure~\cite{Wulff-Nilsen16}. 
Thus, the time spent on each edge becomes $\OO(\log n / \log\log n)$ worst-case for finding it, plus $\OO(\log^3 n)$ amortized time for inserting it.
We have thus shown

\mainalggeneral*

\section{Defining critical-cost and solid-cost}\label{sec:define-costs}

The goal of this section is to properly define the two function
families $\critcost_\tau$ and $\solidcost_\tau$ mentioned earlier so
we can prove the claimed properties.

The general idea is for each function to define a set of
\emph{struts}, which are edges that can be inserted in $G$ without
violating planarity, and then measure the total number of flips needed
to accommodate all of them.

\begin{align*}
  \critcost_\tau(H; u,v) &= \sum_{(x,y)\in \critstruts(G; u,v)}\flipdist_\tau(H, \Emb(G; x,y))
  \\
  \solidcost_\tau(H; u,v) &= \sum_{(x,y)\in \solidstruts(G; u,v)}\flipdist_\tau(H, \Emb(G; x,y))
\end{align*}

\noindent{}We want our struts to have the following properties for any planar
graph $G$ with vertices $u,v$:
\begin{enumerate}[label={S\arabic*)}, ref={S\arabic*}]\sloppy
\item\label{it:struts-planar} $G\cup\solidstruts(G; u,v)$ is simple
  and planar.

\item\label{it:struts-independent} For any $H,H'\in\Emb(G)$ with
  $\flipdist_{\text{clean}}(H,H')=1$, there is at most one strut
  $(x,y)\in\solidstruts(G; u,v)$ such that $\flipdist_{\text{clean}}(H,\Emb(G;
  x,y))\neq\flipdist_{\text{clean}}(H',\Emb(G; x,y))$.

\item\label{it:struts-logn} If $H\in\Emb(G)$ admits $\solidstruts(G;
  u,v)$ then for any $u',v'$ there exists $H'\in\Emb(G)$ that admits
  $\solidstruts(G; u',v')$ such that
  $\flipdist_{\text{clean}}(H,H')\in\OO(\log n)$.

\item\label{it:struts-insert} If $G\cup(u,v)$ is simple and planar, then $\solidstruts(G;u,v)=\solidstruts(G\cup(u,v); u,v)\cup\set{(u,v)}$.

\item\label{it:struts-subset} $\critstruts(G;
  u,v)\subseteq\solidstruts(G; u,v)$.

\item\label{it:struts-existing} For $(u,v)\in G$, $\critstruts(G;
  u,v)=\emptyset$.

\item\label{it:struts-admissible} For $(u,v)\not\in G$,
  $\critstruts(G; u,v)=\set{(u,v)}$ if and only if $G\cup(u,v)$ is
  planar.

\item\label{it:struts-nonadmissible} If $G\cup(u,v)$ is nonplanar, and
  $v=a_0,B_1,a_1,\ldots,B_k,a_k=v$ are the endpoints, biconnected
  components/bridges, and articulation points on $u\cdots v$ in $G$, then
  \begin{enumerate}
  \item\label{it:struts-nonadmissible-biconn} For each $B_i$ such that
    $B_i\cup(a_{i-1},a_i)$ is nonplanar, $\critstruts(G; u,v)$
    contains a set $S_i$ of struts with both endpoints in $B_i$ such
    that $a_{i-1}$ and $a_i$ are triconnected in $B_i\cup S_i$ and
    $B_i\cup S_i$ is planar.
  \item\label{it:struts-nonadmissible-general} For each maximal
    subsequence $a_{\ell-1},B_\ell,a_\ell,\ldots,B_h,a_h$ such that
    $G\cup(a_{\ell-1},a_h)$ is planar, and such that at least one
    $B_i$ with $\ell\leq i\leq h$ is not a bridge, $\critstruts(G;
    u,v)$ contains a strut $(a_{\ell-1},a_h)$.
  \end{enumerate}

\end{enumerate}

\noindent{}Assuming our struts have these properties, we can now
prove our Lemmas about the costs.
\begin{proof}[Proof of Lemma~\ref{lem:costs-nonneg}]
  The inequality $\solidcost_\tau(G;u,v)\geq\critcost(G;u,v)\geq 0$
  follows from property~\ref{it:struts-subset} and the definition of
  $\solidcost_\tau$ and $\critcost_\tau$.
  The inequalities
  $\solidcost_{\text{clean}}\geq\solidcost_{\text{sep}}\geq\solidcost_{\text{P}}$
  and
  $\critcost_{\text{clean}}\geq\critcost_{\text{sep}}\geq\critcost_{\text{P}}$
  also follow, because in general
  $\flipdist_{\text{clean}}\geq\flipdist_{\text{sep}}\geq\flipdist_{\text{P}}$.
  Finally, if $G\cup(u,v)$ is planar, then by
  property~\ref{it:struts-admissible},
  $\critstruts(G;u,v)=\set{(u,v)}$ so
  $\critcost_{\text{clean}}(H;u,v)=\flipdist_{\text{clean}}(H;
  \Emb(G;u,v))$, which is $0$ if and only if $H\in\Emb(G;u,v)$.
\end{proof}

\begin{proof}[Proof of Lemma~\ref{lem:costs-delta}]
  Follows directly from properties~\ref{it:struts-independent}
  and~\ref{it:struts-subset}, and the definition of critical flip.
\end{proof}

\begin{proof}[Proof of Lemma~\ref{lem:costs-exist-decreasing-flip}]
  From the definition of $\critcost_\tau$ and $\solidcost_\tau$ as a
  sum over $\flipdist_\tau$, it is clear that if the cost is nonzero
  there exists a flip of type $\tau$ that decreases at least one of
  the terms. But by property~\ref{it:struts-independent}, this is then
  the only term that changes so this flip also decreases the sum.
\end{proof}

\begin{proof}[Proof of Lemma~\ref{lem:embgood-logn}]
  \sloppy By definition of $H\in\EmbGood(G)$, there exists vertices
  $u_{\min},v_{\min}$ such that
  $\solidcost(H;u_{\min},v_{\min})=0$. By definition of
  $\solidcost_{\text{clean}}$ that means $H$ admits
  $\solidstruts_{\text{clean}}(G;u_{\min},v_{\min})$. Then by
  Property~\ref{it:struts-logn} there exists $H'\in\Emb(G)$ that admits
  $\solidstruts(G;u,v)$ and has
  $\flipdist_{\text{clean}}(H,H')\in\OO(\log n)$.  Since $H'$ admits
  $\solidstruts(G;u,v)$, $H'\in\EmbGood(G;u,v)$, and thus
  $\flipdist_{\text{clean}}(H,\EmbGood(G;u,v)) \leq
  \flipdist_{\text{clean}}(H,H')\in\OO(\log n)$.
\end{proof}

\begin{proof}[Proof of Lemma~\ref{lem:emb-insert}]
  \sloppy By Property~\ref{it:struts-insert}, the only difference
  between $\solidstruts(G;u,v)$ and $\solidstruts(G\cup(u,v);u,v)$ is
  that the first includes the strut $(u,v)$. Since $H\in\Emb(G;u,v)$,
  $\flipdist_{\text{clean}}(H,\Emb(G;u,v))=0$ so removing this term
  does not change the sum.  By \todo{Not so
    fast...!}Property~\ref{it:struts-independent}, the possible flips
  for every other strut $(u',v')$ is unaffected, meaning that
  $\flipdist(H,\Emb(G;u',v'))=\flipdist(H\cup(u,v),
  \Emb(G\cup(u,v);u',v'))$.  In other words,
  $\solidcost(H;u,v)=\solidcost(H\cup(u,v);u,v)$.
\end{proof}

\begin{proof}[Proof of Lemma~\ref{lem:fliptypes-unchanged}]
  Follows directly from the definition of
  $\critcost_\tau$ and $\solidcost_\tau$ as a sum
  $\sum_{(x,y)}\flipdist_\tau(H;x,y)$, and the definition of
  $\flipdist_\tau$. By definition an SR flip can not change any
  $\flipdist_{\text{P}}$, and an articulation flip can not change any
  $\flipdist_{\text{P}}$ or $\flipdist_{\text{sep}}$.
\end{proof}

\subsection{Biconnected planar graphs}\label{sec:struts-biconn}

Let SPQR$(B; u,v)$ denote the solid paths in the pre-split SPQR tree for the biconnected component $B$ with critical vertices $u,v$.

Let $\beta$ be a solid path in SPQR$(B; u,v)$.  The \emph{relevant
  part} of $\beta$ is the maximal subpath that does not end in a P
node.
\paragraph{Single solid SPQR path, simple case}
If $\beta$ consists only of a P node, or the relevant part of
$\beta$ is only a single node, define the struts relevant for $\beta$ as:
\begin{align*}
  \struts(\beta) &=
  \begin{cases}
    \set{(u,v)} &\text{if $\beta$ is the critical path, $(u,v)\not\in B$, and $G\cup(u,v)$ is planar}
    \\
    \emptyset &\text{otherwise}
  \end{cases}
\end{align*}
\paragraph{Single solid SPQR path, general case}Otherwise let $X_1,\ldots,X_d$ be the relevant part of $\beta$.
For $1\leq j<d$ let $\set{s_j,t_j}=X_j\cap X_{j+1}$ be the separation
pair that separates $X_j$ from $X_{j+1}$.
For $1<j<d$ we call the node $X_j$ \emph{cross} if $X_j$ is an R node,
and the virtual edges $(s_{j-1},t_{j-1})$ and $(s_j,t_j)$ do not share
a face in $\Gamma(X_j)$. If $\beta$ is the critical path, we can
assume without loss of generality that $u\in X_1$ and $v\in X_d$, and
we say that $X_1$ (resp. $X_d$) is cross if it is an R node and $u$
(resp. $v$) does not share a face with $(s_1,t_1)$
(resp. $(s_{d-1},t_{d-1})$).

Let $\gamma=X_\ell,\ldots,X_h$ be a maximal subpath of
$X_1,\ldots,X_d$ such that $X_j$ is not a cross node for $\ell<j<h$.
Let $u_\gamma$ be the smallest-labelled vertex in $\Gamma(X_\ell)-\set{s_\ell,t_\ell}$
that shares a face with $(s_\ell,t_\ell)$
(counting $u$ as having label $-\infty$).
Similarly let $v_\gamma$ be the smallest-labelled vertex in
$\Gamma(X_h)-\set{s_{h-1},t_{h-1}}$ that shares a face with
$(s_{h-1},t_{h-1})$ (counting $v$ as having label $-\infty$).

We can now define the struts relevant for $\beta$ as
\begin{align*}
  \struts(\beta) &= \set{(u_\gamma,v_\gamma)\suchthat
    \text{$\gamma$ is such a maximal subpath}
  }
\end{align*}

\paragraph{Combining the struts for a SPQR tree}
\begin{align*}
  \critstruts(B; u,v) &= \bigcup_{\substack{\beta\in\text{SPQR}(B;u,v)\\\beta\text{ is critical}}}
  \struts(\beta)
  \\
  \offstruts(B; u,v) &= \bigcup_{\substack{\beta\in\text{SPQR}(B;u,v)\\\beta\text{ is not critical}}}
  \struts(\beta)
  \\
  \solidstruts(B; u,v) &= \critstruts(B; u,v) \cup \offstruts(B; u,v)
\end{align*}

\subsection{General planar graphs}\label{sec:struts-general}

Let BC$(G; u,v)$ denote the set of solid paths in the forest of
pre-split BC trees for a planar graph $G$ with critical vertices
$u,v$.

Let $\alpha$ be a solid path in BC$(G; u,v)$. The \emph{relevant part} of $\alpha$ is the maximal subpath that does not end in a C node.

\paragraph{Single solid BC path, simple case}If $\alpha$ consists only of a C node, define
\begin{align*}
 \critstruts(\alpha)=\emptyset
 \\
 \offstruts(\alpha)=\emptyset
 \\
 \solidstruts(\alpha)=\emptyset
\end{align*}
\paragraph{Single solid BC path, general case}In~\cite{HR20} we defined a set of struts for each such path and used it to choose a critical path in the SPQR tree for each biconnected component.  Since those struts might make the graph non-planar, we need a new set;
we want to substitute every non-planar strut with a family of ``maximal'' planar struts in the following sense.
Let $B_1,\ldots,B_k$ be the B nodes on
$\alpha$, for $1<i<k$ let $a_i=B_i\cap B_{i+1}$, and let $(a_0,a_k)$
be the strut defined for $\alpha$ in~\cite{HR20}. Note that if
$\alpha$ is the critical path in BC$(G; u,v)$ then we may assume
without loss of generality that $a_0=u$ and $a_k=v$.

\noindent{}Now define
\begin{align*}
  \critstruts(\alpha) &=
  \set*{
    \critstruts(B_i; a_{i-1},a_i)\suchthat
    \text{
      $G\cup(a_{i-1},a_i)$ is nonplanar
    }
  }
  \\
  &\quad\cup
  \set*{
    (a_{\ell-1},a_h)\suchthat
    \parbox[c]{8cm}{
      $(a_{\ell-1},a_h)\not\in G$ and\\
      $B_\ell,\ldots,B_h$ is a maximal subpath of $\alpha$ such that
      $G\cup(a_{\ell-1},a_h)$ is planar
    }
  }
  \\
  \offstruts(\alpha) &= \bigcup_{i=1}^k\offstruts(B_i; a_{i-1},a_i)
  \\
  \solidstruts(\alpha) &= \critstruts(\alpha) \cup \offstruts(\alpha)
\end{align*}
\paragraph{Combining the struts from a forest of BC trees}
Now define:
\begin{align*}
  \critstruts(G; u,v) &= \bigcup_{\substack{\alpha\in\text{BC}(G; u,v)\\\alpha\text{ is critical}}}\critstruts(\alpha)
  \\
  \offstruts(G; u,v) &=
  \left(
  \bigcup_{
    \substack{\alpha\in\text{BC}(G; u,v)\\\alpha\text{ is critical}}
  }\offstruts(\alpha)
  \right)
  \cup
  \left(
  \bigcup_{
    \substack{\alpha\in\text{BC}(G; u,v)\\\alpha\text{ is not critical}}
  }\solidstruts(\alpha)
  \right)
  \\
  \solidstruts(G; u,v) &= \critstruts(G; u,v) \cup \offstruts(G; u,v)
\end{align*}

\subsection{Proving the required properties}\label{sec:struts-properties}

\begin{lemma}
  The definition of $\critstruts(G;u,v)$ and $\solidstruts(G;u,v)$ in
  Section~\ref{sec:struts-biconn} and~\ref{sec:struts-general} have property~\ref{it:struts-planar}--\ref{it:struts-nonadmissible}.
\end{lemma}
\begin{proof}\sloppy
  Every clean separation flip in $\Emb(G)$ corresponds to an edge or a
  P node in a SPQR tree for a biconnected component of $G$. The way
  the struts are chosen for biconnected graphs, means that no two
  struts ``cover'' the same edge or P node. Similarly, each possible
  articulation flip in $\Emb(G)$ correspond to a $C$ node, and no two
  struts cover the same C node. Thus any flip in $H\in\Emb(G)$ can
  affect $\flipdist_\tau(H,\Emb(G;x,y))$ for at most $1$ strut $(x,y)$
  proving Property~\ref{it:struts-independent}.

  Since each strut $(x,y)$ is chosen so $G\cup(x,y)$ is simple and
  planar, and inserting one strut can not prevent the insertion of
  another, $G\cup\solidstruts(G;u,v)$ is planar and
  Property~\ref{it:struts-planar} holds.

  In~\cite{HR20}, we showed that it takes only $\OO(\log n)$ simple
  operations (merging or splitting S,P and C nodes, and changing edges
  between solid and dashed) to get from the pre-split BC/SPQR trees
  for $G$ with respect to $u,v$ to the trees with respect to
  $u',v'$. Each of these operations affect only at most $2$ solid
  paths, and by our definition, at most one strut on each of these
  paths. The total change in $\solidcost_\tau$ for each of these
  $\OO(\log n)$ operations is at most a constant, so
  $\abs{\solidcost_\tau(H;u,v)-\solidcost_\tau(H;u',v')}\in\OO(\log
  n)$\todo{Maybe a separate Lemma?}. If $H$ admits
  $\solidstruts(G;u,v)$ then $\solidcost_{\text{clean}}(H;u,v)=0$, so
  $\solidcost_{\text{clean}}(H;u',v')\in\OO(\log n)$. For each strut
  $(x,y)\in\solidstruts(G;u',v')$ with
  $\flipdist_{\text{clean}}(H,\Emb(x,y))>0$ there exists some flip
  that will reduce this distance. After $\OO(\log n)$ such flips we
  arrive at a new embedding $H'\in\Emb(G)$ with
  $\solidcost_{\text{clean}}(H';u',v')=0$, which means $H'$ admits
  $\solidstruts(G;u',v')$. By construction
  $\flipdist_{\text{clean}}(H,H')\in\OO(\log n)$, so this proves
  Property~\ref{it:struts-logn}.

  Property~\ref{it:struts-subset} follows trivially from the
  definition of $\solidstruts(G;u,v)$ as the (disjoint) union of
  $\critstruts(G;u,v)$ and $\offstruts(G;u,v)$.

  Property~\ref{it:struts-existing} follows by noting that in the
  definition, every edge added as a strut is conditioned on the edge
  not being in the graph.

  Property~\ref{it:struts-admissible} can be seen by considering the
  definition of $\critstruts(\alpha)$. If $(u,v)\not\in G$ and
  $G\cup(u,v)$ is planar, then for the critical path
  $\alpha\in\text{BC}(G;u,v)$, the maximal subpath found in the
  definition is exactly the one that starts in $u$ and ends in $v$. So
  in this case, $\critstruts=\set{(u,v)}$.  If $G\cup(u,v)$ is not
  planar $(u,v)\not\in\solidstruts(G;u,v)$ by
  Property~\ref{it:struts-planar}, and hence
  $(u,v)\not\in\critstruts(G;u,v)$ by Property~\ref{it:struts-subset},
  and in particular $\critstruts(G;u,v)\neq\set{(u,v)}$.

  If $G\cup(u,v)$ is simple and planar, our definition guarantees that
  $(u,v)\in\solidstruts(G;u,v)$. By definition of the pre-split
  BC/SPQR trees, the only change to the solid paths when inserting
  $(u,v)$ is the contraction of the critical path in the BC tree, and
  in all the SPQR trees on the critical path.  All other solid paths
  in the pre-split BC/SPQR tree remain unchanged.  In particular
  $\offstruts(G;u,v)=\offstruts(G\cup(u,v);u,v)$, and together with
  Property~\ref{it:struts-existing} and~\ref{it:struts-admissible} and
  the fact that $\solidstruts(G;u,v)$ is the disjoint union of
  $\critstruts(G;u,v)$ and $\offstruts(G;u,v)$ this proves
  Property~\ref{it:struts-insert}.

  Property~\ref{it:struts-nonadmissible} follows by simple
  inspection. Case~\ref{it:struts-nonadmissible-biconn} follows from
  using maximal subpaths in the definition in
  Section~\ref{sec:struts-biconn}. Similarly,
  Case~\ref{it:struts-nonadmissible-general} follows from using
  maximal subpaths in the definition in Section~\ref{sec:struts-general}.
\end{proof}

\begin{proof}[Proof of Lemma~\ref{lem:dirty-ok}]
  If the clean separation flip at $s,t$ changes $\solidcost_\tau$ or
  $\critcost_\tau$, it must be internal to a unique solid path $\beta$
  in some SPQR tree. Since it is internal, even if $\beta$ is the
  critical path for some $x,y$, then
  $\set{x,y}\cap\set{s,t}=\emptyset$. Thus, $s,t$ are not internal to
  any solid path in the BC tree, and do not have any contribution to
  $\solidcost_\tau$ or $\critcost_\tau$ that can change.

  If the clean separation flip does not change $\solidcost_\tau$ or
  $\critcost_\tau$, it may be incident to the end of a critical path
  $m(x,y)$.  However, at most $2$ of the at most $4$ articulation
  flips contain neighboring blocks on the solid path in the BC tree.
  If $2$ of them do, then the flip does not change whether or not they
  share a face, and so $\solidcost_\tau$ and $\critcost_\tau$ are
  unchanged.  Otherwise at most $1$ of the at most $4$ articulation
  flips change $\solidcost_\tau$ or $\critcost_\tau$, as desired.
\end{proof}

\section{Conclusion}

We have given an amortized $\OO(\log ^3 n)$ time algorithm for updating whether a graph is still planar after the insertion or deletion of an edge. This is close but not equal to the theoretical lower bound of $\Omega(\log n)$~\cite{DBLP:conf/stoc/PatrascuD04}. An interesting open question is whether this time bound can be improved, or whether an algorithm with worst-case polylogarithmic update time exists. 

\paragraph{Acknowledgements.} The authors would like to thank Mikkel Thorup, Kristian de Lichtenberg, and Christian Wulff-Nilsen for their interest and encouragement.

\end{document}